\documentclass[final,journal]{IEEEtran}

\usepackage{amsmath,amssymb, amsthm,cite,graphicx}
\usepackage{mathrsfs}
\usepackage{algorithmic}
\usepackage{algorithm}

\newtheorem{theorem}{Theorem}
\newtheorem{lemma}[theorem]{Lemma}
\newtheorem{proposition}[theorem]{Proposition}
\newtheorem{corollary}[theorem]{Corollary}
\DeclareMathOperator{\maxflow}{\mathrm{maxflow}}

\DeclareMathOperator{\lb}{\mathrm{lb}}
\DeclareMathOperator{\ub}{\mathrm{ub}}

\newcommand{\In}{\mathsf{In}}
\newcommand{\Mid}{\mathsf{Mid}}
\newcommand{\Out}{\mathsf{Out}}
\newcommand{\DC}{\textsf{DC}}
\newcommand{\Source}{\textsf{S}}
\newcommand{\talpha}{\tilde{\alpha}}
\newcommand{\tbeta}{\tilde{\beta}}
\newcommand{\tgamma}{\tilde{\gamma}}

\title{Cooperative Regenerating Codes}

\author{Kenneth W. Shum, \IEEEmembership{Member, IEEE,} and Yuchong Hu
\thanks{The material in this paper was presented in part at the IEEE Int. Conf. on Communications, Kyoto, June, 2011, in part at the Int. Symp. on Network Coding, Beijing, July, 2011, and in part at the IEEE Int. Symp. on Information Theory, St. Petersburg, August, 2011.
}
\thanks{
This work was done while Y. Hu was with Institute of Network Coding, the Chinese University of Hong Kong.
}
\thanks{K.~W.~Shum is with Institute of Network Coding, the Chinese University of Hong Kong.}
\thanks{Emails: wkshum@inc.cuhk.edu.hk., yuchunghu@gmail.com}
\thanks{The work described in this paper was substantially supported by a grant from University Grants Committee of the Hong Kong Special Administrative Region, China (Project No. AoE/E-02/08).}
}

\begin{document}

\maketitle

\begin{abstract}
One of the design objectives in distributed storage system is the
minimization of the data traffic during the repair of failed storage nodes. By repairing multiple failures simultaneously and cooperatively rather than successively and independently, further reduction of repair traffic is made possible. A closed-form expression of the optimal tradeoff between the repair traffic and the amount of storage in each node for cooperative repair is given. We show that the points on the tradeoff curve can be achieved by linear cooperative regenerating codes, with an explicit bound on the required finite field size. The proof relies on a max-flow-min-cut-type theorem from combinatorial optimization for submodular flows. Two families of explicit constructions are given.
\end{abstract}


\section{Introduction}

In order to provide high data reliability, distributed storage systems disperse data to a number of storage nodes. Redundancy is introduced in order to protect against node failures.
There are two common methods in introducing redundancy, namely {\em replication coding} and {\em erasure coding}.
In the former method, a data file is replicated several times, and the resulting pieces of data are stored in different storage nodes. A coding scheme in which a data file is replicated three times is employed by the Google file system~\cite{GFS}. Although replication coding is easy to implement and manage, it has lower storage efficiency than erasure codes, such as Reed-Solomon (RS) codes. In order to achieves higher storage efficiency, RS code is recently adopted in several cloud storage systems, including Oceanstore\cite{Oceanstore} and Windows Azure~\cite{Azure}, etc.

In a large-scale storage system, failure of storage nodes is a frequent event. The deployment of erasure codes incurs a significant overhead of network traffic during the repair process, because we need to download the whole data file from other surviving nodes in order to recover the lost data. The required traffic for repairing a failed node, called {\em repair bandwidth per node}, is of particular importance in bandwidth-limited storage networks. {\em Regenerating codes} was introduced by Dimakis {\em et al.} for the purpose of reducing the repair bandwidth~\cite{DGWR2010}.

There are two modes of repair in regenerating codes. In the first one, called {\em exact repair}, the content of the new node is exactly the same as the content of the failed nodes. Most of the explicit constructions of regenerating codes are for exact repair~\cite{SRKR12, RSK10h, SK11, CHL11, PD11, Thangaraj}. In some works in the literature, such as fractional repetition codes~\cite{ElRouayheb10}, self-repairing codes~\cite{OggierDatta11}, simple regenerating code~\cite{simpleRC} and locally repairable codes~\cite{LRC12,Prakash12,Silberstein12}, a failed node is repaired by downloading data from some specific subsets of surviving nodes. In this paper, however, we focus on the model as in~\cite{DGWR2010},  and assume that the new node can contact and download data from any subset of $d$ surviving nodes during the repair process, where $d$ is a constant called the {\em repair degree}.

The second mode of repair is called {\em functional repair}. With functional repair, the content of the new node are not necessarily identical to the failed nodes, but the property that a data collector connecting to any $k$ nodes is able to decode the data file is preserved. By showing that the minimum repair bandwidth can be calculated by solving a single-source multi-casting problem in network coding theory~\cite{LYC03}, the optimal tradeoff for functional repair between repair bandwidth and the storage in each node is derived in~\cite{DGWR2010}.

Most of the studies on regenerating codes in the literature focus on single-failure recovery. In large-scale distributed storage systems, however, multiple-failure recovery is the norm rather than the exception. Suppose we repair a large distributed storage system periodically, say once every two days. If the number of storage nodes is very large, very likely, we have two or more node failures in a period of time. Multiple failures occur naturally in this scenario. On the other hand,
in some practical systems such as TotalRecall~\cite{Totalrecall}, a recovery process is triggered only after the number of failed nodes has reached a predefined threshold. In this case, even though node failures are detected one by one, the lazy repair policy treats them as a multiple failures. Lastly, in peer-to-peer storage systems with high churn rate, nodes may join and leave the system in batch. This can also be regarded as multiple node failures.

In view of the motivations in the foregoing paragraph,
we address the problem of repairing multiple node failures simultaneously and jointly, by exploiting the opportunity of data exchange among the new nodes. This mode of repair, called {\em cooperative repair}, was first introduced by Hu {\em et al.} in~\cite{HXWZL10}. The new nodes first download some data from the surviving nodes, and then exchange some data among themselves. It is shown in~\cite{HXWZL10} that cooperative repair
is able to further reduce the repair bandwidth, and a coding scheme is given in~\cite{WXHO10}. However, in~\cite{HXWZL10, WXHO10}, only the special case of minimum storage per node is considered. Cooperative repair in a more general setting was investigated by Le Scouarnec {\em et al.}, who derived  in~\cite{LeScouarnec, KSS} the optimal repair bandwidth in two extreme cases, namely, the minimum-repair and minimum-bandwidth cooperative repair.

We will call a regenerating code with the functionality of cooperative repair a {\em cooperative regenerating code}. In this paper, we derive the fundamental tradeoff between the storage per node and the repair bandwidth per node, and give closed-form expressions for the points on the tradeoff curve. The derivation is based on the information flow graph for cooperative repair. As there are potentially unlimited number of data collectors, the information flow graph could be an infinite graph. The unboundedness of the information flow graph incurs technical difficulty in achieving the tradeoff curve by linear network codes. Existing algorithms for network code construction, such as the Jaggi-Sander {\em et al.}'s algorithm~\cite{JaggiSander}, assume that the graph is finite, and requires that the finite field size grows as the number of sink nodes increases. We therefore cannot apply the Jaggi-Sander {\em et al.}'s algorithm directly, unless we truncate the infinite information flow graph to a finite subgraph. If random network coding is employed, the required field size also grows as the number of destination nodes increases~\cite{HMKKESL,Bally}. The techniques in~\cite{HMKKESL,Bally} do not go through if there are infinitely many data collectors. It is therefore not straightforward to see whether we can support arbitrarily large number of repairs without re-starting the system.
Nevertheless, in the single-loss case, Wu in~\cite{Wu10} succeeded in showing, by exploiting the structure of the information flow graph, that we can work over a fixed finite field and sustain the distributed storage system {\em ad infinitum}. In this paper, we generalize the results in~\cite{Wu10} to cooperative repair.

\begin{figure}
\centering
\includegraphics[width=2.8in]{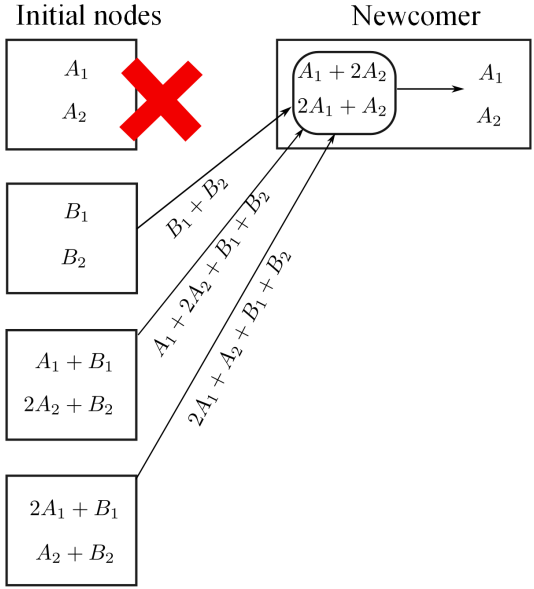}
\caption{Repairing a single node failure.} \label{fig:IA}
\end{figure}

\subsection{An Example of Cooperative Repair}
\label{section:example}
We examine the following example taken from~\cite{WD09} (Fig.~\ref{fig:IA}).
Four native data packets $A_1$, $A_2$, $B_1$ and $B_2$ are distributed to four storage nodes. Each storage node stores two packets. The first one stores $A_1$ and $A_2$, the second stores $B_1$ and $B_2$. The third node contains two parity-check packets $A_1+B_1$ and $2A_2+B_2$, and the last node contains $2A_1+B_1$ and $A_2+B_2$. Here, we interpret a packet as an element in a finite field, and carry out the additions and multiplications as finite field operations. We can take $\mathbb{F}_5$, the finite field of five elements, as the underlying finite field in this example. It can be readily checked that any data collector connecting to any two storage nodes can decode the four original packets.

Suppose that the first node fails. We want to replace it by a new node, called the {\em newcomer}. The naive method to repair the first node is to first reconstruct the four packets by connecting to any other two nodes, from which we can recover the two required packets $A_1$ and~$A_2$. Four packet transmissions are required in the naive method. The repair bandwidth can be reduced from four packets to three by making three connections. Each of the three remaining nodes adds the stored packets and sends the sum of packets to the newcomer, who can then subtract off $B_1+B_2$ and obtain $A_1+2A_2$ and $2A_1+A_2$. The packets $A_1$ and $A_2$ can now be solved readily. Hence, the lost information can be regenerated exactly by sending three packets to the newcomer.

If two storage nodes fail simultaneously, four packet transmissions per newcomer are required if we generate the content in the two new nodes separately (see Fig.~\ref{fig:SR}).
Each of the newcomers has to download four packets from the two surviving nodes. For example, in order to recover packet $B_1$, the first newcomer has to download packets $A_1$ and $A_1+B_1$. For packet $B_2$, packets $A_2$ and $2A_2+B_2$ have to be downloaded. The two new nodes essentially rebuild the whole data file $A_1$, $A_2$, $B_1$ and $B_2$, and re-encode the desired packets. The total repair bandwidth is eight. If exchange of data among the two newcomers is enabled, the total repair bandwidth can be reduced from eight packets to six packets (see Fig.~\ref{fig:CR}). The first newcomer gets $A_1$ and $A_1+B_1$, while the second newcomer gets $A_2$ and $2A_2+B_2$.  The first newcomer then figures out $B_1$ and $2A_1+B_1$ by taking the difference and the sum of the two inputs. The packet $B_1$ is stored and $2A_1+B_1$ is sent to the second newcomer. Likewise, the second newcomer computes $B_2$ and $A_2+B_2$, stores $A_2+B_2$ and sends $B_2$ to the first newcomer. The content of the failed nodes are regenerated after six packet transmissions. This example illustrates the potential benefit of cooperative repair.

\begin{figure}
\centering
\includegraphics[width=2.8in]{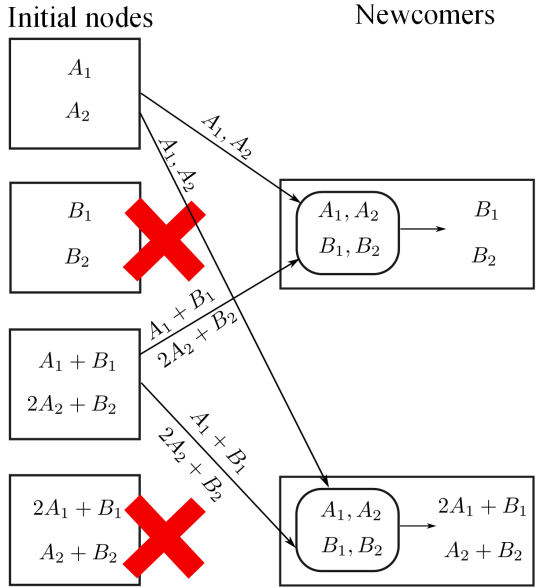}
\caption{Individual repair of multiple failures.} \label{fig:SR}
\end{figure}

\begin{figure}
\centering
\includegraphics[width=2.8in]{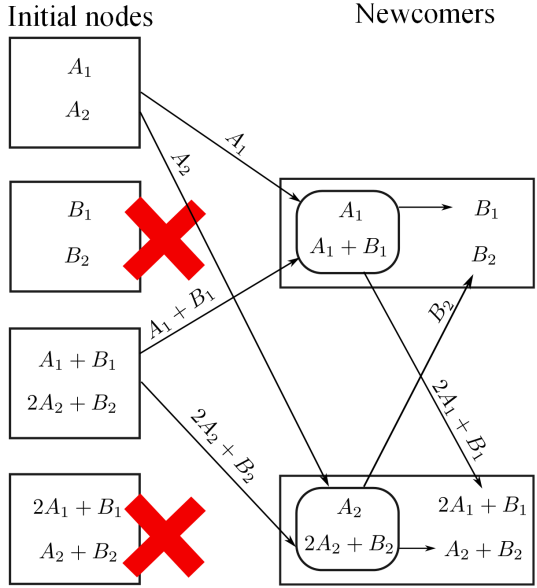}
\caption{Cooperative regeneration of multiple failures.} \label{fig:CR}
\end{figure}

\subsection{Formal Definition of Cooperative Repair}

Let $\mathcal{Q}$ be an alphabet set of size~$q$. We will call an element in $\mathcal{Q}$ a {\em symbol}. The data is regarded as a $B$-tuple $\mathbf{m} \in \mathcal{Q}^B$, with each component drawn from $\mathcal{Q}$. The distributed storage system consists of $n$ nodes, with each node storing $\alpha$ symbols. We index the storage nodes from 1 to $n$.

Time is divided into stages, and we index the stages by non-negative integers. Upon the failures of some storage nodes, we repair the failed nodes and advance to the next stage; the repair process is carried out in the transition from one stage to the next stage. For $t\geq 0$,  let the content of the $i$-th node at the $t$-th stage be denoted by an $\alpha$-tuple $\mathbf{x}(t,i) \in \mathcal{Q}^\alpha$. The distributed storage system is initialized at stage 0 by setting $\mathbf{x}(0,i) = e_i(\mathbf{m})$ for $i=1,2,\ldots, n$, where $e_i: \mathcal{Q}^{B} \rightarrow \mathcal{Q}^\alpha $ is an encoding function.

For a subset $\mathcal{S}$ of $\{1,2,\ldots, n\}$, we let $$\mathbf{x}(t,\mathcal{S}) := (\mathbf{x}(t,i))_{i\in\mathcal{S}}$$
be the content of the storage nodes indexed by $\mathcal{S}$ at the $t$-th stage. The design objective is two-fold.

 {\em (1) File retrieval.}  At each stage, a data collector can reconstruct the data file, $\mathbf{m}$, by connecting to any $k$ out of the $n$ storage nodes. We will call this property the $(n,k)$ {\em recovery property}. Mathematically, this means that for any $k$-subset $\mathcal{S}$ of $\{1,2,\ldots, n\}$ and $t \geq 0$, there is a decoding function
    $$f_{t,\mathcal{S}}: \mathcal{Q}^{k\alpha} \rightarrow \mathcal{Q}^B$$
    such that $f_{t,\mathcal{S}}( \mathbf{x}(t,\mathcal{S}) ) = \mathbf{m}$.

{\em (2) Multi-node recovery.}  When the number of node failures at stage $s-1$ reaches a threshold, say $r$, we replace the failed nodes by $r$ newcomers, and advance to stage~$s$.
  For $s = 1,2,3,\ldots$, let $\mathcal{R}_{s}$ be the set of $r$ storage nodes which fail at stage $s-1$ and are repaired in the transition from stage $s-1$ to stage~$s$. The set $\mathcal{R}_{s}$ contains $r$ elements in  $\{1,2,\ldots, n\}$. For each storage node $i \in \mathcal{R}_s$, let $\mathcal{H}_{s,i}$ be the set of storage nodes at stage $s-1$, called the {\em helpers}, from which data is downloaded to node $i$ during the repair process. We assume that the the repair degree is a constant $d$, regardless of the stage number $s$ and the index of the failed node~$i$, and different newcomers may connect to different sets of $d$ helpers. In other words, the set $\mathcal{H}_{s,i}$ can be any subset of $\{1,2,\ldots, n\} \setminus \mathcal{R}_s$ with cardinality~$d$.

\smallskip
    The repair procedure is divided into three phases.

\smallskip

      In the first phase,  each of the $r$ newcomers downloads $\beta_1$ symbols from the $d$ helpers. For $i\in\mathcal{R}_s$ and $j\in\mathcal{H}_{s,i}$, the
        symbols sent from node $j$ to newcomer~$i$ is denoted by  $g_{s,j,i}(\mathbf{x}(s-1,j))$, where
        $$ g_{s,j,i}: \mathcal{Q}^\alpha \rightarrow \mathcal{Q}^{\beta_1} $$
        is an encoding function.

\smallskip

In the second phase, the $r$ newcomers exchange data among themselves. Every newcomer sends $\beta_2$ symbols to each of the other $r-1$ newcomers. For $i_1, i_2 \in \mathcal{R}_s$ ($i_1\neq i_2$), let
         $$ g_{s,i_1,i_2}': \mathcal{Q}^{d\beta_1} \rightarrow \mathcal{Q}^{\beta_2}$$
be the encoding functions in the second phase, and
         $$\mathbf{y}(s,i_1,i_2) = g_{s,i_1,i_2}'(\{g_{s,j,i_1}(\mathbf{x}(s-1,j)):\, j\in\mathcal{H}_{s,i_1} \}) $$
be the symbols sent from newcomer $i_1$ to newcomer~$i_2$.

\smallskip

In the third phase, for each $i\in\mathcal{R}_s$, the content of the new node $i$, $\mathbf{x}(s,i)$, is obtained by applying a mapping
         $$h_{s,i}: \mathcal{Q}^{d\beta_1+(r-1)\beta_2}\rightarrow \mathcal{Q}^{\alpha}$$
to $g_{s,j,i}(\mathbf{x}(s-1,j))$ for $j\in \mathcal{H}_{s,i}$ and $\mathbf{y}(s,i',i)$ for $i' \in \mathcal{R}_{s}\setminus\{i\}$.

For those storage nodes that do not fail at stage $s-1$, the content of them do not change, i.e., $\mathbf{x}(s,i) = \mathbf{x}(s-1,i)$ for $i \not\in \mathcal{R}_s$.

A {\em cooperative regenerating code}, or a {\em cooperative regeneration scheme}, is a collection of encoding functions $e_i$, $f_{t,\mathcal{S}}$, $g_{s,j,i}$, $g_{s,i_1,i_2}'$ and $h_{s,i}$, such that the $(n,k)$ recovery property holds at all stages $t\geq 0$, for all possible failure patterns $\mathcal{R}_s$ and  all choices of helper sets $\mathcal{H}_{s,i}$,  $s\geq 1$.

A few more definitions and remarks are in order.
\begin{itemize}
\item The multi-node recovery process makes sense only when the total number of storage nodes, $n$, is larger than equal or to the sum of the number of nodes repaired jointly, $r$, and the repair degree, $d$. Henceforth we will assume that $n \geq d+r$. The results in this paper hold for all $n \geq d+r$.

\item If each storage node contains $B/k$ symbols, then the regenerating code is said to have the  {\em maximal-distance separable} (MDS) property.

\item If $\mathbf{x}(t,i) = \mathbf{x}(0,i)$ for all $t\geq 0$ and $i=1,2,\ldots, n$, then the regenerating code is said to be {\em exact}.

\item The {\em repair bandwidth} per newcomer is denoted by
$$\gamma := d\beta_1 + (r-1) \beta_2.$$

\item The encoding functions $g_{s,j,i}$, $g_{s,i_1,i_2}'$, and $h_{s,i}$ depend on the indices of the failed nodes, $\mathcal{R}_s$, the indices of the helper nodes, $\mathcal{H}_{s,i}$, and possibly  $\mathcal{R}_t$ and $\mathcal{H}_{t,i}$ for $t\leq s$, i.e., the cooperative regeneration scheme is causal. For the ease of notation, this dependency is suppressed in the notations.

\item The encoding and decoding are performed over a fixed alphabet set $\mathcal{Q}$ at all stages.

\item In practice, the file size is typically very large and can be regarded as infinitely divisible. It will be convenient to choose a unit of data such that the file size $B$ is normalized to~1, and hence the file size $B$ does not matter in the analysis.  After normalization, a pair $(\gamma/B, \alpha/B)$ is called an {\em operating point}. The first (resp. second) coordinate is the ratio of the repair bandwidth $\gamma$ (resp. storage per node $\alpha$) to the file size~$B$. We use the tilde notation $\tilde{\alpha}= \alpha/B$, $\tilde{\beta}_1 = \beta_1/B$, $\tilde{\beta}_2 = \beta_2/B$, and $\tilde{\gamma} = \gamma/B$ for variables after normalization.  All variables with tilde are between 0 and~1.

\item    An operating point $(\tilde{\gamma}, \tilde{\alpha})$ is said to be {\em admissible} if there is a cooperative regeneration scheme over an alphabet set $\mathcal{Q}$ with parameters $B$, $\alpha$, $\beta$ and $\gamma$, such that $(\tilde{\gamma}, \tilde{\alpha}) = (\gamma/B, \alpha/B)$.
 For given $d$, $k$ and $r$, let $\mathcal{C}_{\mathrm{AD}}(d,k,r)$ be the closure of all admissible operating points achieved by cooperative regenerating codes with parameters $d$, $k$ and $r$. We call $\mathcal{C}_{\mathrm{AD}}(d,k,r)$ the {\em admissible region}.  If the parameters $d$, $k$ and $r$ are clear from the context, we will simply write $\mathcal{C}_{\mathrm{AD}}$.  We let
\begin{equation}
\gamma^*(\tilde{\alpha}) := \min \{x:\, (x,\tilde{\alpha})\in\mathcal{C}_{\mathrm{AD}}(d,k,r)  \}.
\end{equation}
The value of $\gamma^*(\tilde{\alpha})$ is the optimal repair bandwidth when the amount of data stored in a node is $\tilde{\alpha}$.

\item  In the single-loss failure model ($r=1$), it is shown in~\cite{DGWR2010} that we only need to consider $d\geq k$ without  loss of generality. In multiple-loss failure model ($r > 1$), there is no a-priori reason why $d$ cannot be strictly less than $k$. However, the mathematics for the case $d\geq k$ is simpler and more tractable. In this paper, we will assume that $d$ is larger than or equal to~$k$. We will also assume that $k\geq 2$, because regenerating code with $k=1$ is trivial.
\end{itemize}

We summarize the notations  as follows:
\smallskip

\begin{tabular}{rl}
$B$ : & The size of the source file. \\
$n$ : & The total number of storage nodes. \\
$d$ : & Each newcomer connects to $d$ surviving nodes.\\
$k$ : & Each data collector connects to $k$ storage nodes. \\
$r$ : & The number of nodes repaired simultaneously. \\
$\alpha$ : &  Storage per node. \\
$\beta_1$ : &  Repair bandwidth per newcomer in the 1st phase.  \\
$\beta_2$ : &  Repair bandwidth per newcomer in the 2nd phase.  \\
$\gamma$ : &Total repair bandwidth per newcomer.
\end{tabular}

\subsection{Main Results}

The main result of this paper gives a closed-form expression for the region $\mathcal{C}_{\mathrm{AD}}(d,k,r)$. The statement of the main theorem (Theorem~\ref{thm:C}) requires the following notations.

\noindent {\bf Definitions:}
For $j=1,2,\ldots, k$, define
\begin{align}
\tilde{\alpha}_{j}&:= \frac{d-k+j+\frac{r-1}{2}}{D_j}, \label{eq:tilde_alpha1}\\
\tilde{\gamma}_{j}&:= \frac{d+\frac{r-1}{2}}{D_j} \label{eq:tilde_gamma1},
\end{align}
where $D_j$ is a short-hand notation for
\begin{equation}
D_j := k\big(d-k+j+\frac{r-1}{2}\big)-\frac{j(j-1)}{2}.
\label{eq:Dj}
\end{equation}
The points $(\tilde{\gamma}_j, \tilde{\alpha}_j)$  are called {\em operating points of the first type}.

For $\ell = 0,1,\ldots, \lfloor k/r \rfloor$, define
\begin{align}
\tilde{\alpha}_{\ell}' &:= \frac{d-k+r(\ell+1)}{D_\ell'}, \label{eq:tilde_alpha2}\\
\tilde{\gamma}_{\ell}' &:= \frac{d+r-1}{D_\ell'}, \label{eq:tilde_gamma2}
\end{align}
where
\begin{equation}
D_\ell' := k(d+r(\ell+1)-k) - \frac{r^2 \ell (\ell+1)}{2}. \label{eq:Phi}
\end{equation}
The points $(\tilde{\gamma}_\ell', \tilde{\alpha}_\ell')$ are called {\em operating points of the second type}.

\smallskip

For non-negative integer $j$ and positive integer $r$, let
\begin{equation}\Psi_{j,m} := \lfloor j/m \rfloor m^2 + (j-\lfloor j/m \rfloor m)^2.
 \label{eq:Delta}
\end{equation}

Let $\mu:\{0,1,\ldots,k\}\rightarrow \mathbb{R} \cup \{\infty\}$, be a function defined by $\mu(0) := 0$, and
\[
\mu(j) := \begin{cases} \frac{j(d-k)+(j^2+\Psi_{j,r})/2 }{jr-\Psi_{j,r}} & \text{if } \Psi_{j,r} < jr, \\
\infty & \text{if } \Psi_{j,r} = jr.
\end{cases}
\]
for $j=1,2,\ldots, k$. The motivation for  the definition of $\mu(j)$ will be given in Section~\ref{sec:OP}.

\begin{theorem}
The admissible region $\mathcal{C}_{\mathrm{AD}}(d,k,r)$ is equal to the convex hull of the union of
\begin{align}
\Big\{ ( \tilde{\gamma}_j, \tilde{\alpha}_j) &:\, j=2,3,\ldots, k-1,\ d \leq (r-1) \mu(j)  \Big\}, \label{eq:extreme_point1} \\
\Big\{ ( \tilde{\gamma}_{\lfloor j/r \rfloor}', \tilde{\alpha}_{\lfloor j/r \rfloor}') &:\,  j=2,3,\ldots, k-1,\  d > (r-1) \mu(j) \Big\}, \label{eq:extreme_point2}
\end{align}
\begin{equation}
\Big\{ (\tgamma_0'+c,\talpha_0' ):\, c \geq 0 \Big\},
\label{eq:ray1}
\end{equation}
and
\begin{equation}
\Big\{(\tgamma_k ,\talpha_k +c ) :\, c\geq 0 \Big\}.
\label{eq:ray2}
\end{equation}
When $r=1$, we use the convention  $0\cdot \infty = \infty$ in \eqref{eq:extreme_point1} and~\eqref{eq:extreme_point2}.

Furthermore, linear regenerating codes meeting this bound  exist for all $n\geq d+r$, provided that we work over a sufficiently large finite field.
\label{thm:C}
\end{theorem}

We note that each of the sets in \eqref{eq:extreme_point1} and \eqref{eq:extreme_point2} contains at most $k-2$ points. The set in~\eqref{eq:ray1} is a horizontal ray, and the set in \eqref{eq:ray2} is a vertical ray. The proof of Theorem~\ref{thm:C} is given in Sections \ref{sec:lower_bound} to~\ref{sec:LN}.

{\em Remark:} The quantity $\Psi_{j,m}$ defined in \eqref{eq:Delta} can be interpreted as the maximum value of $\sum_{i=1}^j x_i^2$ subject to the constraints $\sum_{i=1}^j x_i = j$ and $0\leq x_i \leq m$ for all~$i$.
If we divide $j$ by $m$, the quotient and remainder are, respectively, $\lfloor j/m \rfloor$ and $j-\lfloor j/m\rfloor$.
We have $\Psi_{0,m} =0$ and $\Psi_{1,m} = 1$ for all $m\geq 1$. Also, for $j\geq 2$ and $m\geq 1$, we have $j<  \Psi_{j,m} \leq jm$.
Equality $\Psi_{j,m} = jm$ holds if and only if $j$ is divisible by~$m$. In particular, we have $\Psi_{j,1} = j$ for all $j\geq 1$.

\noindent {\bf Definitions:} There are two particular operating points of special interest. The first one,
\[ (\tilde{\gamma}_{\mathrm{MSCR}}, \tilde{\alpha}_{\mathrm{MSCR}}) := (\tilde{\gamma}_0', \tilde{\alpha}_0') =  \Big( \frac{d+r-1}{k(d+r-k)}, \frac{1}{k} \Big),
\]
is called the {\em minimum-storage cooperative regenerating} (MSCR) point. This point is the end point of the half-line~\eqref{eq:ray1}.
The second one,
\[ (\tilde{\gamma}_{\mathrm{MBCR}}, \tilde{\alpha}_{\mathrm{MBCR}}) :=  (\tilde{\gamma}_k, \tilde{\alpha}_k)= \frac{2d+r-1}{k(2d+r-k)} (1,1),
\]
is called the {\em minimum-bandwidth cooperative regenerating} (MBCR) point.
This point is the end point of the half-line in~\eqref{eq:ray2}.

An operating point $(\tilde{\gamma}^\flat, \tilde{\alpha}^\flat)$ is said to {\em Pareto-dominate} another point $(\tilde{\gamma}^\sharp,\tilde{\alpha}^\sharp)$ if $\tilde{\gamma}^\flat\leq \tilde{\gamma}^\sharp$ and $\tilde{\alpha}^\flat \leq \tilde{\alpha}^\sharp$.  An operating point $(\tilde{\gamma},\tilde{\alpha})$ is called {\em Pareto-optimal} if it is in $\mathcal{C}_{\mathrm{AD}}(d,k,r)$ and not Pareto-dominated by other operating points in $\mathcal{C}_{\mathrm{AD}}(d,k,r)$.
The MSCR (resp.~MBCR) point is the Pareto-optimal point with minimum $\tilde{\alpha}$ (resp.~$\tilde{\gamma}$).

When $r=1$, Theorem~\ref{thm:C} reduces to the corresponding result for single-loss recovery in~\cite{DGWR2010}. Indeed, we have $\mu(j) = \infty$ for $j=1,2,\ldots, k$ when $r=1$. Using the convention $0\cdot \infty = \infty$, the set in~\eqref{eq:extreme_point1} contains $k-2$ operating points
\begin{equation}
(\tilde{\gamma}_j, \tilde{\alpha}_j)=\frac{2}{2k(d-k+j)-j(j-1)} \big(d, d-k+j\big),
\label{eq:r_equal_1}
\end{equation}
for $j=2,3,\ldots, k-1$, while the set in~\eqref{eq:extreme_point2} is empty. For $r=1$, the extreme points of $\mathcal{C}_{\mathrm{AD}}(d,k,1)$ are the points in~\eqref{eq:r_equal_1}, and
\begin{align}
(\tilde{\gamma}_{\mathrm{MSR}}, \tilde{\alpha}_{\mathrm{MSR}}) &:= (\tilde{\gamma}_0', \tilde{\alpha}_0')=\Big( \frac{d}{k(d+1-k)}, \frac{1}{k} \Big),  \label{eq:MSR} \\
(\tilde{\gamma}_{\mathrm{MBR}}, \tilde{\alpha}_{\mathrm{MBR}})&:= (\tilde{\gamma}_k, \tilde{\alpha}_k)=\frac{2d}{k(2d+1-k)} (1 ,1). \notag
\end{align}

\smallskip

We define the  {\em storage efficiency} as the number of symbols in the data file divided by the total number of symbols in the $n$ storage nodes. In terms of the normalized storage per node, the storage efficiency is equal to $1/(n\tilde\alpha)$. The storage efficiency of an MSCR code is $k/n$.

For MBCR, the storage efficiency is $$\frac{k(2d+r-k)}{n(2d+r-1)}.$$
If we fix $n$, $d$ and $k$, and increase the value of $r$, then the storage efficiency increases. Alternately, if we fix $n$, $k$ and $r$, and increase the value of  $d$, the storage efficiency also increases. However, the storage efficiency cannot exceed $1/2$. One can see this by first upper bounding it by
$$
\frac{k(2d+r-k)}{(d+r)(2d+r-1)},
$$
and then show that
\begin{align*}
& \phantom{=} 1 - 2\frac{k(2d+r-k)}{(d+r)(2d+r-1)} \\
& =
\frac{2(d-k)^2 + (r-k)^2 + k^2 +(2r-1)d + (r-1)d}{(d+r)(2d+r-1)} > 0.
\end{align*}

In Section~\ref{sec:explicit}, two families of cooperative regenerating codes for exact repair are constructed explicitly. Both families have the property $d=k$. The first family matches the MSCR point, and has parameters $B=kr$, $n\geq d+r$, $\alpha=r$ and $\gamma = d+r-1$. The second family matches the MBCR point and has parameters $B=k(k+r)$, $n=d+r$ and $\alpha=\gamma=2d+r-1$.

\subsection{Numerical Illustrations}

We illustrate  the admissible region $\mathcal{C}_{\mathrm{AD}}(5,4,3)$ (with parameters $d=5$, $k=4$, $r=3$) in Fig.~\ref{fig:tradeoff1}. The solid line (marked by squares) is the boundary of the region $\mathcal{C}_{\mathrm{AD}}(5,4,3)$. The set in~\eqref{eq:extreme_point1} contains two points, namely
\begin{align*}
(\tilde{\gamma}_2, \tilde{\alpha}_2) &= \Big(\frac{d+\frac{r-1}{2}}{D_2}, \frac{d-k+2+\frac{r-1}{2}}{D_2} \Big) \\
& = (6/15, 4/15)\doteq (0.4, 0.2667),
\end{align*}
and
\begin{align*}
(\tilde{\gamma}_3, \tilde{\alpha}_3) &= \Big(\frac{d+\frac{r-1}{2}}{D_3}, \frac{d-k+3+\frac{r-1}{2}}{D_3} \Big) \\
& = (6/17, 5/17)\doteq(0.3529, 0.2941).
\end{align*}
The set in~\eqref{eq:extreme_point2} is empty.
The MSCR and MBCR points are, respectively,
\begin{align*}
(\tilde{\gamma}_{\mathrm{MSCR}}, \tilde{\alpha}_{\mathrm{MSCR}}) &=  \Big( \frac{d+r-1}{k(d+r-k)}, \frac{1}{k} \Big) \\
& =  (7/16,1/4)= (0.4375, 0.25),
\end{align*}
and
\begin{align*}
(\tilde{\gamma}_{\mathrm{MBCR}}, \tilde{\alpha}_{\mathrm{MBCR}}) &=  \Big( \frac{2d+r-1}{k(2d+r-k)}, \frac{2d+r-1}{k(2d+r-k)} \Big) \\
& =  (1/3, 1/3)\doteq(0.3333,0.3333).
\end{align*}
For comparison, we also plot in Fig.~\ref{fig:tradeoff1} the optimal tradeoff curve for single-failure repair with parameters $d=5$, $k=4$ and $r=1$ (marked by circles). We observe that the boundary of the admissible region is piece-wise linear.

\begin{figure}
\centering
\vspace{-5mm}
\includegraphics[width=3.8in]{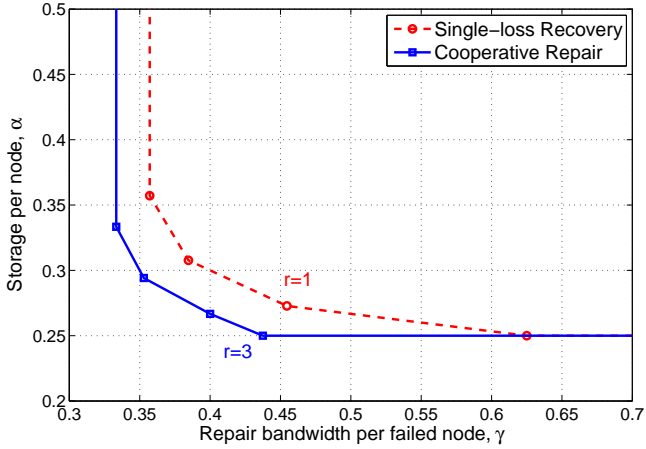}
\caption{Tradeoff between storage and repair bandwidth, $B=1$, $d=5$, $k=4$. The solid line is for $r=3$ and the dashed line is for $r=1$.}
\label{fig:tradeoff1}
\end{figure}

\smallskip

\renewcommand{\algorithmicrequire}{\textbf{Input:}}
\renewcommand{\algorithmicensure}{\textbf{Output:}}

\begin{algorithm}
\caption{Plot the tradeoff curve for cooperative repair} \label{algo_plot}
\begin{algorithmic}[1]

\REQUIRE $d$, $k$, $r$.

\ENSURE The tradeoff curve of storage per node versus repair bandwidth per node.

\STATE $\tgamma \gets \tgamma_{\mathrm{MSCR}}$, $\talpha \gets \talpha_{\mathrm{MSCR}}$.

\FOR { $j=2,3,\ldots, k$}
\IF { $r=1$ or $d \leq (r-1)\mu(j)$ }

\STATE $x \gets \tgamma_j$, $y \gets \talpha_j$.

\ELSE

\STATE $x \gets \tgamma_{\lfloor j/r \rfloor }'$, $y \gets \talpha_{\lfloor j/r \rfloor}'$.

\ENDIF

\STATE Draw a line segment  from $(\tgamma, \talpha)$ to $(x, y)$.

\STATE $\tgamma \gets x$, $\talpha \gets y$.

\ENDFOR


\STATE Draw a horizontal ray  from $(\tgamma_{\mathrm{MSCR}}, \talpha_{\mathrm{MSCR}})$ to $(\infty, \talpha_{\mathrm{MSCR}})$.

\STATE Draw  a vertical ray from $(\tgamma_{\mathrm{MBCR}},\talpha_{\mathrm{MBCR}})$ to
$(\tgamma_{\mathrm{MBCR}},\infty)$.

\end{algorithmic}
\end{algorithm}

Even though the statement in Theorem~\ref{thm:C} is a little bit complicated, we can plot the tradeoff curve by the procedure described in Algorithm~\ref{algo_plot}. As a numerical example, we plot the tradeoff curves with parameters $B=1$, $d=21$, $k=20$, and $r=1,3,5,7,9,11,13$ in Fig.~\ref{fig:tradeoff3}. The curve for $r=1$ is the tradeoff curve for single-node-repair regenerating code. The repair degree $d=21$ is kept constant, and the number of storage nodes can be any integer larger than or equal to $d+13=34$.  We can see in Fig.~\ref{fig:tradeoff3} that we have a better tradeoff curve when the number of cooperating newcomers increases. We indicate the operating points of the first type by dots and operating points of the second type by squares. We observe that all but one operating points of the second type are on the horizontal line $\alpha = 0.05$. The exceptional operation point of the second type lies on the trade-off curve with~$r=3$.

\begin{figure}
\centering
\includegraphics[width=3.8in]{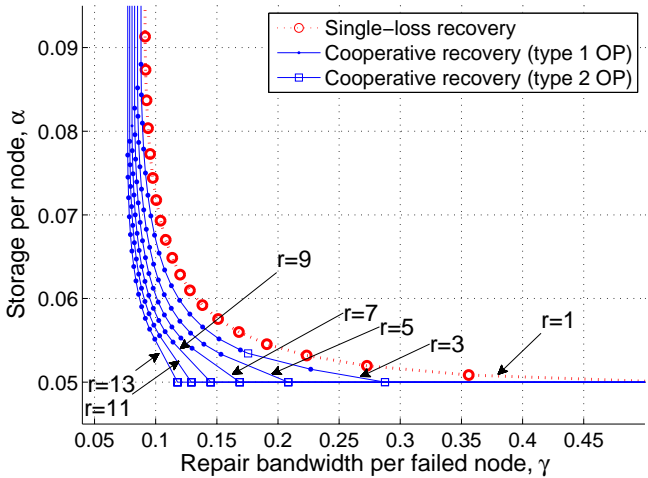}
\caption{Tradeoff between storage and repair bandwidth ($B=1$, $d=21$, $k=20$, $r=1,3,5,7,9,11,13$).}
\label{fig:tradeoff3}
\end{figure}

\smallskip

We compare below the repair bandwidth of three different
modes of repair in a distributed storage system of $n = 7$ nodes. We require that any $k=3$ nodes is sufficient in decoding the original file.
Each node contains the minimum amount of data, i.e., $\tilde{\alpha} = 1/3$.

Suppose that three nodes have failed.

(i) {\em Individual repair without newcomer cooperation}. Each
newcomer connects to the four remaining storage nodes. From~\eqref{eq:MSR}, the normalized repair bandwidth per newcomer is
\[
\tilde{\gamma}_{\mathrm{MSR}}= \frac{d}{k(d+1-k)} = \frac{4}{3(4+1-3)} \doteq 0.6667.
\]

(ii) {\em One-by-one repair}. We repair the failed nodes one by one. The newly repaired nodes are utilized as the helpers during the repair of the remaining failed nodes. The average repair bandwidth per newcomer is
\[
\frac{1}{3} \Big( \frac{4}{3(4+1-3)} + \frac{5}{3(5+1-3)} +\frac{6}{3(6+1-3)}   \Big) \doteq 0.5741.
\]
The first term in the parenthesis is the repair bandwidth of
the first newcomer, who downloads from the four surviving
nodes, the second term is the repair bandwidth of the second
newcomer, who connects to the four surviving nodes and the
first newcomer, and so on.

(iii) {\em Full cooperation among the three newcomers}. With $r=k=3$ and $d=4$, the normalized repair bandwidth per newcomer is
\[
\tilde{\gamma}_{\mathrm{MSCR}} = \frac{d+r-1}{k(d+r-k)}=\frac{4+3-1}{3(4+3-3)} = 0.5.
\]
We thus see that the full cooperation in (iii) gives the smallest repair bandwidth.

\subsection{Organization}

This paper is organized as follows.
In Section~\ref{sec:LBRB}, we review the information flow graph for cooperative repair, and state some definitions and theorems from combinatorial optimization. In Section~\ref{sec:lower_bound},  a lower bound on repair bandwidth for cooperative recovery is derived.
The lower bound is expressed in terms of a linear programming problem. In Section~\ref{sec:OP}, we solve the linear program explicitly.
In Section~\ref{sec:max_flow} we show that the lower bound is tight by using some results from the theory of submodular flow. We prove in Section~\ref{sec:LN} that we can construct linear network codes over a fixed finite field, which match this lower bound on repair bandwidth. Two explicit constructions for exact-repair cooperative regenerating codes are given in Section~\ref{sec:explicit}. Appendix~\ref{app:asym} discusses the scenario of heterogeneous download traffic. Some of the longer proofs are relegated to the remaining appendices.

\section{Preliminaries}
\label{sec:LBRB}

\subsection{Polymatroid and submodular flow}

We collect some definitions and basic facts of submodular functions and polymatroids. We refer the readers to the texts~\cite{SchrijverB, Fujishige, Frank} for more details.

\noindent {\bf Definitions:} Let $\mathbb{R}$ be the set of real numbers and $\mathbb{R}_+$ be the set of non-negative real numbers. For a finite set $\mathcal{V}$, we denote the cardinality of $\mathcal{V}$ by $|\mathcal{V}|$.
We let $\mathbb{R}^\mathcal{V}$ be the set of vectors with components indexed by the elements in $\mathcal{V}$, and  $\mathbb{R}_+^\mathcal{V}$ be subset of vectors in $\mathbb{R}^{\mathcal{V}}$ with non-negative components. In the rest of this paper, a vector will be identified with a real-valued function on $\mathcal{V}$.

Let the set of all subsets of $\mathcal{V}$ be $2^\mathcal{V}$. A set function $f:2^\mathcal{V} \rightarrow\mathbb{R}$ is called {\em submodular} if it satisfies
\begin{equation}
 f(\mathcal{S})+f(\mathcal{T}) \geq f(\mathcal{S}\cap\mathcal{T}) + f(\mathcal{S}\cup \mathcal{T}) \label{eq:submodular_def}
\end{equation}
for all  $\mathcal{S}, \mathcal{T} \subseteq\mathcal{V}$. To show that a function $f$ is submodular, it is sufficient to check that
\[
 f(\mathcal{S} \cup \{u\})+ f(\mathcal{S} \cup \{v\}) \geq f(\mathcal{S})+
  f(\mathcal{S} \cup \{u,v\})
\]
for all subsets $\mathcal{S}\subseteq \mathcal{V}$ and $u,v\in\mathcal{V}$  (See \cite[Thm 44.1]{SchrijverB}).

If \eqref{eq:submodular_def} holds with equality for all $\mathcal{S}$ and $\mathcal{T}$ in $2^\mathcal{V}$, then $f$ is called {\em modular}. For a given vector $\mathbf{x} = (x_i)_{i\in\mathcal{V}}$, we can define a modular function by
$$\mathbf{x}(\mathcal{S}):= \sum_{i\in \mathcal{S}} x_i,
$$
for all subsets $\mathcal{S} \subseteq \mathcal{V}$.

A submodular function $f$ is said to be {\em monotone} if $f(\mathcal{S})\leq f(\mathcal{T})$ whenever $\mathcal{S} \subseteq \mathcal{T}$. Furthermore, a monotone submodular function $f$ satisfying $f(\emptyset)=0$ is called a {\em polymatroidal rank function}, or simply a {\em rank function}.

The {\em polymatroid} corresponding to a rank function $f$ is the polyhedron defined as
\[
\mathcal{P}(f) := \{ \mathbf{x}\in\mathbb{R}_+^\mathcal{V} :\,  \mathbf{x}(\mathcal{S}) \leq f(\mathcal{S}), \ \forall \mathcal{S}\subseteq \mathcal{V} \}.
\]
The face of the polymatroid consisting of the points satisfying $\mathbf{x}(\mathcal{V}) = f(\mathcal{V})$ is called the {\em base-polymatroid} associated with the rank function~$f$. It is well known that the base-polymatroid is non-empty (See e.g.~\cite[Thm. 2.3]{Fujishige}). We will use the symbol $\mathcal{B}(f)$ to denote the base-polymatroid corresponding to rank function $f$,
\[
 \mathcal{B}(f) := \{ \mathbf{x} \in \mathcal{P}(f):\, \mathbf{x}(\mathcal{V}) = f(\mathcal{V})\}.
\]

%
%

For a given vector $\mathbf{x} \in \mathbb{R}_+^\mathcal{V}$, we sort the components of $\mathbf{x}$ in non-increasing order and let the $j$-th largest component in $\mathbf{x}$ be denoted by~$x_{[j]}$, i.e.,
$$ x_{[1]} \geq x_{[2]} \geq \cdots \geq x_{[|\mathcal{V}|]}.
$$
Given two vectors  $\mathbf{x}$ and  $\mathbf{y}$ in  $\mathbb{R}_+^\mathcal{V}$, we say that $\mathbf{x}$ is {\em majorized} by $\mathbf{y}$ if
\[   x_{[1]} +x_{[2]} + \cdots + x_{[i]} \leq y_{[1]} +  y_{[2]} + \cdots +   y_{[i]} ,
\]
for $i=1,2,\ldots, |\mathcal{V}|-1$ and
\[  \sum_{j=1}^{|\mathcal{V}|} x_{[j]} = \sum_{j=1}^{|\mathcal{V}|} y_{[j]}.
\]

In this paper, we will construct polymatroids and rank functions by the following lemma~\cite[p.44]{Fujishige}.

\begin{lemma} Let $\mathcal{V}$ be a finite set and $\mathbf{u}$ be a given vector in $\mathbb{R}_+^{\mathcal{V}}$. The function $f:2^\mathcal{V} \rightarrow \mathbb{R}_+$ defined by
\[
f(\mathcal{S}) := \sum_{j=1}^{|\mathcal{S}|} u_{[j]}
\]
is a rank function. The set of vectors in $\mathbb{R}_+^{\mathcal{V}}$ which are majorized by $\mathbf{u}$ is precisely the base-polymatroid associated with the rank function~$f$. \label{lemma:polymatroid}
\end{lemma}

\begin{proof} (Sketch) For the submodularity,
it is sufficient to check that the condition
\begin{equation}2f(|S|+1) \geq f(|S|) + f(|S|+2)
\label{eq:lemma_sketch}
\end{equation}
for all $\mathcal{S}\subseteq \mathcal{V}$ with $|\mathcal{S}| \leq |\mathcal{V}|-2$. The inequality in~\eqref{eq:lemma_sketch} is equivalent to
$u_{[|S|+1|]} \geq u_{[|S|+2]}$, which holds by construction. The function $f$ is monotone because the function $\sum_{j=1}^i u_{[j]}$ is monotonically nondecreasing as a function of $i$.
\end{proof}

It is obvious that any submodular function $f(\mathcal{S})$ constructed as in Lemma~\ref{lemma:polymatroid} only depends on the size of $\mathcal{S}$.
We give a numerical example for Lemma~\ref{lemma:polymatroid}. Let $f$ be the rank function
\[
 f(\mathcal{S})= \begin{cases}
 0 & \text{ if } \mathcal{S} = \emptyset \\
 2 & \text{ if } |\mathcal{S}| = 1\\
4 & \text{ if } |\mathcal{S}| = 2 \\
 5 & \text{ if } |\mathcal{S}| = 3
 \end{cases}
\]
induced from the vector $\mathbf{u} = (2,2,1)$. The base-polymatroid $\mathcal{B}(f)$ consists of the vectors $(x,y,z)$ in $\mathbb{R}_+^3$  which satisfy
\begin{gather*}
x\leq 2, \ y\leq 2,\ z\leq 2,
x+y \leq 4, \ y+z \leq 4, \ z+x \leq 4, \\
x+y+z = 5.
\end{gather*}
The vectors in $\mathcal{B}(f)$ are precisely the vectors in $\mathbb{R}_+^3$ which are majorized by~$\mathbf{u}$.

\noindent {\bf Definitions:}
Let $H = (\mathcal{V}, \mathcal{E})$ be a directed graph. For a given subset $\mathcal{T}$ of $\mathcal{V}$, define the set of incoming edges and the set of out-going edges, respectively, by
\begin{align*}
 \Delta^- \mathcal{T} &:= \{e =(u,v) \in \mathcal{E}: u\not\in \mathcal{T}, v \in \mathcal{T} \} , \\
 \Delta^+\mathcal{T} &:= \{e =(u,v) \in \mathcal{E}: u \in \mathcal{T}, v \not\in\mathcal{T} \}.
\end{align*}
When $\mathcal{T}$ is a singleton $\{v\}$, $\Delta^- \{ v\}$ is the set of edges which terminate at vertex $v$, and $\Delta^+ \{ v\}$ is the set of edges which emanate from $v$. We will write
\[
 \Delta^-v := \Delta^- \{ v\} \text{ and }  \Delta^+ v := \Delta^+ \{ v\}.
\]

Let $\phi:\mathcal{E} \rightarrow \mathbb{R}$ be a real-valued function on the edges of~$H$. We extend the function $\phi$ naturally to a set function, by defining
\[
 \phi(\mathcal{E}') := \sum_{e \in \mathcal{E}'} \phi(e)
\]
for $\mathcal{E}' \subseteq \mathcal{E}$. The boundary of $\phi$, denoted by $\partial \phi$, is the set function on $2^\mathcal{V}$ defined by
\begin{align*}
\partial \phi (\mathcal{T}) = \phi(\Delta^+ \mathcal{T}) - \phi(\Delta^- {\mathcal{T}}).
\end{align*}
The boundary of $\phi$ is a modular function, and can be interpreted as the net out-flow of the subset of vertices $\mathcal{T}$ with respect to $\phi$.
For a given a submodular function $f : 2^\mathcal{V}\rightarrow \mathbb{R}$, we say that a function $\phi : \mathcal{E} \rightarrow \mathbb{R}$ is an {\em $f$-submodular flow}, if
\begin{equation}
 \partial \phi( \mathcal{T}) \leq f(\mathcal{T})
 \label{def:submodular}
\end{equation}
for all $\mathcal{T} \subseteq \mathcal{V}$. We will simply write ``submodular flow'' instead of ``$f$-submodular flow'' if $f$ is understood from the context.

Let $\lb: \mathcal{E} \rightarrow \mathbb{R} \cup \{-\infty\}$ and
$\ub: \mathcal{E} \rightarrow \mathbb{R} \cup \{\infty\}$ be two functions defined on the edge set, called, respectively, the lower and upper bound on $\mathcal{E}$, satisfying $\lb(e) \leq \ub(e)$ for all $e\in\mathcal{E}$.
For a given subset $\mathcal{E}'$ of the edge set $\mathcal{E}$, we define
\begin{align*}
 \lb(\mathcal{E}') &:= \sum_{e \in \mathcal{E}'} \lb(e), \\
 \ub(\mathcal{E}') &:= \sum_{e \in \mathcal{E}'} \ub(e).
\end{align*}
A submodular flow $\phi$ is said to be {\em feasible} if
$\lb(e) \leq \phi(e) \leq \ub(e)$
for all $e\in \mathcal{E}$.

The following theorem characterizes the existence of a submodular flow. It is a generalization of the max-flow-min-cut theorem, and is essential in the proof of the main theorem in this paper.

\begin{theorem}[Frank \cite{Frank82}] Suppose that $f$ is a submodular function defined on the vertex set $\mathcal{V}$ of a directed graph $(\mathcal{V},\mathcal{E})$  and $\lb$ and $\ub$ be the lower bound and upper bound functions defined on the edge set $\mathcal{E}$, satisfying  $f(\emptyset)=f(\mathcal{V})=0$ and $\lb(e)\leq\ub(e)$ for all $e\in\mathcal{E}$.
There exists a feasible $f$-submodular flow if and only if
\begin{equation}
\lb(\Delta^+ \mathcal{S})  -
\ub( \Delta^- \mathcal{S})
 \leq f(\mathcal{S})
 \label{eq:cut}
\end{equation}
for all subsets $\mathcal{S} \subseteq \mathcal{V}$. Moreover, if $\lb$, $\ub$ and $f$ are integer-valued, then there is a feasible $f$-submodular flow which is integer-valued.
\label{thm:Frank}
\end{theorem}

Proof of Frank's theorem can be found in \cite[Thm 5.1]{Fujishige} or~\cite[Thm 12.1.4]{Frank}.

\subsection{Information Flow Graph and the Max-Flow Bound} \label{sec:mincut}
We review the information flow graph for cooperative repair as defined in~\cite{HXWZL10}.

The information flow graph is divided into stages, starting from stage $-1$.
Given parameters $n$, $k$, $d$ and $r$, any directed graph $G=(\mathcal{V},\mathcal{E})$ which can be constructed according to the following procedure is called an {\em information flow graph}. An example of information flow graph is shown in Fig.~\ref{fig:flow1}.

\begin{itemize}
\item There is one single source vertex $\Source$ at stage $-1$, representing the original data file.

\item  The $n$ storage nodes after initialization are represented by $n$ vertices at stage~0, called $\Out_i$, for $i=1,2,\ldots, n$. There is a directed edge from the source vertex $\Source$ to each of the ``out'' vertices at stage~0.

\item For $s\geq 1$ and for each $j$ in $\mathcal{R}_{s}$, we put three vertices at stage $s$: $\In_j$, $\Mid_j$ and $\Out_j$. For each $j \in \mathcal{R}_s$, there is a directed edge from $\In_j$ to $\Mid_j$ and a directed edge from $\Mid_j$ to $\Out_j$.
    For each $i \in \mathcal{H}_{s,j}$, we put a directed edge from $\Out_i$ at stage $s-1$ to $\In_j$ at stage $s$. The exchange of data among the $r$ newcomers are modeled by putting a directed from $\In_i$ to $\Mid_j$ for all pairs of distinct $i$ and $j$ in $\mathcal{R}_s$.

\item For each data collector who shows up at stage $s$, we put a vertex, with label \DC, to the information flow graph. This vertex is connected to $k$ ``out'' vertices at the $s$-th or earlier stages. The contacted ``out'' vertices did not fail recently up to stage~$s$.
\end{itemize}

\begin{figure}
\centering
\includegraphics[width=3in]{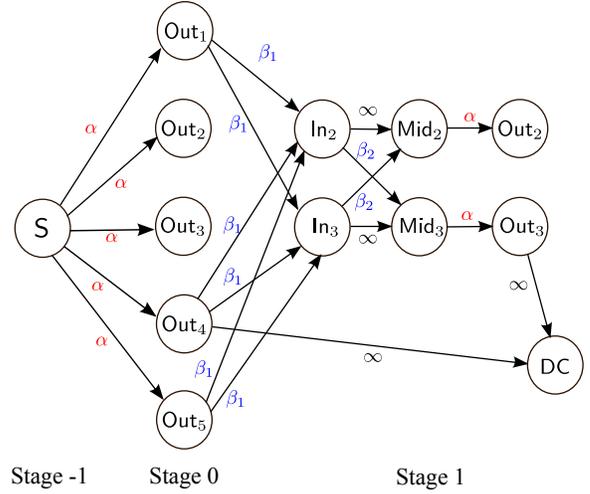}
\caption{An example of information flow graph $G(5,3,2,2;\alpha,\beta_1,\beta_2)$. Nodes 2 and 3 are repaired at stage 1 ($\mathcal{R}_1 = \{2,3\}$).} \label{fig:flow1}
\end{figure}

We assign capacities to the edges as follows.

\begin{itemize}
\item The capacity of an edge terminating at an ``out'' vertex is~$\alpha$. This models the storage requirement in each storage node.

\item The capacity of an edge from an ``in'' vertex to a ``mid'' vertex is infinity. It models the transfer of data inside the newcomer, which does not contribute to the repair bandwidth.

\item The capacity from $\Out_i$ at stage $s-1$ to $\In_j$ at stage $s$ is $\beta_1$, for $i\in\mathcal{H}_{s,j}$. This signifies the amount of data sent from $\Out_i$ to $\In_j$ in the first phase of the repair process. The edge from $\In_j$ to $\Mid_\ell$ at stage $s$, for $j , \ell \in \mathcal{R}_s$ with $j\neq \ell$, is assigned a capacity of $\beta_2$. This signifies the data exchange in the second phase.

\item The edges terminating at a data collector are all of infinite capacity.
\end{itemize}

\begin{figure*}
\centering
\includegraphics[width=4.5in]{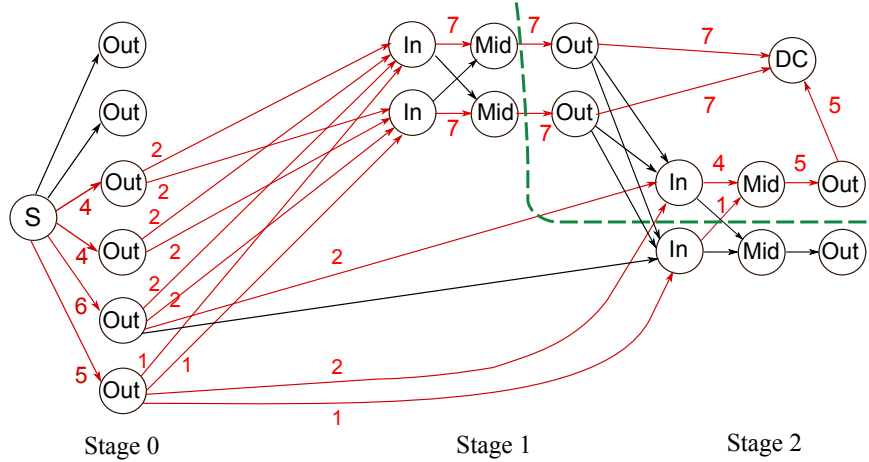}
\caption{An example of flow in an information flow graph. The
parameters are $n=6$, $d=4$, $k=3$, $r=2$, $\alpha=7$, $\beta_1=2$ and $\beta_2=1$. The labels of the edges indicate a flow on the graph (the arrows in red are assigned positive flow value and the arrows in black are assigned
zero value). A cut with capacity 19 is illustrated by a dashed line. }
\label{fig:cutG}
\end{figure*}

The information flow graph so constructed is a directed acyclic graph. It may be  an infinite graph, as the number of stages is unlimited. We will denote an  information flow graph by $G(n,d,k,r;\alpha,\beta_1, \beta_2)$. If the values of parameters are understood from the context, we will simply write $G$.

\noindent \textbf{Definitions}: Let $H = (\mathcal{V},\mathcal{E})$ be a directed graph, in which each edge $e \in \mathcal{E}$ is assigned a non-negative capacity $c(e)$. For two distinct vertices $S$ and $T$ in $\mathcal{V}$, an $(S,T)$-{\em flow} in $H$ is a function $\phi: \mathcal{E} \rightarrow \mathbb{R}_+$, such that $\phi(e) \leq c(e)$ for all $e\in\mathcal{E}$, and $\partial \phi (\{v\}) = 0$ for every vertex $v$ in $\mathcal{V} \setminus\{S,T\}$. A flow $\phi$ is called {\em integral} if $\phi(e)$ is an integer for every edge $e$. The {\em value} of an $(S,T)$-flow $\phi$ is defined as $\phi(\Delta^- T)$.  An $(S,T)$-{\em cut} is a partition $(\mathcal{W}^c, \mathcal{W})$ of the vertex set $\mathcal{V}$ of $H$ such that $S \in {\mathcal{W}^c}$ and $T \in \mathcal{W}$. (The superscript $^c$ stands for the set complement in $\mathcal{V}$.)
The {\em capacity} of an $(S,T)$-cut $(\mathcal{W}^c, \mathcal{W})$ is defined as $$c(\Delta^- \mathcal{W}) := \sum_{e\in \Delta^- \mathcal{W}} c(e),$$
 the sum of the capacities of the edges from $\mathcal{W}^c$ to $\mathcal{W}$.

The {\em max-flow-min-cut theorem} states that the minimal cut capacity and maximal flow value coincide. Furthermore, if the edge capacities are all integer-valued, then there is a maximal flow which is integral. In Appendix \ref{app:maxflowFrank}, we illustrate that the max-flow-min-cut theorem is a special case of Frank's theorem.

\noindent \textbf{Definitions}: For a given data collector \DC\ in the information flow graph $G$, we let $$\maxflow(G,\DC)$$ be the maximal flow value from the source vertex $\Source$ to~$\DC$.

Even though the graph $G$ may be infinite, the computation of the flow from the source vertex to a particular data collector \DC\ at stage $t$ only involves the subgraph of $G$ from stage $-1$ to stage $t$. For each \DC, the problem of determining the max-flow reduces to a max-flow problem in a finite graph.

An example of flow in an information flow graph for $n=6$, $d=4$, $k=3$, $r=2$ is shown in Fig.~\ref{fig:cutG}. The data collector \DC\ is connected to one ``out'' vertex at stage 2 and two ``out'' vertices at stage~1. All  edges from  ``out'' vertex to ``in'' vertex, corresponding to the first phase of the repair process, have capacity $\beta_1=2$. All edges from ``in'' vertex to ``mid'' vertex, corresponding to the second phase, have capacity $\beta_2=1$. All edges terminating at an ``out'' vertex have capacity $\alpha=7$.
The edges with positive flow are labeled (and drawn in red color). The flow value is equal to~19. This is indeed a flow with maximal value, because there is a cut with capacity~19 (shown as the dashed line in Fig.~\ref{fig:cutG}).

According to the max-flow bound of network coding~\cite{ACLY00} \cite[Theorem 18.3]{Raymond08}, if all data collectors are able to to retrieve the original file, then the file size $B$ is upper bounded by
\begin{equation}
B\leq \min_{\DC} \maxflow(G,\DC). \label{eq:maxflow1}
\end{equation}
The minimum in \eqref{eq:maxflow1} is taken over all data collector $\DC$ in graph $G$. This gives an upper bound on the supported file size for a given information flow graph $G$. Since we want to build cooperative regenerating schemes that can repair any pattern of node failures, which are unknown the system is initialized, we take the minimum
\begin{equation}
B\leq  \min_{G} \min_{\DC} \maxflow(G,\DC). \label{eq:maxflow2}
\end{equation}
over all information flow graphs $G(n, d, k , r; \alpha, \beta_1, \beta_2)$.

\noindent \textbf{Definitions}:  For given parameters $n$, $d$, $k$, $r$, we denote by
\begin{equation}
 \mathcal{C}_{\mathrm{MF}}(d,k,r)
\end{equation}
the set of operating points  $((d\beta_1+(r-1)\beta_2)/B, \alpha/B)$ which satisfy the condition in~\eqref{eq:maxflow2}. For a given $\tilde{\alpha} = \alpha/B$, let
\begin{equation}
 \gamma_{\mathrm{MF}}^*(\tilde{\alpha}) := \min \{ x:\, (x,\tilde{\alpha})  \in  \mathcal{C}_{\mathrm{MF}}(d,k,r) \}.
\end{equation}

Any operating point not in  $\mathcal{C}_{\mathrm{MF}}(d,k,r)$ violates the max-flow bound for some information flow graph, and hence is not admissible. We have the following inclusion,
\begin{equation}\mathcal{C}_{\mathrm{MF}}(d,k,r) \supseteq \mathcal{C}_{\mathrm{AD}}(d,k,r).
\label{eq:reverse}
\end{equation}
We note that for fixed $\alpha$, $\beta_1$ and $\beta_2$, if $B$ satisfies \eqref{eq:maxflow2}, then \eqref{eq:maxflow2} is satisfied for all $B'$ between 0 and $B$. Hence, if $(\tilde{\gamma}, \tilde{\alpha}) \in \mathcal{C}_{\mathrm{MF}}(d,k,r)$, then $(c\tilde{\gamma}, c\tilde{\alpha}) \in \mathcal{C}_{\mathrm{MF}}(d,k,r)$ for all $c\geq 1$.

\section{A Cut-set Bound on the Repair Bandwidth} \label{sec:lower_bound}
Consider a data collector \DC\ connected to $k$ storage nodes.  By re-labeling the storage nodes, we can assume without loss of generality that the \DC\ downloads data from nodes 1 to $k$. Suppose that among these $k$ nodes, $\ell_0$ of them do not undergo any repair, and the remaining $k-\ell_0$ nodes are repaired at stage 1 to $s$ for some positive integer $s$. For $j=1,2,\ldots, s$, suppose that there are $\ell_j$ nodes which are repaired at stage $j$ and connected to the data collector~\DC. We have $$\ell_0+\ell_1+\cdots+\ell_s=k$$ and  $$1\leq \ell_j \leq r$$
for $j\geq 1$. After some re-labeling again, we can assume that the $\ell_0$ unrepaired nodes are node 1 to node $\ell_0$, the nodes which are repaired at stage 1 are node $\ell_0+1$ to node $\ell_0+\ell_1$, and so on.

In the information flow graph, the data collector \DC\ is connected to $\ell_j$ ``out'' vertices at stage $j$. A cut $(\mathcal{W}^c, \mathcal{W})$ with $\mathcal{W}$ consisting of the data collector \DC, the $\ell_0$ ``out'' vertices at stage 0 associated with nodes 1 to $\ell_0$, and
\[
\bigcup_{i=\ell_{j-1}+1}^{\ell_j} \{\In_{i}, \Mid_{i}, \Out_{i}\}
\]
at stage $j$, for $j = 1,2,\ldots, s$, is called a cut of type
\[
 (\ell_0, \ell_1, \ell_2, \ldots, \ell_s).
\]
An example of a cut of type $(2,1,1,2)$ is shown in Fig.~\ref{fig:type}. Nodes 3 and 4 are repaired at stage 1, nodes 4 and 7 are repaired at stage 2, and nodes 5 and 6 are repaired at stage~3. The data collector connects to nodes 1 to 6. The vertices in $\mathcal{W}$ are drawn in shaded color in Fig.~\ref{fig:type}.

\begin{figure}
\centering
\includegraphics[width=3.5in]{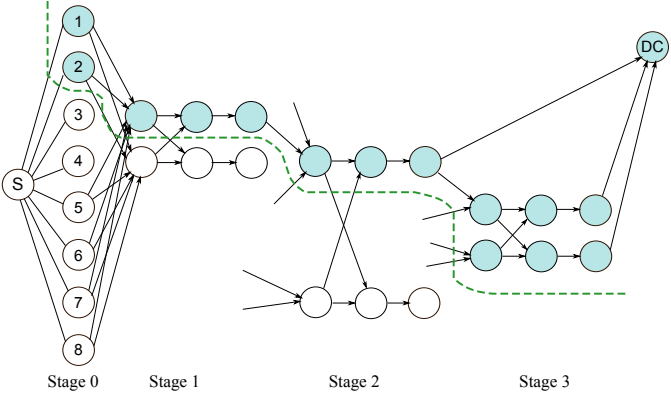}
\caption{A cut of type $(2,1,1,2)$ in a distributed storage system with parameters $d=6$, $k=6$ and $r=2$.}
\label{fig:type}
\end{figure}

\begin{theorem} For any $(s+1)$-tuples of integers $(\ell_0,\ell_1,\ldots, \ell_s)$ satisfying $\sum_{j=0}^s \ell_j = k$ and $1\leq \ell_j \leq r$ for  $j\geq 1$, the file size $B$ is upper bounded by
\begin{equation}
\ell_0 \alpha +\sum_{j=1}^s \Big[\ell_j \big(d-\sum_{i=0}^{j-1}\ell_i\big)\beta_1 + \ell_j(r-\ell_j)\beta_2 \Big].
 \label{eq:cut_capacity0}
\end{equation}
\label{thm:cut0}
\end{theorem}

(The value of $d-\sum_{i=0}^{j-1}\ell_i$ in \eqref{eq:cut_capacity0} is nonnegative, because the summation of the $\ell_i$'s is no larger than $k$, and $k$ is assumed to be less than or equal to $d$.)

\begin{proof}
Let $(\ell_0,\ell_1,\ldots, \ell_s)$ be an $(s+1)$-tuple satisfying the condition in the theorem. Since we take the minimum over all information flow graphs in the max-flow bound~\eqref{eq:maxflow1}, it suffices to show that there exists an information flow graph $G$, in which we can find a cut of type $(\ell_0,\ell_1,\ldots, \ell_s)$, whose capacity is equal to~\eqref{eq:cut_capacity0}.  Then it follows that the supported file size is less than or equal to~\eqref{eq:cut_capacity0}.

Consider an information flow graphs and the cut $(\mathcal{W}^c, \mathcal{W})$ described as in the beginning of this section. The capacities of the edges terminating at the $\ell_0$ ``out'' vertices at stage~0 in $\mathcal{W}$ sum to $\ell_0 \alpha$. This is the first term in~\eqref{eq:cut_capacity0}. For $j=1$, consider an ``in'' vertex at stage 1 in $\mathcal{W}$. We can re-connect the edges terminating at this ``in'' vertex so that there are exactly $d-\ell_0$ edges which are emanating from some ``out'' vertices in $\mathcal{W}^c$. Thus, The $\ell_1$ ``in'' vertices contribute $\ell_1 (d-\ell_0) \beta_1$ to the summation in~\eqref{eq:cut_capacity0}. The term $\ell_1(r-\ell_1)\beta_2$ is the sum of the edge capacities to the ``mid'' vertices in $\mathcal{W}$ at stage~1.

For $j=2,\ldots, s$, we can re-arrange the edges if necessary, so that for each ``in'' vertices at stage $j$ in $\mathcal{W}$, there are exactly $d-\sum_{i=0}^{j-1} \ell_i$ edges which start from some  ``out'' vertices in $\mathcal{W}^c$. Then, the sum of capacities of the edges terminating at some vertices in $\mathcal{W}$ at stage $j$ is $\ell_j \big(d-\sum_{i=0}^{j-1}\ell_i\big)\beta_1 + \ell_j(r-\ell_j)\beta_2$. This completes the proof of Theorem~\ref{thm:cut0}.
\end{proof}

\begin{theorem}
If a data file of size $B$ is supported by a cooperative regenerating code with parameters $n$, $d$, $k$, $r$, $\alpha$, $\beta_1$ and $\beta_2$, then for $s=0,1,\ldots, k$, we have
\begin{equation}
1 \leq  \frac{\alpha}{B} (k-s)+  \frac{s\beta_1}{B} \Big[  d-k+\frac{s+1}{2}\Big] +\frac{\beta_2}{B} s(r-1) \label{eq:LP1}
\end{equation}
and
\begin{equation}
1 \leq \frac{\alpha}{B}(k-s)+
\frac{\beta_1}{B} \Big[  s(d-k) + \frac{s^2+\Psi_{s,r}}{2}\Big] +\frac{\beta_2}{B} (sr - \Psi_{s,r}),  \label{eq:LP2}
\end{equation}
where $\Psi_{s,r}$ is given in~\eqref{eq:Delta}.
\label{thm:LP}
\end{theorem}

\begin{proof}
The upper bound in~\eqref{eq:LP1} comes from a cut of type
\[
(\ell_0,\ell_1,\ldots, \ell_s)=
 (k-s,\underbrace{1,1,\ldots,1}_s).
\]
The last $s$ components are all equal to~1.
The derivation of~\eqref{eq:LP1} follows from
\begin{align*}
&  \sum_{j=1}^s \ell_j  (d-\sum_{i=0}^{j-1} \ell_i)\\
 & =
(d-k+s)+(d-k+s-1)+\cdots + (d-k+1) \\
&= s\Big[d-k+ \frac{s+1}{2} \Big].
\end{align*}

The upper bound  in~\eqref{eq:LP2} comes from a cut of type
\[ (\ell_0,\ell_1,\ldots, \ell_{Q+1}) =
 (k-s, \underbrace{r,r,\ldots, r}_Q,R),
\]
where $Q$ and $R$ are defined as the quotient and remainder when we divide $s$ by $r$, respectively. ($Q$ and $R$ are integers satisfying $s=Qr+R$ and $0\leq R < r$.)

Straightforward calculations show that
\begin{align*}
&\phantom{=} \sum_{j=1}^{Q+1} \ell_j  (d-\sum_{i=0}^{j-1} \ell_i)  \\
&= Qr\Big[d-k+s- \frac{(Q-1)r}{2}\Big] + (d-k+s-Qr)R \\
&= s(d-k+s) + \frac{1}{2}(Qr^2 - Q^2r^2-2QrR) \\
&= s(d-k+s) + \frac{1}{2}(Qr^2 - s^2 + R^2) \\
&= s(d-k) + \frac{1}{2}(s^2+\Psi_{s,r}).
\end{align*}
We have used the notation $\Psi_{s,r}=Qr^2+R^2$. On the other hand, we have
\begin{align*}
\sum_{j=1}^{Q+1} \ell_j(r-\ell_j)  = R(r-R)
&=(s-Qr)r-R^2 = sr-\Psi_{s,r}.
\end{align*}
This proves the inequality in~\eqref{eq:LP2}.
\end{proof}

{\em Remarks:}\

(i) When $s=0$, the two inequalities in \eqref{eq:LP1} and~\eqref{eq:LP2} are identical and can be simplified to
\begin{equation}B \leq k \alpha. \label{eq:LP3}
\end{equation}

(ii) When $s=1$, \eqref{eq:LP1} and~\eqref{eq:LP2} are also identical and can be written as
\begin{equation}
 B \leq (k-1)\alpha + (d-k+1)\beta_1 + (r-1)\beta_2.
 \label{eq:LP4}
\end{equation}

(iii) We note that the coefficients of $\alpha$, $\beta_1$ and $\beta_2$ in \eqref{eq:LP1} and \eqref{eq:LP2} are non-negative.

(iv) In the special case of a single-loss repair, i.e., when $r=1$, the coefficients of $\beta_2$ in \eqref{eq:LP1} and \eqref{eq:LP2} vanish.

\noindent  {\bf Example:} We can now show that the example of the cooperative regenerating code mentioned in the introductory section is optimal. The system parameters are $B=4$, $\alpha=2$ and $d=k=r=2$.
After putting $s=1$ and $s=2$ in~\eqref{eq:LP2}, we get
\begin{align*}
4 & \leq 2 +\beta_1 + \beta_2, \\
4 & \leq 4\beta_1.
\end{align*}
(We have used the fact that $\Psi_{1,2} = 1$ and $\Psi_{2,2} = 4$.)
If we want to minimize the repair bandwidth $\gamma = 2\beta_1+\beta_2$, over the region $\beta_1 \geq 1$ and $\beta_1+\beta_2 \geq 2$ in the $\beta_1$-$\beta_2$ plane,
the optimal solution is attained at $(\beta_1^*,\beta_2^*)=(1,1)$. The optimal repair bandwidth is thus equal to~$2\beta_1^*+\beta_2^*=3$. The above analysis also shows that if the repair bandwidth is equal to the optimal value~3, the values of $\beta_1$ and $\beta_2$ must both be equal to 1. This is indeed the case in the example given in the introduction.

We note that the bounds in \eqref{eq:LP1} and \eqref{eq:LP2} only depend on the ratios $\alpha/B$, $\beta_1/B$  and $\beta_2/B$.
This motivates the following linear programming problem, with the ratios $\beta_1/B$ and $\beta_2/B$ as variables.

\noindent  {\bf Definitions:} Let $\tilde{\alpha} := \alpha/B$, $\tilde{\beta}_1 := \beta_1/B$, $\tilde{\beta}_2 := \beta_2/B$, and $\tilde{\gamma} := \gamma/B$ be the normalized values of $\alpha$, $\beta_1$, $\beta_2$ and $\gamma$, respectively. Consider the following optimization problem:
\begin{equation}
 \text{Minimize } \tilde{\gamma} = d\tilde{\beta}_1 + (r-1)\tilde{\beta}_2 \label{eq:LP_objective}
\end{equation}
\begin{gather*}
\text{subject to } \text{\eqref{eq:LP1} and \eqref{eq:LP2}  for } s=1,2,\ldots,k, \text{ and } \\
\tilde{\beta}_1, \tilde{\beta}_2\geq 0.
\end{gather*}
This is a  parametric linear programming problem with $\tilde{\alpha}$ being the parameter. Let
$\gamma^*_{\mathrm{LP}}(\tilde{\alpha})$
be the optimal value of this linear program, and
\begin{align} \mathcal{C}_{\mathrm{LP}}(k,d,r) &
:= \{(\tilde{\gamma},\tilde{\alpha})\in\mathbb{R}^2:\,
 \text{ the linear program in } \notag  \\
& \qquad \text{\eqref{eq:LP_objective} has feasible solution }(\tbeta_1,\tbeta_2) \text{ and } \notag \\
& \qquad \tilde{\gamma} = d\tilde{\beta}_1 + (r-1)\tilde{\beta}_2 \} . \label{eq:C_LP}
\end{align}

The region $\mathcal{C}_{\mathrm{LP}}$ is a convex region. Suppose $(\tgamma,\talpha)$ and $(\tgamma',\talpha')$ are in $\mathcal{C}_{\mathrm{LP}}$. This means that we can find $(\tbeta_1, \tbeta_2)$ (resp. $(\tbeta_1',\tbeta_2'))$ satisfying the linear constraints \eqref{eq:LP1} and \eqref{eq:LP2} of the linear program with parameter $\talpha$ (resp. $\talpha'$) for $s=1,2,\ldots,k $, such that $\tgamma = d \tbeta_1 + (r-1)\tbeta_2$ (resp. $\tgamma' = d \tbeta_1' + (r-1)\tbeta'_2$). If
$(\tgamma'',\talpha'') = \lambda (\tgamma,\talpha) + \lambda' (\tgamma',\talpha')$
is a linear combination of $(\tgamma,\talpha)$ and $(\tgamma',\talpha')$, for some constant $0\leq \lambda, \lambda'\leq 1$ and $\lambda+\lambda'=1$, then $\lambda (\tbeta_1, \tbeta_2) + \lambda' (\tbeta_1', \tbeta_2')$ satisfies  the constraints of the linear program with parameter $\lambda \talpha + \lambda' \talpha'$.

At this point, we have established the following relationship
\begin{equation}
\mathcal{C}_{\mathrm{LP}} \supseteq \mathcal{C}_{\mathrm{MF}} \supseteq \mathcal{C}_{\mathrm{AD}} .
\label{eq:inclusion}
\end{equation}
The second inclusion follows from the max-flow bound in network coding, and the first from a weaker form of max-flow-min-cut theorem, namely, the value of any flow is no larger than the capacity of any cut (the weak duality theorem). In the formulation of the linear program, we only consider some specific cuts in the information flow graph. Not all possible cuts are taken into account. Nevertheless, in later sections, we will show by other means that equalities hold in~\eqref{eq:inclusion}.

The bound in Theorem~\ref{thm:LP} is based on the assumption that the download traffic is homogeneous, meaning that a newcomer downloads equal amount of data from $d$ surviving nodes, and each pair of newcomers exchanges equal amount of data. In Appendix~\ref{app:asym}, we show that  at the minimum-storage point, the relaxation of the homogeneity in download traffic does not help in further reducing the repair bandwidth. In the remaining of this paper, we will assume that the download traffic is homogeneous.

\noindent {\bf Example:} Consider a cooperative regenerating code with parameters $d=5$, $k=4$ and $r=3$. The number of nodes $n$ can be any integer larger than or equal to~8. We have following constraints based on~\eqref{eq:LP1} and~\eqref{eq:LP2}:
\begin{equation}
\begin{bmatrix}1\\1\\1\\1\\1\\1\\1\\1 \end{bmatrix} \leq
\begin{bmatrix}
4 & 0 & 0 \\
3& 2 & 2 \\
2&6&2\\
2&5&4\\
1&12& 0\\
1&9&6\\
0&17&2\\
0&14&8
\end{bmatrix}
\begin{bmatrix} \tilde{\alpha} \\ \tilde{\beta_1} \\ \tilde{\beta_2} \end{bmatrix},
\label{eq:LP543}
\end{equation}
with the inequality being understood componentwise. We consider the case when $\tilde{\alpha}=1/4$, i.e., the minimum-storage case. We minimize
$d\tilde{\beta}_1 + (r-1)\tilde{\beta}_2$ subject to $\tilde{\beta}_1, \tilde{\beta}_2 \geq 0$ and the constraints in~\eqref{eq:LP543} by linear programming. The linear constraints and the objective function are illustrated graphically in Fig.~\ref{fig:LP}. The seven solid lines (in blue color) in Fig.~\ref{fig:LP} are the boundary of the half planes associated with the seven constraints (row 2 to row 8) in~\eqref{eq:LP543}. The objective function is shown as a dashed line (in red) passing through the optimal point. The feasible region is the area to the right and above these seven straight lines. The optimal solution $\tilde{\beta}_1=\tilde{\beta}_2 = 0.625$
is indicated by the square in Fig.~\ref{fig:LP}. The optimal repair bandwidth is
\[
\gamma_{\mathrm{LP}}^*(1/4)=  (d+r-1)\cdot 0.625 = (5+3-1)\cdot 0.625 = 4.375.
\]

We take note of a few points on the line $\tilde{\beta}_1=2\tilde{\beta}_2$ in Fig.~\ref{fig:LP}, which will play an important role in solving the linear programming explicitly in the next section.
The point $P_1$ is the intersection point of the straight line associated with row 2, i.e.,  $1=3\tilde{\alpha}+2\tilde{\beta}_1+2\tilde{\beta}_2$, and the line $\tilde{\beta}_1=2\tilde{\beta}_2$. The point $P_2$ is the intersection point of the straight lines associated with rows 3 and 4 in~\eqref{eq:LP543}. The point $P_3$ is the intersection point of the straight lines associated with rows 5 and 6, and so on.

\begin{figure}
\centering
\includegraphics[width=3.7in]{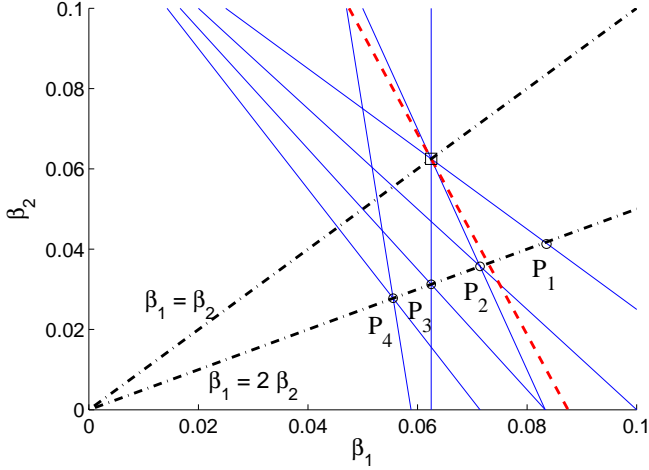}
\caption{Repair bandwidth minimization as a linear program ($d=5$, $k=4$, $r=3$, and $\alpha=1/4$). In this figure, the file size $B$ is normalized to~1, so that $\beta_1$ and $\beta_2$ are the same as $\tilde\beta_1$ and $\tilde\beta_2$, respectively. The objective function $5\tilde\beta_1+2\tilde\beta_2$ is minimized at $\tilde\beta_1 = \tilde\beta_2 = 0.625$. }
\label{fig:LP}
\end{figure}

\section{Solving the Parametric Linear Program}
\label{sec:OP}

In a general parametric linear program with  variables $\mathbf{x} = (x_1,x_2,\ldots, x_n)$, we want to minimize the dot product of $\mathbf{x}$ and a coefficient vector $\mathbf{c}$, subject to $\mathbf{A} \mathbf{x} = \mathbf{b} + \lambda \mathbf{b}^*$, where $\lambda$ is a real-valued parameter, $\mathbf{A}$ is an $m\times n$ matrix, and $\mathbf{b}$ and $\mathbf{b}^*$ are $m$-dimensional vectors. It is well-known that the optimal value of a parametric linear program is a piece-wise linear convex function of the parameter $\lambda$~\cite{KattaMurty}.

The linear program~\eqref{eq:LP_objective} in the previous section is parametric, with
$\talpha$ as the parameter. For a given value of $\talpha,$ and we want to minimize $d\tbeta_1+(r-1)\tbeta_2$, subject to  the constraints \eqref{eq:LP1} and \eqref{eq:LP2}, for $s=1,2,\ldots,k $, and $\tbeta_1,\tbeta_2\geq 0$. The optimal value $\gamma_{\mathrm{LP}}^*(\talpha)$ is a piece-wise linear convex function of the parameter $\talpha$. If $\talpha < 1/k$, the constraint in~\eqref{eq:LP3} is violated. Thus $\gamma^*_{\mathrm{LP}}(\talpha) = \infty$ for $\talpha < 1/k$.
As $\talpha$ increases, the feasible region of the linear program is enlarged, and thus $\gamma_{\mathrm{LP}}^*(\talpha)$ is monotonically non-increasing as a function of~$\talpha$.

Consider the boundary of $\mathcal{C}_{\mathrm{LP}}$, which is the piece-wise linear graph
\[
\{(\gamma_{\mathrm{LP}}^*(\talpha), \talpha):\, \talpha \geq 1/k \} \cup \{(\gamma_{\mathrm{LP}}^*(1/k)+c, 1/k):\, c \geq 0\}.
\]
An operating point $ (\gamma_{\mathrm{LP}}^*(\talpha), \talpha)$ on the boundary of $\mathcal{C}_{\mathrm{LP}}$ is called a {\em corner point} if there is a change of slope,
\[
\frac{ \gamma_{\mathrm{LP}}^*(\talpha+h) - \gamma_{\mathrm{LP}}^*(\talpha)} {h} >
\frac{ \gamma_{\mathrm{LP}}^*(\talpha-h) - \gamma_{\mathrm{LP}}^*(\talpha)} {(-h)},
\]
for all sufficiently small and positive $h$.
In this section we derive all corner  points of the parametric linear program~\eqref{eq:LP_objective}.

\noindent  {\bf Definitions:} For $j=1,2,\ldots, k$, we let $L_j(\tilde{\alpha})$ be the straight line in the $\tilde{\beta}_1$-$\tilde{\beta}_2$ plane with equation
\begin{equation}
1 = (k-j)\tilde{\alpha} + \Big(j(d-k)+ \frac{j^2+\Psi_{j,r}}{2}\Big)\tilde{\beta}_1   +(jr-\Psi_{j,r})\tilde{\beta}_2 \tag{\ref{eq:LP2}'}
\end{equation}
and $L_j'(\talpha)$ be the straight line with equation
\begin{equation}
1 = (k-j)\tilde{\alpha}+ \Big(j(d-k)+ \frac{j^2+j}{2}\Big)\tilde{\beta}_1 +(jr-j)\tilde{\beta}_2, \tag{\ref{eq:LP1}'}
\end{equation}
where $\Psi_{j,r}$ is defined in~\eqref{eq:Delta}.

When $r=1$, we note that for all $j=1,2,\ldots, k$, the lines $L_j(\tilde{\alpha})$ and  $L_j'(\tilde{\alpha})$ coincide, and they are vertical lines in the $\tbeta_1$-$\tbeta_2$ plane (because $\Psi_{j,1}=j$).

We record some geometric facts in the following lemma.

\begin{lemma} Suppose $r>1$.

\begin{enumerate}

\item \label{lemma:itemA} For $j=1,2,\ldots, k$, the magnitude of the slope of $L_j(\talpha)$ is equal to $\mu(j)$.

\item \label{lemma:itemB} For $j=1,2,\ldots k$, the slope of  line $L_j'(\tilde{\alpha})$ is
\[
- \frac{d-k+(j+1)/2}{r-1},
\]
and the magnitude is strictly less than $d/(r-1)$.

\item \label{lemma:itemC} If $r$ divides $j$, then the slope of the line $L_j(\tilde{\alpha})$ is infinite.

\item \label{lemma:itemD} The line $L_1(\tilde{\alpha})$ is identical to the line $L_1'(\tilde{\alpha})$, and the slope has magnitude $\mu(1) < d/(r-1)$.

\item \label{lemma:itemE} For  $j=2,3,\ldots, k$,
the magnitude of the slope of $L_j(\talpha)$ is strictly larger than the magnitude of the slope of $L'_j(\talpha)$.  $L_j(\talpha)$  and $L_j'(\talpha)$ intersect at a point lying on the line $\tbeta_1 = 2\tbeta_2$ in the $\tbeta_1$-$\tbeta_2$ plane.

\item \label{lemma:itemF} $\mu(k) > d/(r-1)$.
\end{enumerate}
\label{lemma:P_j}
\end{lemma}

\begin{proof}\

\ref{lemma:itemA})\ Obvious.

\ref{lemma:itemB})\ The slope of the line $L_j'(\talpha)$ has magnitude
\[
  \frac{j(d-k)+(j^2+j)/2}{jr-j} =  \frac{d-k+(j+1)/2}{r-1},
\]
which is strictly less than $d/(r-1)$ for $r\geq 2$ and  $j\leq k$.

\ref{lemma:itemC}) It follows from the fact that $\Psi_{j,r}=jr$ if $r$ divides $j$.

\ref{lemma:itemD}) When $j=1$, we have $\Psi_{j,r}=j=1$ for all $r \geq 2$.

\ref{lemma:itemE})
For $j=2,3,\ldots, k$, the determinant
\[
\begin{vmatrix}
j(d-k)+\frac{j^2+\Psi_{j,r}}{2} & jr - \Psi_{j,r} \\
j(d-k)+\frac{j^2+j}{2} & jr - j
\end{vmatrix}
\]
is equal to
\[
j(\Psi_{j,r}-j)[d-k +(r+j)/2].
\]
Since $\Psi_{j,r}>j$ for $j\geq 2$, and $d \geq k$ by assumption, the determinant is positive, and thus the magnitude of the slope of $L'(\talpha)$ is strictly larger than the magnitude of the slope of $L(\talpha)$. By subtracting (\ref{eq:LP1}') from (\ref{eq:LP2}'), we obtain $\tbeta_1 = 2\tbeta_2$ after some simplifications.

\ref{lemma:itemF}) The inequality $\mu(k) > d/(r-1)$ is equivalent to
\begin{equation}
\begin{vmatrix}
  k(d-k) + \frac{k^2+\Psi_{k,r}}{2} & kr-\Psi_{k,r} \\
  d & r-1
  \end{vmatrix} > 0.
\label{eq:itemF}
\end{equation}
To prove the above inequality, we distinguish two cases:

{\em Case 1, $k<r$:} We have $\Psi_{k,r} = k^2$ in this case. Hence, the determinant in~\eqref{eq:itemF} can be simplified to
\[
\begin{vmatrix}
  kd & kr-k^2 \\
  d & r-1
  \end{vmatrix} = dk(k-1) ,
\]
which is clearly positive.

{\em Case 2, $k\geq r$:} Write $k=Qr+R$, where $Q$ and $R$ are, respectively, the quotient and the remainder we obtain when $k$ is divided by $r$. Since $k\geq r$ by assumption, we have $Q\geq 1$. The determinant in~\eqref{eq:itemF} becomes
\[
\begin{vmatrix}
k(d-k) + (k^2+\Psi_{k,r})/2 & R(r-R) \\
d & r-1
\end{vmatrix}.
\]
If $R=0$, then the determinant is
$$(r-1)(k(d-k)+(k^2+\Psi_{k,r})/2) > 0.$$
(Recall that we assume $d \geq k$ in this paper.)

For $1\leq R < r$, this determinant can be lower bounded by
\begin{align*}
&\begin{vmatrix}
k(d-k) + (k^2+\Psi_{k,r})/2 & R(r-1) \\
d & r-1
\end{vmatrix} \\
&= (r-1)[k(d-k) + (k^2+\Psi_{k,r})/2 - Rd] \\
&= (r-1)[(k-r)(d-k) + (k^2+\Psi_{k,r})/2 - Rk ]\\
&= (r-1)\Big[ (k-r)(d-k) + \frac{Q^2r^2 + Q r^2}{2} \Big] > 0.
\end{align*}
This completes the proof of $\mu(k) > d/(r-1)$.
\end{proof}

\smallskip

Motivated by part \ref{lemma:itemD}) and \ref{lemma:itemE}) in the previous lemma, we make the following definition.

\noindent  {\bf Definitions:} For $j= 1,2,3,\ldots, k$, let $P_{j}(\talpha)$ be the intersection point of $L_j(\talpha)$, $L_j'(\talpha)$ and the line $\tbeta_1= 2\tbeta_2$.

\begin{lemma} For $j=1,2,\ldots, k$, the coordinates of $P_{j}(\talpha)$ in the $\tbeta_1$-$\tbeta_2$ plane is
\begin{equation}
 \frac{1-(k-j)\talpha}{j(2d-2k+r+j)} (2,1). \label{eq:P}
\end{equation}
\end{lemma}

\begin{proof}
Put $\tbeta_1 = 2\tbeta_2$ in (\ref{eq:LP1}').
\end{proof}

For $d=5$, $k=4$ and $r=3$, the points $P_j(1/4)$, for $j=1,2,3,4$, are shown in Fig.~\ref{fig:LP}.
By the above lemma, we can explicitly calculate their coordinates:
\begin{align*}
P_1(1/4) &= (1/12,1/24)=(0.0833,0.0417), \\
P_2(1/4) &= (1/14,1/28)=(0.0714,0.0357), \\
P_3(1/4) &= (1/16,1/32)=(0.0625,0.0313), \\
P_4(1/4) &= (1/18,1/36)=(0.0556,0.0278) .
\end{align*}

From the expression \eqref{eq:P}, we observe that if we increase $\talpha$ gradually, the points $P_1(\talpha)$ to $P_k(\talpha)$ will ``slide down'' along the line $\tbeta_1=2\tbeta_2$ with various speed. In the following, we compute the value of $\talpha$ such that $P_j(\talpha)$ and $P_{j-1}(\talpha)$ coincide, for $j=2,3,\ldots, k$. It suffices to solve the following system of two linear equations
\begin{align*}
1 &= (k-j) \talpha + \Big( j(d-k) + \frac{j^2+j}{2} \Big) \tbeta_1+ j(r-1) \frac{\tbeta_1}{2}, \\
1 &= (k-j+1) \talpha + \Big( (j-1)(d-k) + \frac{(j-1)^2+j-1}{2} \Big) \tbeta_1 \\
& \qquad + (j-1)(r-1) \frac{\tbeta_1}{2}.
\end{align*}
for $\talpha$ and $\tbeta_1$. The short-hand notation $D_j$ defined in \eqref{eq:Dj} is precisely the determinant of this system of equations. We can write the solution as
\begin{align*}
\tbeta_1 &=\frac{1}{k(d-k+j+\frac{r-1}{2})-\frac{j(j-1)}{2}} = \frac{1}{D_j} , \\
\talpha &= \frac{1}{D_j} (d-k+j + \frac{r-1}{2} ).
\end{align*}
This gives the operating point of the first type in~\eqref{eq:extreme_point1}. For $j=2,3,\ldots, k$, the constant ${\talpha_j}$ in~\eqref{eq:tilde_alpha1} is defined such that
$$P_j({\talpha}_j) = P_{j-1}({\talpha}_j).$$
The corresponding repair bandwidth is
\[
 \tgamma = d\tbeta_1+(r-1)\tbeta_2 = \frac{1}{D_j} (d + \frac{r-1}{2} ).
\]
We have thus derived the  operating points of the first type.

\smallskip

For the operating points of the second type,
we begin with the observation that the lines $L_1(1/4)$, $L_2(1/4)$ and $L_3(1/4)$ in Fig.~\ref{fig:LP} intersect at the same point on the line $\tbeta_1=\tbeta_2$.
The operating points of the second type are obtained by generalizing this observation. For notational convenience, we let $L_0(\talpha)$ be the set of points in the $\tbeta_1$-$\tbeta_2$ plane satisfying the equation $1 = \talpha k$, i.e., it is either the whole plane if $\talpha = 1/k$ or the empty set if $\talpha > 1/k$.

\begin{lemma}
Let $\ell$ be an integer between 0 and $\lfloor k/r \rfloor$.
We can choose $\talpha$ such that $L_j(\talpha)$, for $j=\ell r$, $\ell r +1,\ldots, \ell r+r$, and the line $\tbeta_1=\tbeta_2$ have a common intersection point in the $\tbeta_1$-$\tbeta_2$ plane.
\label{lemma1}
\end{lemma}

\begin{proof}
Let $j$ be an integer between $\ell r$ and $\ell r + r$. We write $j=\ell r + c$ for some integer $c$ in the range $0\leq c \leq r$. In terms of $r$ and $c$, we get
\[
 \Psi_{\ell r + c,r} = \ell r^2 + c^2.
\]
For $0 \leq c \leq r$,
we re-write the equation of $L_{\ell r+ c}(\talpha)$ in (\ref{eq:LP2}') as
\begin{align}
 1 &= (k-\ell r - c) \talpha + \Big( (\ell r + c) (d-k) \notag\\
 & \qquad  + \frac{\ell^2 r^2 + \ell r^2 + 2\ell r c + 2c^2}{2} \Big) \tbeta_1 + c(r-c) \tbeta_2. \label{eq:equation_second_type}
\end{align}
We want to prove that the above equation, for $c=0,1,\ldots, r$, and $\tbeta_1=\tbeta_2$, have a common solution.

After substituting $\tbeta_1 = \tbeta_2$, in \eqref{eq:equation_second_type}, we obtain
\begin{align*}
1 &= c\big(-\talpha +(d-k+r\ell+r)\tbeta_1 \big)  \\
& \qquad +(k-\ell r)\talpha + \Big(\ell r (d-k) + \frac{r^2\ell(\ell+1)}{2} \Big) \tbeta_1.
\end{align*}
We note that the terms involving $c^2$ in \eqref{eq:equation_second_type} are canceled.
If we take $\talpha = (d-k+r\ell+r)\tbeta_1$, we can eliminate $c$ in the above equation and get
\begin{align*}
1  &= (k-\ell r)(d-k+r\ell +r)\tbeta_1  \\
& \qquad \ + \Big(\ell r (d-k)
+ \frac{r^2\ell(\ell+1)}{2} \Big) \tbeta_1 ,
\end{align*}
which can be further simplified to
\[
\tbeta_1 = \frac{1}{k(d+r(\ell+1)-k) - \frac{r^2\ell(\ell+1)}{2}} = \frac{1}{D_\ell'}.
\]
Hence, when $\talpha = (d+r(\ell+1)-k)/D_\ell'$, the point $(1/D_\ell', 1/D_\ell')$ in the $\tbeta_1$-$\tbeta_2$ plane is a common intersection point of $L_j(\talpha)$, for $j=\ell r$, $\ell r + 1, \ldots, \ell r + r$.
\end{proof}

\noindent  {\bf Definition:} For $\ell=0,1,2,\ldots, \lfloor k/r\rfloor$, define $Q_\ell$ as the point
\begin{equation}
Q_\ell := (1/D_\ell', 1/D_\ell')
\label{eq:QL}
\end{equation}
in the $\tbeta_1$-$\tbeta_2$ plane.

The points $Q_\ell$, for $\ell=0,1,2,\ldots, \lfloor k/r\rfloor$, correspond to the operating points of the second type in~\eqref{eq:extreme_point2}. When $\ell =0 $, we have
\begin{align*}
Q_0 & = \Big(\frac{1}{k(d+r-k)}, \frac{1}{k(d+r-k)} \Big),
\end{align*}
which corresponds to the MSCR point $(\tgamma_{\mathrm{MSCR}},
\talpha_{\mathrm{MSCR}}) = ((d+r-1)/(k(d+r-k)),1/k)$.

An illustration is shown in Fig.~\ref{fig:t2}. The point $Q_1$ is the point marked by a square on the line $\tbeta_1 = \tbeta_2$. This is the common intersection point of $L_3(7/117)$, $L_4(7/117)$, $L_5(7/117)$ and $L_6(7/117)$. Lines $L_3(7/117)$ and $L_6(7/117)$ coincide.

\begin{figure}
\centering
\includegraphics[width=3.8in]{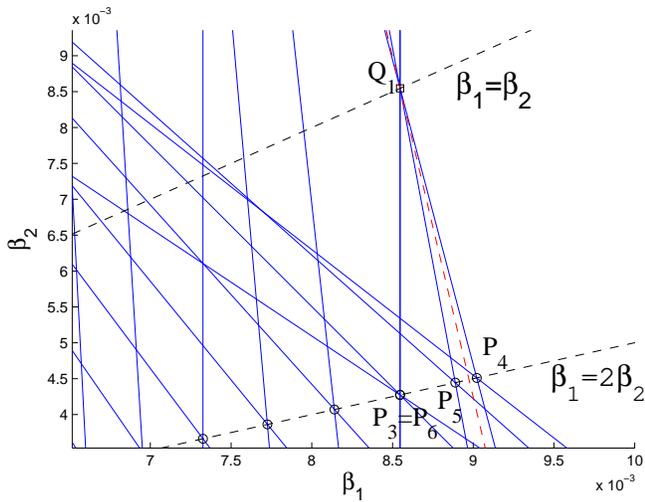}
\caption{The linear programming problem for $d=19$, $k=18$, $r=3$, $B=1$, and $\talpha = 7/117= 0.0598$. The file size $B$ is normalized to~1, so that $\beta_1$ and $\beta_2$ are the same as $\tilde\beta_1$ and $\tilde\beta_2$, respectively. The objective function $19\tbeta_1+2\tbeta_2$ is minimized at the point $Q_1 = (1/117, 1/117) = (0.00854,0.00854)$.}
\label{fig:t2}
\end{figure}

\begin{theorem}
The corner points of the parametric linear program in~\eqref{eq:LP_objective} are precisely the operating points in
\begin{equation}
\Big\{ (\tgamma_j, \talpha_j):\, j=2,3,\ldots, k-1, \ d\leq (r-1)/\mu(j)\Big\},
\end{equation}
and
\begin{equation}
\Big\{ (\tgamma_{\lfloor j/r\rfloor}', \talpha_{\lfloor j/r\rfloor}'):\, j=2,3,\ldots, k-1,\ d> (r-1)/\mu(j)\Big\},
\end{equation}
and $(\tgamma_k, \talpha_k)$ and $(\tgamma_0', \talpha_0')$.
\label{thm:cornerpoint}
\end{theorem}

The proof is technical and is given in Appendix~\ref{app:cornerpoint}.

{\em Remarks:} In the special case of single-loss recovery, i.e., when $r=1$, the variable $\tbeta_2$ can take any value without affecting the repair bandwidth, because the second phase of repair is vacuous. This is reflected by the geometrical fact that the line $L_j(\talpha)$ and $L_j'(\talpha)$ representing the linear constraints are vertical lines in the $\tbeta_1$-$\tbeta_2$ plane. Naturally, we take $\tbeta_2=0$ in the repair of a single failed node. However, in order to give a unified treatment covering both single-loss recovery $r=1$ and mutli-loss recovery $r\geq 2$, we allow the variable $\tbeta_2$ to take positive value in the single-loss case. When $r=1$, the $\tbeta_2$ coordinates of $P_j(\talpha_j)$ and $Q_\ell$ are nonzero, but it does not matter because in the calculation of repair bandwidth $d\tbeta_1 + (r-1)\tbeta_2$, we multiply $\tbeta_2$ by~0. The results in the next two sections hold for {\em all} $r\geq 1$.

\section{Construction of Maximal Flow}
\label{sec:max_flow}

In this section, the parameters $B$, $\alpha$, $\beta_1$ and $\beta_2$ are assumed to be integers. There is no loss of generality because we can always scale them up by a common factor.

We modify the information flow graph by adding more ``out'' vertices, so that at each stage, each storage node is associated with a unique ``out'' vertex. If the storage node is not repaired at stage $s$, we draw a directed edge with infinite capacity from the ``out'' node at stage $s-1$ to it. With these new vertices, all inter-stage edges are between two consecutive stages. A data collector connects to $k$ ``out'' vertices at the same stage.

A modified information flow graph is denoted by $G^m(n,d,k,r;\alpha,\beta_1,\beta_2)$. As an example,
the modified information flow graph $G^m(6,4,3,2;7,2,1)$ for the example in Fig.~\ref{fig:flow1} is shown in Fig.~\ref{fig:modified}.

\begin{figure}
\centering
\includegraphics[width=3.5in]{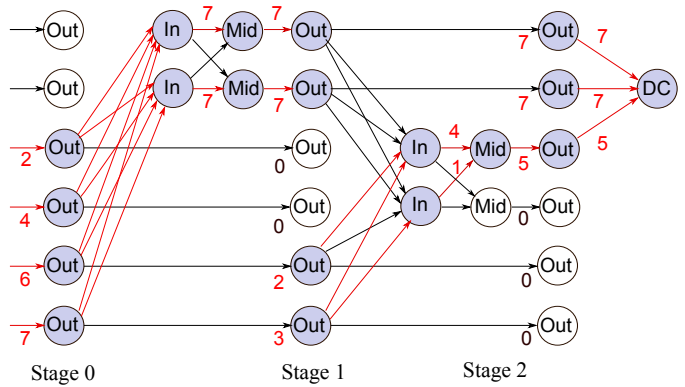}
\caption{An example of modified information flow graph ($n=6$, $d=4$, $k=3$, $r=2$, $\alpha=7$, $\beta_1=2$, $\beta_2=1$). Nodes 1 and 2 are repaired in the transition from stage 0 to stage 1 ($\mathcal{R}_1=\{1,2\}$). Nodes 3 and 4 are repaired in the transition from stage 1 to stage 2 ($\mathcal{R}_2=\{3,4\}$).} \label{fig:modified}
\end{figure}

In this section we study the ``vertical'' cuts  that separate two consecutive stages in the modified information flow graph.

\noindent  {\bf Definition:}
A vector $\mathbf{v}\in\mathbb{R}_+^n$ is called {\em transmissive at stage $s$} ($s\geq 0$), if in all possible modified information flow graph $G^m(n,d,k,r;\alpha,\beta_1,\beta_2)$, we can assign a non-negative real number $F(e)$ to the edges $e$ at and before stage~$s$, such that

(i) for every edge $e$ at or before stage $s$, $F(e)$ does not exceed the capacity of edge $e$,

(ii) for all vertices at stage $1$ to $s-1$, the  in-flow is equal to the out-flow, i.e.,
\[
 \partial F(\{\nu\}) = 0
\]
for all vertices $\nu$ between stage 1 and $s-1$ in the modified information flow graph,

(iii) the in-flow of the $i$-th ``out'' vertex at stage $s$ is equal to the $i$-th component in the given vector~$\mathbf{v}$.

Let $\Upsilon_s$ to be the set of transmissive vectors at stage $s$. A vector $\mathbf{v}\in\mathbb{R}_+^n$ which is transmissive at all stages is called {\em transmissive}, i.e., a vector is transmissive if and only if it belongs to $\cap_{s \geq 0} \Upsilon_s$.

Some comments on transmissive vectors are in order.

(a) To determine whether a vector is transmissive at stage $s$, we have to consider all possible information flow graphs with at least $s$ stages.

(b) No data collector is involved in the definition of transmissive vectors. The number of non-zero components in a transmissive vector may be more than~$k$.

(c) A vector which is transmissive at one stage may not be transmissive at another stage. For example, the vector $(\alpha, \alpha, \ldots, \alpha)$ with all components equal to $\alpha$ is in $\Upsilon_0$, but not in $\Upsilon_s$ for $s\geq 1$. This is why we need to take the intersection $\cap_{s \geq 0} \Upsilon_s$ in the definition of transmissive vectors.

It is trivial that the all-zero vector is a transmissive vector.
We next show that non-trivial transmissive vectors exist. In Theorem~\ref{thm:MBCR}, we show that the vectors in a certain base-polymatroid are transmissive, corresponding to the operating points of the first type. In Theorem~\ref{thm:MSCR}, we show that the vectors in another base-polymatroid are transmissive, corresponding to the operating points of the second type.

\smallskip

For $z= 0 , 1, \ldots, k-2$, let
\begin{equation}
\mathbf{p}_z :=
( \underbrace{\alpha, ... , \alpha}_{z+1 \text{ times}}, \underbrace{
 \alpha-2, \alpha-4 , ..., \alpha-2(k-z-1)}_{k-z-1 \text{ terms}}, \underbrace{0, ...,0}_{n-k \text{ times}} ),
\label{eq:AP}
\end{equation}
with components in non-increasing order.
For $j=0,1,\ldots, n$, let \begin{equation}
\theta_j:=
(\min(k,j) \cdot \alpha ) - \sum_{i=0}^{\min(k,j)-z-1} 2i
\label{eq:modular_MBCR}
\end{equation}
be the sum of the first $j$ components of the vector $\mathbf{p}_z$. (If the upper limit of a summation is negative, the summation is equal to 0 by convention.) Note that $\theta_0=0$ and
\begin{align*}
& \phantom{=}\theta_k = \theta_{k+1} = \cdots =\theta_n  =k\alpha - \sum_{i=1}^{k-z-1} 2i \\
&=k \alpha - (k-z-1)(k-z).
\end{align*}

\begin{theorem}
Let $z$ be an integer between 0 and $k-2$, and let
\begin{align*}
 \alpha &=2(d-z)+r-1, \text{ and }\\
 \beta_1 &= 2, \beta_2  = 1.
\end{align*}
If $\mathbf{h}\in \mathbb{R}_+^n$ is majorized by $\mathbf{p}_z$,
then $\mathbf{h}$ is transmissive.
Hence,  we can construct a flow to any possible data collector with flow value $$\theta_k = k(2d+r-k)-z-z^2.
$$
Furthermore, if the components of the vector $\mathbf{h}$  are non-negative integers, then the flow can be chosen to be integral.
\label{thm:MBCR}
\end{theorem}

Let $f$ be the rank function on $\{1,2,\ldots, n\}$ defined by
$ f(\mathcal{S}) = \theta_j$ for $\mathcal{S}\subseteq \{1,2,\ldots, n\}$ with $|\mathcal{S}| = j$. By Lemma~\ref{lemma:polymatroid}, the base polymatroid $\mathcal{B}(f)$ consists of the vectors in $\mathbb{R}_+^n$ which are majorized by vector $\mathbf{p}_z$ in~\eqref{eq:AP}. Theorem~\ref{thm:MBCR} asserts that the vectors in $\mathcal{B}(f)$ are transmissive.

For a distributed storage system with parameters as in the example in Fig.~\ref{fig:modified}, we can apply Theorem~\ref{thm:MBCR} with $z=1$ and show that any vector in $\mathbb{R}_+^6$ majorized by $(7, 7, 5, 0, 0, 0)$ is transmissive.

The proof of Theorem~\ref{thm:MBCR}  relies on the layered structure of the modified information flow graph, and the important property that the subgraph obtained by restricting to one stage is isomorphic to the subgraph obtained by restricting to another stage. This allows us to reduce the analysis to only one stage.

Consider the subgraph of the modified information flow graph consisting of the vertices at stage $s$ and the $n$ ``out'' vertices at stage $s-1$. We call this the {\em auxiliary graph}, and let $\mathcal{V}'$ be the vertex set of this auxiliary graph. By re-labeling the storage nodes, we assume without loss of generality that nodes 1 to $r$ are regenerated at stage~$s$. The first $r$ ``out'' vertices at stage $s-1$ are disconnected from the rest of the auxiliary graph. In order to distinguish the ``out'' vertices at stage $s-1$ and $s$, we re-label the $n-r$ ``out'' vertices at stage $s-1$ by $v_{r+1}, v_{r+2},\ldots, v_{n}$. An example for $n=6$ and $d=r=3$ is given in  Fig.~\ref{fig:auxiliary}.

\begin{figure}
\centering
\includegraphics[width=2.3in]{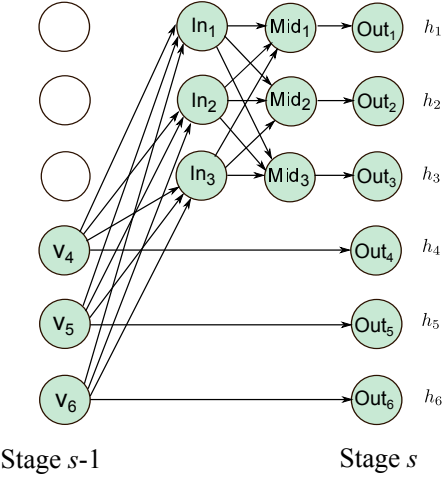}
\caption{An example of an auxiliary graph.} \label{fig:auxiliary}
\end{figure}

The construction of the flow in Theorem~\ref{thm:MBCR} is recursive. We consider the vertices on the left of the auxiliary graph as input vertices and the vertices on the right as output vertices.
Let $\mathbf{h}$ be a vector in $\mathcal{B}(f)$. The  components in $h$ are majorized by the vector in~\eqref{eq:AP} and the sum of the components is equal to $\theta_n$. The vector $\mathbf{h}$ is regarded as the demand from the ``out'' vertices on the right-hand side of the auxiliary graph. We want to look for a valid flow assignment in the auxiliary graph such that the flow to each ``out'' vertices is equal to the corresponding components in $\mathbf{h}$, and meanwhile, the input flow assignment is in the base polymatroid $\mathcal{B}(f)$.

Define a submodular function $\sigma:2^{\mathcal{V}'} \rightarrow \mathbb{R}_+$ as follows. Let $\mathcal{O}_{s-1}$ be the set of ``out'' vertices $\{v_{r+1}, v_{r+2},\ldots, v_n\}$ at stage $s-1$, and $\mathcal{O}_{s}$ be the set of ``out'' vertices $\{\Out_1, \Out_2,\ldots, \Out_n\}$ at stage $s-1$. Given a subset $\mathcal{S}$ of vertices in the auxiliary graph, define
\[
\sigma(\mathcal{S}) := f(\mathcal{S}\cap \mathcal{O}_{s-1}) - \mathbf{h}(\mathcal{S}\cap \mathcal{O}_s).
\]
The notation $\mathbf{h}(\mathcal{S}\cap \mathcal{O}_s)$ in the above definition means
\[
\mathbf{h}(\mathcal{S}\cap \mathcal{O}_s) = \mathop{\sum_{i}}_{\Out_i\in\mathcal{S}}  h_i.
\]
The function $\sigma(\mathcal{S})$ is submodular because it is the sum of a submodular function $f(\mathcal{S}\cap \mathcal{O}_{s-1})$ and a modular function $ -\mathbf{h}(\mathcal{S}\cap \mathcal{O}_s)$. Also, we note that
$$\sigma(\mathcal{V}')= f(\mathcal{O}_{s-1}) - h_1 -h_2-\ldots- h_n =\theta_n - \theta_n=0.$$

We define upper bounds and lower bounds on the edges in the auxiliary graph. For $i=r+1, r+2,\ldots, n$, the  edge joining $v_i$ and $\Out_i$ has lower bound and upper bound equal to~$h_i$. An edge terminating at an ``in'' vertex has lower bound 0 and upper bound $\beta_1$. An edge from $\In_i$ to $\Mid_j$ for $i\neq j$, has lower bound 0 and upper bound $\beta_2$, while an edge from from $\In_i$ to $\Mid_j$ for $i = j$, has lower bound 0 and upper bound $\infty$. An edge from a ``mid'' vertex to an ``out'' vertex has lower bound 0 and upper bound~$\alpha$. We summarize the lower and upper bounds on the edges in the auxiliary graph as follows.
\[
\begin{array}{|c|c|c|} \hline
\text{Edge } e & \lb(e) & \ub(e) \\ \hline \hline
 (v_i, \Out_i) & h_i & h_i \\ \hline
 (v_i, \In_j) & 0 & \beta_1 \\ \hline
 (\In_i, \Mid_j), i\neq j& 0 & \beta_2 \\ \hline
 (\In_i, \Mid_j), i= j& 0 & \infty \\ \hline
 (\Mid_i, \Out_i) & 0 & \alpha \\ \hline
\end{array}
\]

To apply Theorem~\ref{thm:Frank}, we need to verify that condition~\eqref{eq:cut}
holds for all subsets $\mathcal{S}\subseteq \mathcal{V}'$.

\begin{lemma} With notation as in Theorem~\ref{thm:MBCR}, we have
\begin{equation}
\lb(\Delta^+{\mathcal{S}}) - \ub(\Delta^- {\mathcal{S}}) \leq \sigma(\mathcal{S}), \label{eq:AAA}
\end{equation}
for all $\mathcal{S} \subseteq \mathcal{V}'$.
\label{lemma:MBCR}
\end{lemma}

The proof of Lemma~\ref{lemma:MBCR} is given in Appendix~\ref{app:lemmaMBCR}.

\begin{proof}[Proof of Theorem~\ref{thm:MBCR}]
We proceed by induction on stages. Let $\mathbf{h}$ be a vector in $\mathcal{B}(f)$. Since each component of $\mathbf{h}$ is less than or equal to~$\alpha$, we can always assign a flow on the edges from the source vertex to the vertices at stage 0 such that $\mathbf{h}^{(0)}=\mathbf{h}$, without violating any capacity constraint. Hence $\mathbf{h}$ is transmissive at stage~0.

Suppose that all vectors in $\mathcal{B}(f)$ are transmissive at stage $s-1$. Consider the auxiliary graph consisting of the vertices at stage $s$ and the $n$ ``out'' vertices at stage $s-1$. By applying Frank's theorem (Theorem~\ref{thm:Frank}), there exists a feasible submodular flow  on the auxiliary graph. Let $\phi$ be a submodular flow on the auxiliary graph.

By the defining property of a submodular flow, we have
\[
\partial \phi(\{ \Out_i \}) = -\phi(\Delta^- \Out_i) \leq -h_i,
\]
and
\[
  \partial \phi(\{v_{r+1},v_{r+2},\ldots, v_n\}) \leq f(\mathcal{O}_{s-1}) = \theta_n.
\]
Let $\mathcal{S}_0$ be the subset
\[
\mathcal{S}_0 := \{\In_1,\In_2,\ldots, \In_r, \Mid_1,\Mid_2,\ldots \Mid_r\}
\]
of vertices in the auxiliary graph.
We have
\begin{align*}
0 = \sigma(\mathcal{S}_0 ) &\geq
  \partial \phi (\mathcal{S}_0)
 =  -\partial \phi (\mathcal{S}_0^c) \\
& =
-\sum_{i=1}^n \partial \phi(\{\Out_i\}) - \partial \phi(\{ v_{r+1},\ldots, v_{n} \} ) \\
&\geq \sum_{i=1}^n h_i - f(\{v_{r+1},v_{r+2},\ldots, v_n\}) \\
& = \theta_n-\theta_n =0.
\end{align*}
Therefore, all inequalities above are in fact equalities. Thus $\phi(\Delta^- \Out_i) = h_i$ for all $i$.

To show that the flow conservation constraint is satisfied for the ``in'' and ``mid'' vertices in the auxiliary graph, we add the inequalities $\partial \phi(\{\In_i\})\leq 0$ and $\partial \phi(\{\Mid_i\})\leq 0$ for $i=1,2,\ldots, r$, and get
\[
 0 \geq \sum_{i=1}^r (\partial \phi(\{\In_i\})+\partial \phi(\{\Mid_i\})) = \partial \phi(\mathcal{S}_0) = \sigma(\mathcal{S}_0)=0.
\]
We note that $\partial \phi(\mathcal{S}_0) = \sigma(\mathcal{S}_0)$ follows from last paragraph.
Since equality holds in the above inequality, we have
$$\partial \phi(\{\In_i\})= \partial \phi(\{\Mid_i\})= 0$$
for $i=1,2,\ldots, r$.

If we take any subset $\mathcal{A}$ of $\{v_{r+1}, v_{r+2}, \ldots, v_n\}$ at stage $s-1$, from the definition of a submodular flow, we obtain
\[
\partial \phi(\mathcal{A}) = \phi(\Delta^+{\mathcal{A}}) \leq \sigma(\mathcal{A}) = f(\mathcal{A}).
\]
The ``input'' at the $(s-1)$-th stage is thus transmissive at stage $s-1$.
By the induction hypothesis, we can assign real values to the edges from stage $-1$ to $s-1$ in the modified information flow graph, such that the flow conservation constraint is satisfied and the in-flow of the ``out'' vertices at stage $s-1$ is precisely the inputs of the corresponding vertices in the auxiliary graph. This gives a flow at the $s$-th stage of the modified information flow graph yielding the desired vector $\mathbf{h}$, and proves that $\mathbf{h}$ is transmissive at stage $s$.

If the components of $\mathbf{h}$ are non-negative integers, by the second statement in Theorem~\ref{thm:Frank}, we can find a flow which is integral. This completes the proof of Theorem~\ref{thm:MBCR}.
\end{proof}

\begin{theorem} For $j=2,3,\ldots, k$, the operating point $(\tilde{\gamma}_j,\tilde{\alpha}_j)$  is in $\mathcal{C}_{\mathrm{MF}}(d,k,r)$. Thus, all operating points of the first type are in $\mathcal{C}_{\mathrm{MF}}(d,k,r)$.
\label{thm:admissible_MBCR}
\end{theorem}

\begin{proof}
Consider a data collector \DC\ who connects to $k$ storage nodes at stage~$s$. Let $z$ be an integer between 0 and $k-2$.
We want to construct a flow from the source node to \DC\ such that the flow of the $k$ links from the $k$ ``out'' vertices to the data collector are precisely the non-zero components in~\eqref{eq:AP}, i.e.,
$$\underbrace{\alpha, \alpha, \ldots \alpha}_{z+1 \text{ times}}, \alpha-2, \alpha-4,\ldots, \alpha-2(k-z+1).$$

By Theorem~\ref{thm:MBCR}, for any failure pattern, we can always find a flow with flow value $k(2d+r-k)-z-z^2$, $\alpha=2(d-z)+r-1$ and $\gamma=2d+r-1$. Hence, for $z=0,1,\ldots, k-2$, the operating point
\[
\frac{1}{k(2d+r-k)-z-z^2} \big(2d+r-1, 2(d-z)+r-1 \big)
\]
is in $\mathcal{C}_{\mathrm{MF}}(d,k,r)$. After a change of the indexing variable by
$$z=k-j,$$
we check that the denominator in the above fraction is
\begin{align*}
&\phantom{=} k(2d+r-k)-(k-j)-(k-j)^2 \\
&=k(2d+r-k) - k+j-k^2 + 2kj - j^2 \\
&=k(2d-2k+2j+r-1) + j - j^2 \\
&= 2D_j.
\end{align*}
Thus, for $j=2,3,\ldots, k$, the operating point
\[
(\tilde{\gamma}_j, \tilde{\alpha}_j) = \frac{1}{D_j} \big(d+\frac{r-1}{2}, d-k+j+\frac{r-1}{2} \big)
\]
is in $\mathcal{C}_{\mathrm{MF}}(d,k,r)$.
\end{proof}

\smallskip

Analogous to Theorem~\ref{thm:MBCR} and Theorem~\ref{thm:admissible_MBCR}, we have the following two theorems for the operating points of the second type. For $\ell=0,1,\ldots, \lfloor k/r \rfloor$, let \begin{align}
\mathbf{q}_\ell :=
( \underbrace{\alpha, \ldots, \alpha}_{k-\ell r \text{ times}},
 \underbrace{\alpha-r, \ldots, \alpha-r}_{r \text{ times}}, \underbrace{\alpha-2r, \ldots,\alpha-2r}_{r \text{ times}}, \notag \\
\ldots
,\underbrace{\alpha-\ell r, \ldots,\alpha- \ell r}_{r \text{ times}}
   ,\underbrace{0,  \ldots,0}_{n-k \text{ times}}
).
\label{eq:AP2}
\end{align}

For $j=0,1,\ldots, n$, let
\begin{equation}
\varphi_j := \min(k,j)\alpha - \sum_{i=0}^{\min(k,j)-k+\ell r} \lceil i/r \rceil r
\end{equation}
be the sum of the first $j$ components of the vector $\mathbf{q}_\ell$.
We check that
\begin{align*}
\varphi_k = \varphi_{k+1} = \cdots = \varphi_n &=  (k-\ell r)\alpha + r \sum_{i=1}^\ell (\alpha - ir) \\
&= k\alpha -r^2 \sum_{i=1}^\ell i = k\alpha - r^2 \frac{\ell(\ell+1)}{2}.
\end{align*}

\begin{theorem}
Let $\ell$ be an integer between 0 and $\lfloor k/r \rfloor$, and let
\begin{align*}
 \alpha &=d +r(\ell+1)-k, \text{ and }\\
 \beta_1 &= \beta_2  = 1.
\end{align*}
Every vector $\mathbf{h} \in \mathbb{R}_+^n$ majorized by $\mathbf{q}_\ell$ is transmissive. Hence,  we can construct a flow to any possible data collector with flow value
\[
 \phi_k = k(d+r(\ell+1)-k) - \frac{r^2\ell(\ell+1)}{2}.
\]
Furthermore, if the components of $\mathbf{h}$ are non-negative integers, then the flow can be chosen to be integral.
\label{thm:MSCR}
\end{theorem}

Theorem~\ref{thm:MSCR} asserts that the vectors in  the base-polymatroid $\mathcal{B}(g)$ associated with the rank function $g$ defined by $g(\mathcal{S}) = \varphi_{|\mathcal{S}|}$, for  $\mathcal{S}\in\{1,2,\ldots,n \}$, are transmissive.

The proof of Theorem~\ref{thm:MSCR} is given in Appendix~\ref{app:MSCR}.

\begin{theorem}
For $\ell=0,1,2,\ldots, \lfloor k/r \rfloor$, the operating point $(\tilde{\gamma}_\ell',\tilde{\alpha}_\ell')$  is in $\mathcal{C}_{\mathrm{MF}}(d,k,r)$. Thus, all operating points of the second type are in $\mathcal{C}_{\mathrm{MF}}(d,k,r)$.
\label{thm:admissible_MSCR}
\end{theorem}

The proof of Theorem~\ref{thm:admissible_MSCR} is similar to the proof of Theorem~\ref{thm:admissible_MBCR} and is omitted.

\smallskip

In summary, we have shown that all the corner points in Theorem~\ref{thm:cornerpoint} are in $\mathcal{C}_{\mathrm{MF}}(d,k,r)$. This implies that all operating points in $\mathcal{C}_{\mathrm{LP}}(d,k,r)$ are also in $\mathcal{C}_{\mathrm{MF}}(d,k,r)$. We have thus proved

\begin{corollary}
$\mathcal{C}_{\mathrm{LP}}(d,k,r) = \mathcal{C}_{\mathrm{MF}}(d,k,r)$.
\label{thm:extreme_point}
\end{corollary}

\section{Linear Network Codes for Cooperative Repair}
\label{sec:LN}

The objective of this section is to show that the Pareto-optimal operating points in $\mathcal{C}_{\mathrm{MF}}$ can be achieved by linear network coding, with an explicit bound on the required finite field size.

Let $\mathbb{F}_q$ denote the finite field of size $q$, where $q$ is a power of prime. The size of $\mathbb{F}_q$ will be determined later in this section. In this section and the next section, we scale the value of $B$, $\beta_1$, $\beta_2$, and $\alpha$, so that they are all integers, and normalize the unit of data such that an element in $\mathbb{F}_q$ is one unit of data. The whole data file is divided into a number of chunks, and each chunk contains $B$ finite field elements.
As each chunk of data will be encoded and treated in the same way, it suffices to describe the operations on one chunk of data. A packet is identified with an element in $\mathbb{F}_q$, and we will use ``an element in $\mathbb{F}_q$'', ``a packet'' and ``a symbol'' synonymously.

 A chunk of data is represented by a $B$-dimensional column vector $\mathbf{m} \in \mathbb{F}_q^B$. The data packet stored in a storage node is a linear combination of the components in $\mathbf{m}$, with coefficients taken from $\mathbb{F}_q$. The coefficients associated with a packet form a vector, called the {\em global encoding vector} of the packet. For $i=1,2,\ldots, n$, and $t \geq 0$, the packets stored in node~$i$ are represented by  $\mathbf{M}_{i}^{(t)} \mathbf{m}$, where $\mathbf{M}_i^{(t)}$ is an $\alpha\times B$ matrix and the rows of $\mathbf{M}_{i}^{(t)}$ are the global encoding vectors of the packets in node~$i$ at stage~$t$.
We use superscript $^{(t)}$ to signify that a variable is pertaining to stage~$t$.  We will assume that the global encoding vectors are stored together with the packets in the storage nodes. The overhead on storage incurred by the global encoding vectors can be made vanishingly small when the number of chunks is very large. The $(n,k)$ recovery property is translated to the requirement that the totality of the global encoding vectors in any $k$ storage nodes span the vector space~$\mathbb{F}_q^B$.

\smallskip

The realization of cooperative repair using a linear network code is described as follows.

{\em Stage 0:} For $i=1,2,\ldots,n $, node $i$ is initialized
by storing the $\alpha$ components in $\mathbf{M}_i^{(0)} \mathbf{m}$.

{\em Stage $t$:} We suppose without loss of generality that node 1 to node $r$ fail at stage $t$, and we want to regenerate them at stage $t+1$.

\begin{itemize}
\item {\em Phase 1.} For $j = 1, 2, \ldots , r$ and $i \in \mathcal{H}_{t,j}$,
    the $\beta_1$ packets sent from node $i$ to node $j$  are linear
combinations of the packets stored in node $i$ at stage~$t$. For $\ell=1,2,\ldots, \beta_1$, let the $\ell$-th packet sent from node $i$ to node $j$ be $\mathbf{p}_{ij\ell}^{(t)} \mathbf{M}_i^{(t)} \mathbf{m}$,  where $\mathbf{p}_{ij\ell}^{(t)}$ is a $1\times \alpha $ row vector over $\mathbb{F}_q$.

\item {\em Phase 2.} Stack the $d\beta_1$ received packets by node $j$ into a column vector called $\mathbf{u}_j^{(t)}$.
For $j_1, j_2 \in\{1, 2,\ldots, r\}$ and $j_1 \neq j_2$, node $j_1$ sends $\beta_2$ packets to node $j_2$. For  $\ell=1,2,\ldots, \beta_2$, the $\ell$-th packet sent from node $j_1$ to node $j_2$ is
$\mathbf{q}_{j_1,j_2,\ell}^{(t)} \mathbf{u}_{j_1}^{(t)}$, where $\mathbf{q}_{j_1,j_2,\ell}^{(t)}$ is a ($d\beta_1$)-dimensional row vector over $\mathbb{F}_q$.
\end{itemize}

The $(r-1)\beta_2$ packets received by newcomer $j$ during phase~2 are put together to form an $((r-1)\beta_2)$-dimensional column vector $\mathbf{v}_j^{(t)}$. For $\ell=1,2,\ldots, \alpha$, newcomer $j$ takes the inner product of the vector obtained by concatenating $\mathbf{u}_j^{(t)}$ and $\mathbf{v}_j^{(t)}$, and a vector $\mathbf{r}_{j\ell}^{(t)}$ of length $(d\beta_1+(r-1)\beta_2)$. The resulting finite field element is stored as the $\ell$-th packet in the memory.

The vectors $\mathbf{p}_{ij\ell}^{(t)}$'s, $\mathbf{q}_{j_1,j_2,\ell}^{(t)}$'s and $\mathbf{r}_{j}^{(t)}$'s are called the {\em local encoding vectors}. The components in the local encoding vectors are variables assuming values in $\mathbb{F}_q$. The total number of ``degrees of freedom'' in choosing the local encoding vectors is
$$
N=  rd\beta_1\alpha + r(r-1)\beta_2(d\beta_1) + r\alpha(d\beta_1+(r-1)\beta_2).
$$
We will call these $N$ variables the {\em local encoding variables} at stage~$t$.

The local encoding vectors are chosen in order to satisfy a special property. In the followings, $\mathbf{p}$ is a vector of dimension $n$, whose components are non-negative integers summing to the file size~$B$.

Let $\mathbb{Z}_+^n$ be the set of all vectors of dimension $n$ with non-negative integral components, and
$\mathbf{r}$ be a vector in $\mathbb{Z}_{+}^n$. For each $t\geq 0$ and $\mathbf{h}= (h_1, h_2, \ldots, h_n)$ in $\mathbb{Z}_{+}^n$ majorized by $\mathbf{r}$, let $D_\mathbf{h}^{(t)}$ be the determinant of the matrix obtained by putting together the first $h_i$ rows of $\mathbf{M}^{(t)}_i$ for $i=1,2,\ldots, n$.

\noindent {\bf Regularity property with respect to $\mathbf{r}$:} We say that the regularity property with respect to $\mathbf{r}$ is satisfied if $$D_\mathbf{h}^{(t)}\neq 0$$ for all $t\geq 0$ and all vectors $\mathbf{h}$ in $\mathbb{Z}_+^n$ majorized by $\mathbf{r}$.

\smallskip

We borrow the terminology in~\cite{Kamath_ITA} and call the vector $\mathbf{r}$  the {\em rank accumulation profile}.

We are interested in regularity property with respect to either
$\mathbf{p}_z$ or $\mathbf{q}_\ell$, defined in
\eqref{eq:AP} and \eqref{eq:AP2}, respectively.  The regularity property implies the $(n,k)$ recovery property, because there are precisely $k$ non-zero entries in the rank accumulation profiles in~\eqref{eq:AP} and \eqref{eq:AP2}, and the sum of the components in~\eqref{eq:AP} or \eqref{eq:AP2} is equal to the file size $B$. For example, if we consider $z=1$ in \eqref{eq:AP}, then we have $\alpha= 2(d-1)+r-1$, and the rank accumulation profile in \eqref{eq:AP} becomes
\[
(\alpha, \alpha, \alpha-2, \alpha-4, \ldots, \alpha-2(k-2), \underbrace{0, \ldots,0}_{n-k \text{ times}}).
\]
If the regularity property with respect to this rank accumulation profile is satisfied, then the global encoding vectors in any storage node have rank $\alpha$, the global encoding vectors in any pair of storage nodes have rank $2\alpha$, the global encoding vectors in any three storage nodes have rank $3\alpha-2$, and so~on.

The construction depends on the layered structure of the modified information flow graph defined in the last section, and the factorization of the ``transfer function'' into products of matrices.
We concatenate all packets in the $n$ storage nodes at stage~$t$ into an $(n\alpha)$-dimensional vector, and write
\[
 \mathbf{s}^{(t)} := \mathbf{M}^{(t)}  \mathbf{m},
\]
where $\mathbf{M}^{(t)}$ is the $(\alpha n)\times B$ matrix \[
\mathbf{M}^{(t)} := \begin{bmatrix}
  \mathbf{M}_1^{(t)}  \\
  \mathbf{M}_2^{(t)}  \\
  \vdots\\
  \mathbf{M}_n^{(t)}
 \end{bmatrix}.
\]
At stage 0, the distributed storage system is initialized by $\mathbf{s}^{(0)}= \mathbf{M}^{(t)} \mathbf{m}$. The entries in $\mathbf{M}^{(0)}$ are variables, with values drawn from~$\mathbb{F}_q$.

For $t\geq 1$, the packets at stage $t$ can be obtained  by multiplying $\mathbf{s}^{(t-1)}$ by an $(n\alpha)\times (n\alpha)$ transfer matrix $\mathbf{T}^{(t)}$,
\begin{equation}
 \mathbf{s}^{(t)} = \mathbf{T}^{(t)} \mathbf{s}^{(t-1)}.
 \label{eq:T}
\end{equation}
Suppose that
nodes 1 to $r$ fail and are repaired at stage $t$. The matrix $\mathbf{T}^{(t)}$ can be partitioned into
\begin{equation}
\mathbf{T}^{(t)} = \left[
\begin{array}{c|c}
\mathbf{0} & \mathbf{A} \\ \hline
\mathbf{0} & \mathbf{I}
\end{array}
\right],
\label{eq:AA}
\end{equation}
where $\mathbf{I}$ is the identity matrix of size $(n-r)\alpha \times (n-r)\alpha$, and $\mathbf{A}$ is an $r\alpha \times (n-r)\alpha$ matrix. The entries of $\mathbf{A}$ are multi-variable polynomials with the $N$ local encoding variables at stage~$t$ as the variables. In summary, we can write
\[
\mathbf{s}^{(t)}
= \mathbf{T}^{(t)} \mathbf{T}^{(t-1)} \cdots \mathbf{T}^{(1)}  \mathbf{M}^{(0)} \mathbf{m}.
\]

A multi-variable polynomial is said to be {\em non-zero} if, after expanding it as a summation of terms, there is at least one term with non-zero coefficient.
The {\em local degree} with respect to a given variable is defined as the maximal exponent of this variable, with the maximal taken over all terms.
 A multi-variable polynomial induces a function, called the {\em evaluation mapping}, by substituting the variables by values in~$\mathbb{F}_q$. The next lemma gives sufficient condition under which the induced evaluation mapping is not identically zero.

\begin{lemma}
If $F$ is a non-zero multi-variable polynomial over~$\mathbb{F}_q$ with local degree with respect to each variable strictly less than~$q$,
then we can assign values to the variables such that the polynomial is evaluated to a non-zero value.
\label{lemma:Zippel}
\end{lemma}

We refer the reader to \cite[p.143]{IrelandRosen} or \cite[IV.1.8]{Lang} for a proof of Lemma~\ref{lemma:Zippel}.

\begin{lemma}
The entries of the matrix $\mathbf{A}$ in~\eqref{eq:AA} are multi-variable polynomials with local degree at most 1 in each of the local encoding variables.
\label{lemma:local_degree}
\end{lemma}

\begin{proof}
We can see this by fixing all but one local encoding variables. Then each packet generated during the repair process is an affine function of the variable which is not fixed.
\end{proof}

In the following, we treat the two different types of Pareto-optimal operating points separately.

{\bf Pareto-optimal operating point of the first type:}
Let $z$ be an integer between 0 and $k-2$.  We want to construct a linear cooperative regenerating code with parameters
\begin{align*}
B &= 2D_{k-z} = k(2d+r-k)-z - z^2, \\
\beta_1&=2,\ \beta_2=1,
\alpha = 2(d-z)+r-1, \text{ and }
\gamma = 2d+r-1,
\end{align*}
and rank accumulation profile $\mathbf{p}_z$ given as in~\eqref{eq:AP}.
Let $\mathcal{P}_z$ be the subset of vectors in $\mathbb{Z}_+^n$ which are majorized by $\mathbf{p}_z$, and $|\mathcal{P}_z|$ be the cardinality of~$\mathcal{P}_z$. We will show by mathematical induction that the regularity property can be maintained as the number of stages increases.

At stage 0, we choose the entries in $\mathbf{M}^{(0)}$ such that the regularity property with respect to $\mathbf{p}_z$ holds at stage~0, i.e.,
the determinant $D_{\mathbf{h}}^{(0)}$ defined in the regularity property is non-zero for all $\mathbf{h} \in \mathcal{P}_z$. This is equivalent to choosing the entries in $\mathbf{M}^{(0)}$ such that $\prod_{\mathbf{h}\in\mathcal{P}_z} D_{\mathbf{h}}^{(0)} \neq 0$.
For each $\mathbf{h} \in \mathcal{P}_z$, the entries in $D_{\mathbf{h}}^{(0)}$ are distinct variables. Hence, the local degree of each entry with respect to each local encoding variable is equal to one. We can loosely upper bound the local degree of $\prod_{\mathbf{h}\in\mathcal{P}_z}D_{\mathbf{h}}^{(0)}$ by $|\mathcal{P}_z|$.
By Lemma~\ref{lemma:Zippel}, we can pick $\mathbf{M}^{(0)}$ such that the regularity property is satisfied at $t=0$ if $q > |\mathcal{P}_z|$.

Let $t$ be a stage number larger than or equal to~1.
Suppose that $D_{\mathbf{h}}^{(t-1)}$ is non-zero for all $\mathbf{h} \in \mathcal{P}_z$. For each $\mathbf{h}\in \mathcal{P}_z$, we let $\mathbf{T}_{\mathbf{h}}^{(t)}$ be the $B \times (\alpha n)$ submatrix of $\mathbf{T}^{(t)}$ obtained by extracting the rows associated with  $\mathbf{h}$. If the rows of $\mathbf{T}^{(t)}$ is divided into $n$ blocks, with each block consisting of $\alpha$ rows, then $\mathbf{T}_{\mathbf{h}}^{(t)}$ is obtained by retaining the first $h_i$ rows of the $i$-th block of rows of $\mathbf{T}_{\mathbf{h}}^{(t)}$, for $i=1,2,\ldots, n$. The entries in $\mathbf{T}_{\mathbf{h}}^{(t)}$ involve the local encoding variables to be determined, but the entries in $\mathbf{M}^{(t-1)}$ are fixed elements in~$\mathbb{F}_q$. The determinant $D_{\mathbf{h}}^{(t)}$ can be written as
\[
D_{\mathbf{h}}^{(t)} = \det(\mathbf{T}_{\mathbf{h}}^{(t)} \mathbf{M}^{(t-1)}) .
\]

By Theorem~\ref{thm:MBCR}, there is an integral flow in the auxiliary graph with input $\mathbf{g}$ and output $\mathbf{h}$, for some integral transmissive vector~$\mathbf{g}$. This means that if the local encoding variables are chosen appropriately, the square submatrix of $\mathbf{T}_{\mathbf{h}}^{(t)}$ obtained by retaining the columns associated with $\mathbf{g}$ is a permutation of the identity matrix, while the other columns not associated with $\mathbf{g}$ are zero.  The square submatrix of $\mathbf{M}^{(t-1)}$ obtained by retaining the rows associated with $\mathbf{g}$ has non-zero determinant by the induction hypothesis.  We can thus choose the local encoding variables such that $D_{\mathbf{h}}^{(t)}$ is evaluated to a non-zero value. In particular, $D_{\mathbf{h}}^{(t)}$ is
a non-zero polynomial with the local encoding variables as the variables.

After multiplying $D_{\mathbf{h}}^{(t)}$ over all $\mathbf{h}\in \mathcal{P}_z$, we see that $\prod_{\mathbf{h}\in\mathcal{P}_j}D_{\mathbf{h}}^{(t)}$ is also a non-zero polynomial.
Each local encoding variable appears in at most $r \alpha$ rows in the determinant $D_{\mathbf{h}}^{(t)}$. By Lemma~\ref{lemma:local_degree}, the local degree of $\prod_{\mathbf{h}\in\mathcal{P}_z}D_{\mathbf{h}}^{(t)}$ can be upper bounded by
$r\alpha |\mathcal{P}_z|$. By Lemma~\ref{lemma:Zippel},  we can choose the local encoding vector at stage $t$ such that the regularity property will continue to hold at stage $t$ provided that
$$q > r\alpha|\mathcal{P}_z| =  r(2(d-z)+r-1)|\mathcal{P}_z|.$$
The cardinality of $\mathcal{P}_z$ is a constant that does not depend on the total number of stages nor the total number of data collectors.
After a change of indexing variable $z=k-j$, we see that the operating points $(\tgamma_j, \talpha_j)$, for $j=2,3,\ldots, k$, can be achieved by linear network coding over a sufficiently large finite field.

\smallskip

{\bf Pareto-optimal operating point of the second type:}
Let $\ell$ be an integer between 0 and $\lfloor k/r \rfloor$,  and set
\begin{align*}
B &= D_\ell' = k(d+r(\ell+1)-k) - \frac{r^2 \ell (\ell+1)}{2}, \\
\beta_1&=\beta_2=1, \
\alpha = d-k+r(\ell+1), \text{ and }
\gamma =d+r-1.
\end{align*}

Consider  the rank accumulation profile $\mathbf{q}_\ell$ defined in~\eqref{eq:AP2}.
Let $\mathcal{Q}_\ell$ be the subset of vectors in $\mathbb{Z}_+^n$ which are majorized by $\mathbf{q}_\ell$.
By similar arguments for the operation point of the first type, we can guarantee that the regularity property with respect to $\mathbf{q}_\ell$ is satisfied at all stages provided that the size of the finite field is lower bounded by $$q > r(d-k+r(\ell+1))|\mathcal{Q}_\ell|.$$

The next theorem summarizes the main result in this section.

\begin{theorem}
\label{thm:LNC}
If the size of the finite field $q$ is larger than
\begin{align*}
&\max_{j=2,\ldots,k} r(2(d-k+j)+r-1)|\mathcal{P}_{k-j}|, \text{ and} \\
&\max_{\ell=0,\ldots, \lfloor k/r\rfloor} r(d-k+r(\ell+1))|\mathcal{Q}_\ell|,
\end{align*}
then we can implement
linear network codes over $\mathbb{F}_q$ for functional and cooperative repair, attaining the boundary points of $\mathcal{C}_{\mathrm{MF}}$. Thus,
$ \mathcal{C}_{\mathrm{MF}} = \mathcal{C}_{\mathrm{AD}}$.
\end{theorem}

\begin{proof}
We have already shown that the corner points of $\mathcal{C}_{\mathrm{MF}}$ can be achieved by linear network coding. By an analog of ``time-sharing'' argument, we see that all boundary points of $\mathcal{C}_{\mathrm{MF}}$ are achievable by linear network coding, if the finite field size is sufficiently large. Therefore,  $\mathcal{C}_{\mathrm{AD}} \supseteq \mathcal{C}_{\mathrm{MF}}$.
The reverse inclusion $\mathcal{C}_{\mathrm{AD}} \subseteq \mathcal{C}_{\mathrm{MF}}$ is shown in~\eqref{eq:reverse}. We conclude that $\mathcal{C}_{\mathrm{MF}} = \mathcal{C}_{\mathrm{AD}}$.
\end{proof}

The cardinality of $\mathcal{P}_z$ and $\mathcal{Q}_\ell$ depend on parameters $n$, $k$, $d$ and $r$, but do not depend on the number of stages. Hence a fixed finite field is sufficient to maintain the $(n,k)$ recovery property at all stages. The proof Theorem~\ref{thm:C} is now completed.


\begin{corollary}
The operating point of the first type (in particular the MBCR point $ (\tilde{\gamma}_{\mathrm{MBCR}}, \tilde{\alpha}_{\mathrm{MBCR}})$) is achieved if and only if $\beta_1 = 2\beta_2$. On the other hand, the operating point of the second type (in particular the MSCR point $ (\tilde{\gamma}_{\mathrm{MSCR}}, \tilde{\alpha}_{\mathrm{MSCR}}) $) is achieved if and only if $\beta_1 = \beta_2$.
\label{cor:MSCRMBCR}
\end{corollary}

\section{Two Families of Explicit Cooperative Regenerating Codes}
\label{sec:explicit}

In this section we present two families of explicit constructions of optimal cooperative regenerating codes for exact repair, one for MSCR and one for MBCR. The constructed regenerated codes are {\em systematic}, meaning that the native data packets are stored somewhere in the storage network. Hence, if a data collector is interested in part of the data file, he/she can contact some particular storage nodes and download directly without any decoding.
Both constructions are for the case $d=k$. We note that all single-failure regenerating codes for $d=k$ are trivial, but in the multi-failure case, something interesting can be done when $d=k$.
As in the previous section, a finite field element is referred to as a packet.

\subsection{Construction of MSCR Codes for Exact Repair}

In this construction, the number of packets in a storage node is identical to the number of nodes contacted by a newcomer, namely $\alpha=r$. The parameters of the cooperative regenerating code in the first family are
\begin{align*}
 d &= k, \
 B = kr,\
 n \geq d+r, \\
 \alpha &= r,\
 \gamma=d+r-1.
\end{align*}
The operating point
\[
(\tilde{\alpha}, \tilde{\gamma}) = \frac{1}{B}(\alpha, \gamma) = \big( \frac{1}{k}, \frac{d+r-1}{kr}\big)
\]
attains the MSCR point when $d=k$.

We divide the data file into chunks. Each chunk contains $B=kr$ elements in finite field $\mathbb{F}_q$.
We need $r$ matrices $\mathbf{G}_j$, for $j=1,2,\ldots, r$, as building blocks. For each $j$, the matrix  $\mathbf{G}_j$ is an $n\times k$ matrix over $\mathbb{F}_q$ (with $n > k$), satisfying that property that any $k\times k$ submatrix is non-singular. For example, $\mathbf{G}_j$ may be a Vandermonde matrix with distinct rows. Hence, the finite field size can be any prime power larger than or equal to~$n$. We can also use the same $n\times k$ matrix for all $\mathbf{G}_j$'s, but the construction also works if the $\mathbf{G}_j$'s are different.
For $1\leq a \leq n$, and any $a$ distinct integers $i_1, i_2,\ldots, i_a$ between 1 and $n$, we let the matrix obtained by retaining rows $i_1, i_2,\ldots, i_a$ in $\mathbf{G}_j$ by $\mathbf{G}_j[i_1, i_2,\ldots, i_a]$.

In a chunk of data, there are $B=kr$ source packets.
We divide the $kr$ source packets into $r$ groups, with each group containing $k$ packets. The $r$ groups of packets are represented by
$k$-dimensional column vectors, $\mathbf{m}_1$, $\mathbf{m}_2,\ldots, \mathbf{m}_r$.  For $i=1,2,\ldots, n$, and $j=1,2,\ldots, r$, we store $\mathbf{G}_j[i] \cdot\mathbf{m}_j$ as
the $j$-th packet stored in the $i$-th storage node, where ``$\cdot$''
denotes the dot product of two vectors.
In other words, the $r$ packets stored in node $i$ are
\[
 \mathbf{G}_1[i] \cdot \mathbf{m}_1,\
 \mathbf{G}_2[i] \cdot \mathbf{m}_2,\ \ldots, \
 \mathbf{G}_r[i] \cdot \mathbf{m}_r.
 \]

Suppose that a data collector connects to  nodes $i_1$, $i_2,\ldots,i_k$. It downloads all the $kr$ packets stored in these $k$ nodes, namely, $\mathbf{G}_j[i_\ell] \cdot \mathbf{m}_j$, for $\ell =1,2,\ldots, k$, and $j=1,2,\ldots, r$. For each $j$, the $k$ symbols  $\mathbf{G}_j[i_\ell]\cdot \mathbf{m}_j$, $\ell =1,2,\ldots, k$, can be put together as a column vector
$$ \mathbf{G}_j[i_1,i_2,\ldots, i_k]\cdot
\mathbf{m}_j.$$
The $k\times k$ matrix $\mathbf{G}_j[i_1,i_2,\ldots, i_k]$ is non-singular by construction. We can thus solve for $\mathbf{m}_j$. This establishes the $(n,k)$ recovery property.

Suppose that nodes $i_1$, $i_2,\ldots, i_r$ fail. We want to repair them exactly with repair bandwidth $d+r-1$ per newcomer. In the first phase of the repair process, the $r$ newcomers have to agree upon an ordering among themselves, so that we can talk about the first newcomer, second newcomer, and third newcomer, etc. Suppose that node $i_1$ is the first newcomer, $i_2$ is the second newcomer, and so on.
For $j=1,2,\ldots, r$, newcomer $i_j$ connects to any $d$ surviving storage nodes, say nodes $\nu_{j,1}$, $\nu_{j,2}, \ldots, \nu_{j,d}$, and downloads packet $\mathbf{G}_j[\nu_{j,x}] \cdot \mathbf{m}_j$ from node $\nu_{j,x}$, for $x=1,2,\ldots, d$. We note that no arithmetic operation is required in the first phase, because the packet $\mathbf{G}_j[\nu_{j,x}] \cdot \mathbf{m}_j$ can be read from the memory of node $\nu_{j,x}$ directly. The traffic required in the first phase is $rd$ packet transmissions.
At the end of the first phase,  newcomer $i_j$ can decode $\mathbf{m}_j$ by inverting the $k\times k $ matrix $\mathbf{G}_j[\nu_{j,1}, \nu_{j,2}, \ldots, \nu_{j,d}]$.

In the second phase of the repair process, for $j=1,2,\ldots, r$, newcomer $i_j$ computes and sends $\mathbf{G}_{j}[i_\ell] \cdot \mathbf{m}_{j}$ to newcomer $i_\ell$, for $\ell\in\{1,2,\ldots, r\} \setminus\{ j \}$. This can be done because $\mathbf{m}_{j}$ has been decoded in the first phase, and $\mathbf{G}_{j}[i_\ell]$ is known to every newcomer.
A total of $r(r-1)$ packet transmissions are required in the second phase. To complete the regeneration process, newcomer $i_j$ computes and stores $\mathbf{G}_j[i_j]\cdot \mathbf{m}_{j}$.
The total repair bandwidth equals $r(d+r-1)$ and matches the MSCR operating point.

The example in Section~\ref{section:example}  can be obtained by this construction, with parameters $d=k=r=\alpha=2$, $n=B=4$, and
\[
 \mathbf{G}_1 = \begin{bmatrix}1&0\\0&1\\1&1\\2&1 \end{bmatrix}, \quad
 \mathbf{G}_2 = \begin{bmatrix}1&0\\0&1\\2&1\\1&1 \end{bmatrix}.
\]

\subsection{Construction of MBCR Codes for Exact Repair}
The second construction matches the MBCR point.
The parameters  are
\begin{align*}
 d &= k, \
 B = k(k+r),\\
 n &= d+r, \
 \alpha = \gamma = 2d+r-1.
\end{align*}
The operating point matches the MBCR point for $d=k$,
\[
 (\tilde{\alpha}, \tilde{\gamma}) = \frac{1}{B}(\alpha, \gamma) =
 \big(\frac{2d+r-1}{k(k+r)}, \frac{2d+r-1}{k(k+r)} \big).
\]

In this construction, we need $n$ matrices $\mathbf{H}_i$ as building blocks. For $i=1,2,\ldots, n$, $\mathbf{H}_i$ is an $(n-1)\times k$ matrix over $\mathbb{F}_q$, such that any $k\times k $ submatrix is non-singular. As in the previous construction, we can use Vandermonde matrices for instance, and the field size requirement is thus $q\geq n-1$.

We divide the data into chunks, such that each chunk of data consists of $B=kn$ data packets.  In each chunk we denote the $kn$ data packets by $x_0$, $x_1, \ldots, x_{kn-1}$. We divide these $kn$ packets into $n$ groups. The first group consists of $x_0, x_1,\ldots, x_{k-1}$, the second group consists of $x_{k}, x_{k+1}, \ldots, x_{2k-1}$, and so on. For $i=1,2,\ldots, n$, we represent the  packets in the $i$-th group by row vector
$$\mathbf{x}_{i}:=(x_{(i-1)k}, x_{(i-1)k+1}, \ldots, x_{(i-1)k+k-1}).$$

For $1\leq a \leq n-1$, and any $a$ distinct integers $i_1, i_2,\ldots, i_a$ between 1 and $n-1$, we let the matrix obtained by retaining rows $i_1, i_2,\ldots, i_a$ in $\mathbf{H}_i$ by $\mathbf{H}_i[i_1, i_2,\ldots, i_a]$.  We present the encoding by an $n\times n $ array $\mathsf{A}$ (see Table.~\ref{table:arrayA} for an example). The content of array $\mathsf{A}$ is obtained as follows.
\begin{enumerate}
\item
For $i=1,2,\ldots, n$, the diagonal entry $\mathsf{A}(i,i)$ contains the $k$ packets in $\mathbf{x}_i$.

\item
For $i=1,2,\ldots, n-1$ and $j=i+1,i+2,\ldots, n$, the entry $\mathsf{A}(i,j)$ contains one packet $\mathbf{H}_j[i]\cdot \mathbf{x}_j$.

\item
For $i=2,3,\ldots, n$ and $j=1,2,\ldots, i-1$, the entry $\mathsf{A}(i,j)$ contains one packet $\mathbf{H}_j[i-1]\cdot \mathbf{x}_j$.
\end{enumerate}

We note that for each $i=1,2,\ldots,n$, each of the packets $\mathbf{H}_i[1]\cdot \mathbf{x}_i$,
$\mathbf{H}_i[2]\cdot \mathbf{x}_i, \ldots,
\mathbf{H}_i[n-1]\cdot \mathbf{x}_i$, appears once and exactly once in the $i$-th column of the array. Each diagonal entry of $\mathsf{A}$ contains $k$ packets, while each off-diagonal entry of $\mathsf{A}$ contains 1 packet. For $i=1,2,\ldots, n$, the $i$-th node stores the content of $\mathsf{A}(i,1)$, $\mathsf{A}(i,2),\ldots, \mathsf{A}(i,n)$ in the $i$-th row of array $\mathsf{A}$. The number of packets in a storage  node is
$$k+(n-1) = d+(d+r-1) = 2d+r-1.$$ The encoding  has the important property that the $i$-th node stores a copy of the packets in the $i$-th group of packets uncoded so that node $i$ can compute any packet in the $i$-th column of the array~$\mathsf{A}$.

\begin{table}
\caption{The array $\mathsf{A}$ in the explicit construction of MBCR code for $n=5$, $d=k=3$ and $r=2$.}
\[
\begin{array}{|c|c|c|c|c|} \hline
\mathbf{x}_1 & \mathbf{H}_2[1] \cdot \mathbf{x}_2 & \mathbf{H}_3[1] \cdot \mathbf{x}_3& \mathbf{H}_4[1] \cdot \mathbf{x}_4& \mathbf{H}_5[1] \cdot \mathbf{x}_5 \\ \hline
\mathbf{H}_1[1] \cdot \mathbf{x}_1 & \mathbf{x}_2 & \mathbf{H}_3[2] \cdot \mathbf{x}_3 & \mathbf{H}_4[2] \cdot \mathbf{x}_4 & \mathbf{H}_5[2] \cdot \mathbf{x}_5\\ \hline
\mathbf{H}_1[2] \cdot \mathbf{x}_1 & \mathbf{H}_2[2] \cdot \mathbf{x}_2 & \mathbf{x}_3 & \mathbf{H}_4[3] \cdot \mathbf{x}_4 & \mathbf{H}_5[3] \cdot \mathbf{x}_5 \\ \hline
\mathbf{H}_1[3] \cdot \mathbf{x}_1 & \mathbf{H}_2[3] \cdot \mathbf{x}_2 & \mathbf{H}_3[3] \cdot \mathbf{x}_3 & \mathbf{x}_4 & \mathbf{H}_5[4] \cdot \mathbf{x}_5 \\ \hline
\mathbf{H}_1[4] \cdot \mathbf{x}_1 & \mathbf{H}_2[4] \cdot \mathbf{x}_2  & \mathbf{H}_3[4] \cdot \mathbf{x}_3  & \mathbf{H}_4[4] \cdot \mathbf{x}_4 & \mathbf{x}_5 \\ \hline
\end{array}
\]
\label{table:arrayA}
\end{table}

\begin{table*}
\caption{An MBCR code for $n=5$, $d=k=3$ and $r=2$.}
\[{
\begin{array}{|c||c|c|c|c|c|} \hline
\text{Node }1 & x_0, x_1,x_2  &  x_3        &  x_{6}     & x_{9}              & x_{12}\\ \hline
\text{Node }2 & x_0           & x_3,x_4, x_5&  x_{7}     & x_{10}             &         x_{13}\\ \hline
\text{Node }3 & x_{1}         & x_4        & x_6,x_7,x_8 & x_{11}&         x_{14}\\ \hline
\text{Node }4 & x_{2}         & x_{5}      & x_{8}       & x_{9},x_{10},x_{11} &x_{12}+x_{13}+x_{14} \\ \hline
\text{Node }5 & x_0+x_1+x_2   & x_3+x_4+x_5& x_6+x_7+x_8 & x_{9}+x_{10}+x_{11} & x_{12}, x_{13}, x_{14} \\ \hline
\end{array}
}
\]
\label{table:MBCR5}
\end{table*}

For example, consider the parameters $n=5$, $k=d=3$, $r=2$, $\alpha=7$. Let
\[
\mathbf{H}_1 = \mathbf{H}_2 =
\mathbf{H}_3 = \mathbf{H}_4 =
\mathbf{H}_5 =
\begin{bmatrix}
1&0&0 \\
0&1&0 \\
0&0&1 \\
1&1&1
\end{bmatrix}
\]
be matrices over $\mathbb{F}_2$. We note that any three rows of $\mathbf{H}_i$ are linearly independent over $\mathbb{F}_2$.
A chunk of data consists of $B=15$ packets $x_0$, $x_1,\ldots, x_{14}$, each packet contains one bit.  The content of the storage nodes is shown in the array in Table~\ref{table:MBCR5}.  The packets in each row of the array are the content of the corresponding node.

Using the property of the matrices $\mathbf{H}_i$ that any $k$ rows of $\mathbf{H}_i$ form a non-singular matrix, it is straightforward to check that the $B$ packets in a chunk can be decoded from the content of any $k$ storage nodes.

Suppose that nodes $i_1$, $i_2,\ldots, i_r$ fail, where $i_1, i_2, \ldots, i_r$ are $r$ distinct integers between 1 and $n$.  We generate the content of the new nodes $i_1,i_2,\ldots, i_r$ as follows.

\begin{enumerate}
\item For $j$ in $\{1,2,\ldots, n\}\setminus \{i_1,i_2,\ldots, i_r\}$ and $i\in\{i_1,i_2,\ldots, i_r\}$, the surviving node $j$ computes the packet in $\mathsf{A}(i,j)$ and sends it to newcomer $i$. This is possible because the $j$-th node stores a copy of the packets in the $j$-th group of packets uncoded, and hence can compute any packet in the $j$-th column of the array~$\mathsf{A}$.

\item
For $i\in\{i_1,i_2,\ldots, i_r\}$, the surviving node with index $j$ in $\{1,2,\ldots, n\}\setminus \{i_1,\ldots, i_r\}$ sends the packet in $\mathsf{A}(j,i)$ to the new node~$i$.  After receiving $k$ packets, the new node $i$, for $i\in\{i_1,i_2,\ldots, i_r\}$, is able to recover the $k$ packets in $\mathbf{x}_i$.

\item For  $i,i'\in\{i_1,i_2,\ldots, i_r\}$, $i\neq i'$, the new node $i$ computes the packet in $\mathsf{A}(i,i')$ and sends it to the new node $i'$.
\end{enumerate}

The number of packet transmissions in steps 1, 2 and 3 are $r(n-r)$, $r(n-r)$ and $r(r-1)$, respectively. The total number of packet transmissions in the repair process is thus
\[
 r(n-r+n-r+r-1) = r(2d+r-1),
\]
achieving the minimum repair bandwidth at the MBCR point.

For example, suppose nodes 4 and 5 fails in the example in Table~\ref{table:MBCR5}.
In the first step, node 1 transmits $x_2$ to node 4 and $x_0+x_1+x_2$ to node 5. Node 2 transmits $x_5$ to node 4 and $x_3+x_4+x_5$ to node 5. Node 3 transmits $x_8$ to node 4 and $x_6+x_7+x_8$ to node~5.
In the second step, nodes 1, 2 and 3 send packets $x_9$, $x_{10}$ and $x_{11}$ to node 4, and packets $x_{12}$, $x_{13}$ and $x_{14}$ to node~5.
Finally, node 4 computes $x_9+x_{10}+x_{11}$ and sends it to node 5. Node 5 computes $x_{12}+x_{13}+x_{14}$ and sends it to node~4.
We also observe that in Table~\ref{table:MBCR5}, each row has rank 7, every pair of two rows have rank 12, and every three rows have rank~15.

\section{Concluding Remarks}

We invoke an existence theorem of submodular flow to obtain the value of max-flow in the special class of graph induced from the cooperative scheme for functional repair. By exploiting the layered structure of the information flow graph, the computation of max-flow is decomposed into the analysis of a section of the infinite graph. A closed-form expression of the trade-off between storage and repair bandwidth is derived by determining the rank accumulation profiles at the corner points of the trade-off curve. We also show that the corners point can be achieved by linear network codes.

In the literature, most of the existing works related to the application of submodular flow to deterministic networks focus on the  the computation of max-flow algorithmically. For example, submodular function minimization  are used in \cite{YazdiSavari, Goemans} to determine the capacity of deterministic linear networks introduced in~\cite{ADT11}. Combinatorial algorithm for the computation of the capacity deterministic linear networks can be found in~\cite{AmaudruzFragouli09}. Submodular flow technique is also used in~\cite{Chekuri12} to compute multi-commodity flows in polymatroidal networks~\cite{Lawler}, and in \cite{Milo12} for minimum-cost multicast with decentralized sources. Extension to a more general polylinking flow network is given in~\cite{Fujishige11}.

The MBCR code construction in this paper is generalized in~\cite{Jiekak}. In~\cite{WangZhang}, a construction for all possible parameters on the MBCR operating point is given.  An explicit construction of MSCR code for $k=2$ is presented in~\cite{LeScouarnec12}.
Optimal cooperative regenerating codes beyond the ones presented in this paper and in~\cite{Jiekak, WangZhang, LeScouarnec12} is an interesting direction for future studies.

\appendices

\section{Derivation of the Max-flow-min-cut Theorem from Frank's Theorem}
\label{app:maxflowFrank}

To see that the max-flow-min-cut theorem is a special case of Frank's theorem, we consider a weighted directed graph  $H=(\mathcal{V},\mathcal{E})$ with two distinguished vertices $S$ and $T$. We denote the capacity of an edge $e\in\mathcal{E}$ by $c(e)$, which is a non-negative real number. For a subset $\mathcal{E}'$ of $\mathcal{E}$, we let $c(\mathcal{E'})$ be the sum of the capacities of the edges in $\mathcal{E}'$.

Suppose that the source vertex $S$ has no incoming edge and the sink vertex $T$ has no out-going edge. Let the minimal cut capacity be denoted by~$M$. In the following, we prove the existence of a flow with value $M$ by Theorem~\ref{thm:Frank}. Define a function $f:\mathcal{V}\rightarrow \mathbb{R}$ by
\[
 f(x) = \begin{cases}
 M & \text{ if } x=  S ,\\
 -M & \text{ if } x=  T, \\
 0 & \text{ otherwise},
 \end{cases}
\]
and extend it to a set function by
defining
\[f(\mathcal{S}) = \sum_{x\in \mathcal{S}} f(x)
\]
for $\mathcal{S}\subseteq \mathcal{V}$.
The set function $f$ is a modular, and hence is submodular.

For $e\in\mathcal{E}$, let the lower bound  $\lb(e)$ be identically zero, and the upper bound  $\ub(e)$ be the corresponding edge capacity $c(e)$. Hence,  \eqref{eq:cut} is equivalent to
\begin{equation}
-c(\Delta^-{\mathcal{S}})  \leq f(\mathcal{S}).
\label{eq:app_Frank}
\end{equation}

Now we check that \eqref{eq:app_Frank} is satisfied for all $\mathcal{S}\subseteq \mathcal{V}$ by considering two cases.

(i) $f(\mathcal{S})\geq 0$. The condition in~\eqref{eq:app_Frank} holds because right-hand side is non-negative, while
the left-hand side is less than or equal to~0.

(ii) $f(\mathcal{S})< 0$. This case occurs only when
$\mathcal{S}$ contains the sink vertex $T$ but not source vertex~$S$.
The condition in~\eqref{eq:app_Frank} can be re-written as
$ M \leq c(\Delta^- {\mathcal{S}})$.
The value $c(\Delta^- {\mathcal{S}})$ is the cut capacity of $(\bar{\mathcal{S}}, \mathcal{S})$, which is at least $M$ by our assumption that the minimal cut capacity is $M$.

It can be easily checked that $f(\emptyset) = f(\mathcal{V})=0$. Thus all the conditions in Theorem~\ref{thm:Frank} are satisfied. By Theorem~\ref{thm:Frank}, there exists a  feasible $f$-submodular flow, say $\phi$. We next verify that $\phi$ is indeed an $(S,T)$-flow in $H$. By the definition of submodular flow,
we have $\partial \phi(\{S\}) \leq M$,  $\partial \phi(\{T\}) \leq -M$, and $\partial \phi(\{v\}) \leq 0$ for vertex $v$ not equal to $S$ or $T$. Using the fact that $\sum_{v\in\mathcal{V}} \partial \phi(\{v\})=0$, which holds in general for any real-valued function $\phi$ on $\mathcal{E}$, we obtain
\begin{align*}
M \leq -\partial \phi(\{T\})
&= \partial \phi(\{S\})  + \sum_{v\not\in\{S, T\}} \partial \phi(\{v\}) \\
&\leq  \partial \phi(\{S\}) \leq M.
\end{align*}
Since equalities hold throughout the above chain of inequalities, we have $\partial \phi(\{v\}) = 0$ for all vertices $v$ other than $S$ and~$T$, i.e., $\phi$ satisfies the flow conservation property. Finally, we have $\partial \phi(\{S\}) =M = -\phi(\{T\})$. This is the same as saying that the flow value is equal to~$M$.

\section{The MSCR Point under Heterogeneous Traffic}
\label{app:asym}

Homogeneous traffic is assumed in the main text of this paper.  We show in this appendix that the assumption of homogeneous traffic is not essential at the minimum-storage point, i.e., the repair bandwidth cannot be decreased even if traffic is heterogeneous.

Let $\alpha = B/k$. Let the average number of packets per link in the first (resp. second) phase be $\bar{\beta}_1$ (resp. $\bar{\beta}_2$). The total number of packets transmitted in the first (resp. second) phase is thus $r d \bar{\beta}_1$ (resp. $r (r-1) \bar{\beta}_2$). In this heterogeneous traffic mode, it is only required that the traffic in the first (resp. second) phase of all repair processes are identical. It contains the homogeneous traffic model as a special case if each newcomer downloads $\bar{\beta}_1$ packets per link in the first phase and $\bar{\beta}_2$ packets per link in the second phase.

\begin{theorem}
If $\alpha= B/k$, then the average repair bandwidth per newcomer under the heterogeneous traffic model is lower bounded by $$B\frac{d+r-1}{k(d+r-k)}.$$
\end{theorem}

\begin{proof}
Consider the scenario where nodes 1 to $r$ fail at stage 1. Suppose that each newcomer connects to nodes $r+1$, $r+2,\ldots, r+d$ during the repair process. 
For $i=r+1,r+2,\ldots, r+d$ and $j=1,2,\ldots, r$, let the capacity of the link from surviving node $i$ to newcomer $j$ be $\beta_1(i,j)$. Let the capacity of the link from $\In_{j_1}$ to $\Mid_{j_2}$, for $j_1\neq j_2$, be $\beta_2(j_1,j_2)$. The average link capacities in the first and second phase can be written, respectively, as
\begin{align*}
 \bar{\beta}_1 &= \frac{1}{dr} \sum_{i=r+1}^{r+d} \sum_{j=1}^r \beta_1(i,j), \\
 \bar{\beta}_2 &= \frac{1}{r(r-1)} \sum_{j_2=1}^r  \sum_{j_1=1 \atop j_1 \neq j_2}^r \beta_2(j_1,j_2).
\end{align*}
The repair bandwidth per newcomer is thus
\[
\frac{1}{r} \sum_{i=r+1}^{r+d} \sum_{j=1}^r \beta_1(i,j) + \frac{1}{r} \sum_{j_2=1}^r  \sum_{j_1=1 \atop j_1 \neq j_2}^r \beta_2(j_1,j_2) = d \bar{\beta}_1 + (r-1) \bar{\beta}_2.
\]

\begin{figure}
\centering
\includegraphics[width=2.5in]{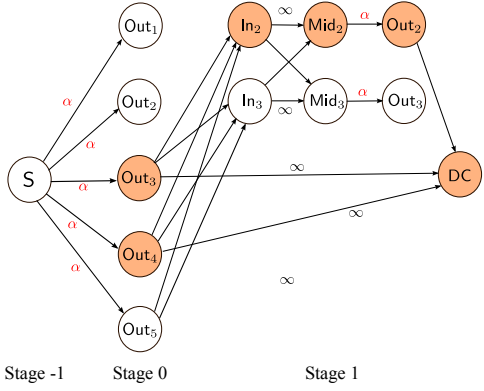}
\caption{A data collector connects to one storage node among the first $r$ nodes and $k-1$ storage nodes among the nodes $r+1$ to $r+d$.}
\label{fig:nonuniform1}
\end{figure}

Consider the set of a data collectors which connects to one of the $r$ newcomers and $k-1$ nodes among nodes $r+1$ to $r+d$.
For a data collector \DC\ connecting to node $j$, for some $j\in\{1,2,\ldots, r\}$, and nodes $i_1, i_2,\ldots, i_{k-1} \in \{r+1,r+2,\ldots, r+d\}$, consider the cut $(\mathcal{W}^c, \mathcal{W})$, with
\[
\mathcal{W} = \{\DC, \In_j, \Mid_j, \Out_j, \Out_{i_1}, \Out_{i_2}, \ldots, \Out_{i_{k-1}}\}.
\]
This cut yields an upper bound on the file size $B$.
An example is given in Fig.~\ref{fig:nonuniform1}, with $\mathcal{W}$ drawn in shaded color.

There are $r \binom{d}{k-1}$ distinct data collectors in this set.  If we sum over all $r \binom{d}{k-1}$ corresponding inequalities, we obtain
\begin{align*}
 r \binom{d}{k-1} B &\leq r \binom{d}{k-1} (k-1) \alpha \\
 & \quad + \binom{d-1}{k-1} \sum_{i=r+1}^{r+d} \sum_{j=1}^r \beta_1(i,j) \\
 & \quad + \binom{d}{k-1} \sum_{j_2=1}^r  \sum_{j_1=1 \atop j_1 \neq j_2}^r \beta_2(j_1,j_2).
\end{align*}
The first term on the right-hand side  comes from the fact that each of the $r\binom{d}{k-1}$ inequalities contributes $(k-1)\alpha$. For the second term, we note that there are are $\binom{d-1}{k-1}$ choices for the ``out'' nodes to be included in $\mathcal{W}$. Hence for each $i$, $j$, the term $\beta_1(i,j)$ is multiplied by $\binom{d-1}{k-1}$. By similar argument we can obtain the third term.

After dividing both sides by $r\binom{d}{k-1}$, we obtain
\begin{equation}
 B \leq (k-1)\alpha + (d-k+1) \bar{\beta}_1+ (r-1) \bar{\beta}_2.
 \label{eq:nonuniform1}
\end{equation}

In the rest of the proof we distinguish two cases.

Case 1:  $k\geq r$. Consider the class of data collectors which download from nodes 1 to $r$, and $k-r$ nodes among nodes $r+1$ to $r+d$. For a data collector \DC\ in this class, say connecting to nodes 1 to $r$, and $i_1$, $i_2, \ldots, i_{k-r}\in\{r+1,r+2,\ldots, r+d\}$, we have an upper bound on $B$ from  the cut $(\mathcal{W}^c,\mathcal{W})$ with $\mathcal{W}$ specified by
\[
 \mathcal{W}=\{ \DC, \Out_{i_1}, \Out_{i_2},\ldots, \Out_{i_{k-r}}\} \cup \bigcup_{j=1}^r \{\In_j, \Mid_j, \Out_j\}.
\]
If we sum over the  $\binom{d}{k-r}$ inequalities arising from these cuts, we get
\begin{align*}
\binom{d}{k-r} B & \leq \binom{d}{k-r}(k-r) \alpha \\
& \quad + \binom{d-1}{k-r}(k-r) \sum_{i=r+1}^{r+d} \sum_{j=1}^r \beta_1(i,j).
\end{align*}
Upon dividing both sides by $\binom{d}{k-r}$, we obtain
\begin{equation}
 B \leq (k-r) \alpha + (d-k+r)r \bar{\beta}_1.
 \label{eq:nonuniform2}
\end{equation}
With $\alpha = B/k$, we infer from \eqref{eq:nonuniform1} and \eqref{eq:nonuniform2} that
\begin{equation}
d\bar{\beta}_1 +(r-1) \bar{\beta}_2 \geq \frac{B(d+r-1)}{k(d+r-k)}.
\label{eq:nonuniformMSCR}
\end{equation}

Case 2:  $k< r$. Consider the class of data collectors who connects to $k$ nodes among nodes 1 to $r$. To a data collector \DC\ connecting to $j_1, j_2,\ldots, j_k \in\{1,2,\ldots, r\}$, we associate it with the cut $(\mathcal{W}^c, \mathcal{W})$ with $\mathcal{W}$ given by
\[
\mathcal{W} = \{\DC\} \cup \bigcup_{\ell=1}^k \{\In_{i_\ell}, \Mid_{i_\ell}, \Out_{i_\ell} \}.
\]
The sum of the $\binom{r}{k}$ resulting upper bounds on $B$ is
\begin{align*}
 \binom{r}{k} B & \leq \binom{r-1}{k-1} \sum_{i=1}^r \sum_{j=r+1}^{r+d} \beta_1(i,j)  \\
 & \quad +  \binom{r-1}{k-1} \sum_{j_2=1}^r  \sum_{j_1=1 \atop j_1 \neq j_2}^r \beta_2(j_1,j_2).
\end{align*}
After dividing both sides by $\binom{r}{k}$, we get
\begin{equation}
 B \leq kd \bar{\beta}_1 + k(r-k) \bar{\beta}_2.
 \label{eq:nonuniform3}
\end{equation}
From \eqref{eq:nonuniform1} and \eqref{eq:nonuniform3}, we can deduce \eqref{eq:nonuniformMSCR}.
\end{proof}

\section{Proof of Theorem~\ref{thm:cornerpoint}}
\label{app:cornerpoint}

We first prove two lemmas. The first one is about the lower envelope of a collection of straight lines.

\begin{lemma}
Let  $y = m_j x + b_j$ for $j=1, 2, \ldots, N$,  be $N$ straight lines in the $x$-$y$ plane, satisfying the following conditions:

(a)
The slopes are negative with decreasing magnitudes, i.e.,
\begin{gather*}
-m_1 > -m_2 > \cdots > -m_N > 0.
\end{gather*}

(b) For $j=2,3,\ldots, N$,  the $x$-coordinates of the intersection point of $y=m_jx+b_j$ and $y = m_{j-1}x+b_{j-1}$, denoted by $x_j$, are strictly increasing, i.e., $ x_2 < x_3 < \cdots< x_{N}$.

Then we have
\begin{align*}
&\max_{1\leq j\leq N} \{m_jx+b_j\}  \\
&= \begin{cases}
m_1 x + b_1 & \text{for }  x < x_2 ,\\
m_2 x + b_2 & \text{for }  x_2 \leq x < x_3, \\
\vdots & \vdots \\
m_{N-1} x + b_{N-1} & \text{for }  x_{N-1} \leq x < x_N ,\\
m_N x + b_N & \text{for } x \geq x_{N}.
\end{cases}
\end{align*}
\label{lemma0}
\end{lemma}

\begin{proof}
For $i=2,3,\ldots, N$,
since the slope of $L_{i-1}$ is more negative then the slope of $L_{i}$, we get
\[
\begin{cases}
m_{i-1} x + b_{i-1} > m_{i} x + b_{i} & \text{for } x < x_i, \\
m_{i-1} x + b_{i-1} = m_{i} x + b_{i} & \text{for } x = x_i, \\
m_{i-1} x + b_{i-1} < m_{i} x + b_{i} & \text{for } x > x_i .
\end{cases}
\]
For $x$ between 0 and $x_2$, we have $x<x_i$ for all $i$. Hence
$$ m_1x+b_1 > m_2 x+b_2 > \cdots > m_N x+ b_N.$$
Therefore,
$\max_{j} \{m_j x+b_j\} = m_1 x + b_1$, for $x < x_2$.

Consider  $x$ in the interval $[x_{i-1},x_i)$, for some $i=3,4,\ldots, N-1$. Since $x \geq x_{i-1} > x_{i-2} > \cdots > x_2$, we get
\[m_i x+ b_i \geq m_{i-1} x + b_{i-1} > \cdots > m_1 x + b_1.\]
On the other hand, since $x<x_i<\cdots<x_{N}$, we get
\[m_i x+b_i > m_{i+1} x+b_i > \cdots >m_{N} x +b_N.\]
 Therefore
$
\max_{j=1,\ldots, N} \{m_jx+b_j\} = m_i x + b_i
$
for $x_{i-1} \leq x < x_i$.
The proof of the last case $x \geq x_{N}$ is similar and is omitted.
\end{proof}

The second lemma is a special case of duality in linear programming. It gives a sufficient condition for checking the optimality of a given point in the feasible region. The short proof is given below for the sake of completeness.

\begin{lemma}
Consider a linear programming with objective function $c_1 x + c_2 y$, where $x$ and $y$ are variables and $c_1$ and $c_2$ are constants, subject to constraints $a_{i1} x + a_{i2} y \geq b_i$, for $i=1,2,\ldots ,N$, and $x , y \geq 0$. We will only consider the case where $c_1$, $c_2$, $a_{i1}$ and $a_{i2}$, for $i=1,2,\ldots, N$, are non-negative. If $(\bar{x}, \bar{y})$ is a point which

\noindent \ (a) satisfies all constraints,

\noindent \ (b) attains equality in two particular constraints whose slopes are distinct, and the slope of the objective function is between these two slopes,

\noindent then $(\bar{x}, \bar{y})$ is the optimal solution to the linear programming problem.
\label{lemma_dual}
\end{lemma}

\begin{proof}
It suffices to show that the value of the objective function cannot be smaller than $c_1 \bar{x}+c_2 \bar{y}$ without violating any constraints. By re-indexing the constraints, suppose without loss of generality that $(\bar{x},\bar{y})$ satisfies the constraints $a_{i1} x + a_{i2} y \geq b_i$, for $i=1,2$, with equality. Suppose that the magnitude of the slope of the first constraint is strictly larger than the magnitude of the slope of the second constraint,
\begin{equation}
a_{11}/a_{12} >  a_{21}/a_{22},
\label{eq:assumption0}
\end{equation}
and $c_1/c_2$ is sandwiched between them,
\begin{equation}
a_{11}/a_{12} \geq c_1/c_2 \geq  a_{21}/a_{22}.
\label{eq:assumption}
\end{equation}

Let
$ \mathbf{A} := \begin{bmatrix}
 a_{11} & a_{12} \\
 a_{21} & a_{22}
 \end{bmatrix}$.
By the assumption in (b), we have
\[
\mathbf{A} \begin{bmatrix}x \\ y \end{bmatrix} \geq  \begin{bmatrix} b_1 \\ b_2 \end{bmatrix} = \mathbf{A} \begin{bmatrix} \bar{x} \\ \bar{y} \end{bmatrix}
\]
for any feasible solution $(x,y)$.

The determinant of $\mathbf{A}$ is positive by~\eqref{eq:assumption0}. Thus, $\mathbf{A}$ is invertible. The values of $p$ and $q$  defined by
\begin{align*}
\begin{bmatrix} p & q \end{bmatrix}  &:= \begin{bmatrix} c_1 & c_2 \end{bmatrix} \mathbf{A}^{-1} = \frac{\begin{bmatrix}c_1 a_{22} - c_2 a_{21}&  c_2 a_{11} - c_1 a_{12} \end{bmatrix}}{a_{11}a_{22}-a_{12}a_{21}}
\end{align*}
are non-negative by~\eqref{eq:assumption}.
For any feasible solution $(x,y)$,
\begin{align*}
c_1 x + c_2 y = \begin{bmatrix} p & q \end{bmatrix} \mathbf{A} \begin{bmatrix} x \\ y \end{bmatrix} \geq \begin{bmatrix} p & q \end{bmatrix} \mathbf{A} \begin{bmatrix} \bar{x} \\ \bar{y} \end{bmatrix}
= c_1 \bar{x} + c_2 \bar{y}.
\end{align*}
This proves that the optimal value is $c_1 \bar{x} + c_2 \bar{y}$.
\end{proof}

\medskip

We divide the proof of Theorem~\ref{thm:cornerpoint} into several propositions. We need a few more notations.

\noindent {\bf Definitions:} For $j=1,2,\ldots, k$,
let $g_j(\talpha)$ be the $\tbeta_2$-coordinate of $P_{j}(\talpha)$, i.e., $$g_j(\talpha):= \frac{1-(k-j)\talpha}{j(2d-2k+r+j)},$$
and let
 \[
  \hat{\beta}_2(\talpha) := \max_{1\leq j \leq k} g_j(\talpha).
 \]

\begin{proposition} For $j=1,2,\ldots,k $, let $\talpha_j$ be defined as in \eqref{eq:tilde_alpha1}, and for $\ell=0,1,\ldots, \lfloor k/r \rfloor$, let $\talpha_\ell'$ be defined as in \eqref{eq:tilde_alpha2}.

\begin{enumerate}
\item For $j=2,3,\ldots, k$, we have
$g_j(\tilde{\alpha}_j) = g_{j-1}(\tilde{\alpha}_j)$.

\item
\[
 \frac{1}{k} =\talpha_1 < \tilde{\alpha}_2 < \tilde{\alpha}_3 < \cdots <\tilde{\alpha}_{k} = \frac{2d+r-1}{k(2d+r-k)}.
\]

\item  $\hat{\beta}_2(\talpha)$ is a  piece-wise linear function of~$\talpha$,
    \[
   \hat{\beta}_2(\talpha) = \begin{cases}
    g_1(\talpha) & \text{for } \tilde{\alpha}_1 \leq \talpha < \tilde{\alpha}_2, \\
    g_2(\talpha) & \text{for } \tilde{\alpha}_2 \leq \talpha < \tilde{\alpha}_3 ,\\
    \vdots & \vdots \\
    g_{k-1}(\talpha) & \text{for } \tilde{\alpha}_{k-1} \leq \talpha < \tilde{\alpha}_k, \\
    g_k(\talpha) & \text{for } \talpha \geq \talpha_k.
        \end{cases}
    \]

\item For $j=1,2,\ldots, k-1$, when the parameter $\talpha$ is in the range $\talpha_j \leq \talpha \leq \talpha_{j+1}$, we have
\[
 P_j(\talpha) \succ P_{j+1}(\talpha) \succ P_{j+2}(\talpha) \succ \cdots \succ P_k(\talpha),
\]
and
\[
 P_j(\talpha) \succ P_{j-1}(\talpha) \succ P_{j-2}(\talpha) \succ \cdots \succ P_1(\talpha).
\]
When $\talpha \geq \talpha_{k}$, we have
\[
 P_k(\talpha) \succ P_{k-1}(\talpha) \succ P_{k-2}(\talpha) \succ \cdots \succ P_1(\talpha).
\]

\item  For $\ell=0,1,\ldots \lfloor k/r \rfloor-1$, we have
$$\talpha_{\ell r} < \tilde{\alpha}_\ell' < \talpha_{(\ell+1) r}.$$
\end{enumerate}
\label{lemmaC}
\end{proposition}

\begin{proof}
(1)  It follows from the fact that $\talpha_j$ is defined as the value such that $P_j(\talpha_j) = P_{j-1}(\talpha_j)$.

(2) We compare two consecutive terms in the sequence $(\tilde{\alpha}_i)_{i=1}^{k}$.
For $j=2,3,\ldots, k$, we have
\begin{align*}
&\talpha_{j} > \talpha_{j-1} \\
\Leftrightarrow & \frac{d-k+j+\frac{r-1}{2}}{k(d-k+j+\frac{r-1}{2}) - \frac{j(j-1)}{2}} > \\
& \quad \frac{d-k+j-1+\frac{r-1}{2}}{k(d-k+j-1+\frac{r-1}{2}) - \frac{(j-1)(j-2)}{2}}\\
\Leftrightarrow & \frac{1}{k-\frac{j(j-1)}{2d-2k+2j+r-1}} >
\frac{1}{k-\frac{(j-1)(j-2)}{2d-2k+2j-2+r-1}} \\
\Leftrightarrow & \frac{j(j-1)}{2d-2k+2j+r-1} > \frac{(j-1)(j-2)}{2d-2k+2j-2+r-1}.
\end{align*}
Since the value of  $j$ is strictly larger than 1 in the above  inequalities, we can cancel the factor $j-1$. After some more algebraic manipulation, we obtain
\begin{align*}
\talpha_{j} > \talpha_{j-1} \Leftrightarrow    4d-4k+2j+2r-2 &> 0.
\end{align*}
The last inequality holds because $d\geq k$, $j\geq 2$ and $r\geq 1$.

(3)
We will apply Lemma~\ref{lemma0} to prove the third part. We have already verify part (b) of Lemma~\ref{lemma0}.
For the condition in part (a) of Lemma~\ref{lemma0}, we check that the magnitude of the slope of the straight line $y=g_j(\talpha)$ in the $\talpha$-$y$ plane is
$$\frac{k-j}{j(2d+2k+r+j)}.$$
When $j$ increases, the numerator decreases and the denominator increases. Hence the magnitude of the slope is a decreasing function of~$j$.

(4) It follows from the proof of Lemma~\ref{lemma0}.

(5) We can prove the asserted inequalities by straightforward calculation. The details are omitted.
\end{proof}

\begin{proposition}\

\begin{enumerate}
\item For $j=2,3,\ldots, k-1$,  the point $P_j(\talpha)$ is a feasible solution to the linear program in~\eqref{eq:LP_objective} when the parameter $\talpha$ is in the range $\talpha_j \leq \talpha \leq \talpha_{j+1}$.

\item For $\ell=0,1,\ldots, k$, the point $Q_\ell$ is a feasible solution to the   linear program in~\eqref{eq:LP_objective} if $\talpha = \talpha_\ell'$.
\end{enumerate}
\label{prop:feasible}
\end{proposition}

\begin{proof}
(1)
Consider the point $P_j(\talpha)$ for some $j=2,3,\ldots, k-1$. Let $\talpha$ be a real number between $\talpha_j$ and $\talpha_{j+1}$. For $\talpha$ in this range, we have $P_i(\talpha) \prec P_j(\talpha)$ for all $i\in\{1,2,\ldots, k\}\setminus\{j\}$, by the third part in Prop.~\ref{lemmaC}. We will give a geometric proof that $P_j(\talpha)$  satisfies the inequalities in \eqref{eq:LP1} and \eqref{eq:LP2} for $s=1,2,\ldots, k$.
We use the property that the slope of the linear constraints in \eqref{eq:LP1} and \eqref{eq:LP2} are either negative or infinite.
Recall that for any $i=1,2,\ldots, k$, $P_i(\talpha)$ is the intersection point of $L_i(\talpha)$ and $L_i'(\talpha)$. As the slope of $L_i(\talpha)$ and $L_i'(\talpha)$ are either infinite or negative, any point $P$ in the $\tbeta_1$-$\tbeta_2$ plane that Pareto-dominates $P_i(\talpha)$ satisfies the inequality \eqref{eq:LP1} and \eqref{eq:LP2} for $s=i$. Since $P_i(\talpha) \prec P_j(\talpha)$ for $\talpha_j \leq \talpha \leq \talpha_{j+1}$, we conclude that $P_j(\talpha)$ satisfies all constraints in the linear program, and is thus feasible.

(2) From the last part Prop.~\ref{lemmaC}, we have $\talpha_{\ell r} < \tilde{\alpha}_\ell' < \talpha_{(\ell+1) r}$.  By the third part of Prop.~\ref{lemmaC}, we have
\[\hat{\beta}_2(\talpha_\ell') > g_j(\talpha_\ell')
\]
and thus
\[
  P_j(\talpha_\ell') \prec Q_\ell
\]
for
$$j\in \{1,2,\ldots, k\} \setminus \{ \ell r, \ell r+1,\ldots, \min\{k,(\ell+1)r\} \}.
$$
This proves that $Q_\ell$ satisfies the constraint \eqref{eq:LP1} and \eqref{eq:LP2} for $s$ in$ \{1,2,\ldots, k\} \setminus \{ \ell r, \ell r+1,\ldots, \min\{k,(\ell+1)r\} \}$.

We give a geometric proof for the remaining constraints. (See e.g. Fig.~\ref{fig:t2}.)
For $\talpha = \talpha_\ell'$, the point $P_{\ell r}(\talpha_\ell')$ is vertically below $Q_\ell$ in the $\tbeta_1$-$\tbeta_2$ plane.
Consider an integer $$j \in \{\ell r, \ell r+1, \ldots, \min\{k,(\ell+1)r\} \}.$$ The
point $P_j(\talpha_\ell')$ is on the line $\tbeta_1 = 2\tbeta_2$ and is to the right of  $Q_\ell$ and $P_{\ell r}(\talpha_\ell')$. Since the slope of $L_j'(\talpha_\ell')$ is negative and has magnitude less than $\mu(j)$, the line $L_j'(\talpha_\ell')$ intersects the vertical line segment between  $Q_\ell$ and $P_{\ell r}(\talpha_\ell')$. Therefore $Q_\ell$ is lying above the line $L_j'(\talpha_\ell')$. Also, by definition, $Q_\ell$ is lying on the line $L_j(\talpha_\ell')$. This proves that $Q_\ell$ satisfies the constraints in \eqref{eq:LP1} and \eqref{eq:LP2} for $s = j$.
\end{proof}

\begin{proposition}
For $j=1,2,\ldots,k-1$, if
\begin{equation}
\begin{vmatrix}
   j(d-k) + (j^2+\Psi_{j,r})/2 & jr-\Psi_{j,r} \\
   d & r-1
 \end{vmatrix} \geq 0
  \label{slope_assumption1}
\end{equation}
then
\[
\gamma_{\mathrm{LP}}^*( \talpha ) = \frac{(2d+r-1)(1-(k-j)\talpha)}{j(2d-2k+r+j)}
\]
for $\tilde{\alpha}_{j}\leq \talpha < \tilde{\alpha}_{j+1}$.

Also, we have
\[
 \gamma_{\mathrm{LP}}^*(\talpha) = \frac{2d+r-1}{k(2d+r-k)} = \tgamma_{\mathrm{MBCR}}
\]
for $\talpha \geq \tilde{\alpha}_{k} = \talpha_{\mathrm{MBCR}}$.
\label{lemmaD}
\end{proposition}

\begin{proof} Consider $\talpha$ in the interval $[\tilde{\alpha}_{j}, \tilde{\alpha}_{j+1} )$, for some $j\in\{1,2,\ldots, k-1\}$. Suppose that the condition in~\eqref{slope_assumption1} is satisfied.
We want to show that $P_{j}(\talpha)$
is the optimal solution to the linear programming problem in~\eqref{eq:LP_objective}. We have proved in Prop.~\ref{prop:feasible} that $P_{j}(\talpha)$ is a feasible solution. By Lemma~\ref{lemma_dual}, it remains to show that $P_{j}(\talpha)$ satisfies two constraints with equality, and the slope of the objective function is  between the slopes of these two constraints.

Since the condition in~\eqref{slope_assumption1} is satisfied,  the magnitude of the slope of $L_j(\talpha)$, namely
\[
\mu(j)= \frac{j(d-k)+(j^2+\Psi_{j,r})/2}{jr-\Psi_{j,r}},
\]
is larger than or equal to $d/(r-1)$. On the other hand, the magnitude of the slope of $L_j'(\talpha)$,
\[
 \frac{d-k+(j+1)/2}{r-1}
\]
is strictly less than $d/(r-1)$. By Lemma~\ref{lemma_dual},  $P_{j}(\talpha)$ is the optimal solution to the linear program in~\eqref{eq:LP_objective}. Thus
$$\gamma_{\mathrm{LP}}^*(\talpha)=(2d+r-1) g_j(\talpha) = (2d+r-1)\frac{1-(k-j)\talpha}{j(2d-2k+r+j)}$$
for $\talpha_{j} \leq \talpha < \talpha_{j+1}$.

For the second part of the proposition, $P_k(\talpha_k)$ is a feasible solution for $\talpha \geq \talpha_{k}$. From part (\ref{lemma:itemF}) of Lemma~\ref{lemma:P_j}, we know that $\mu_k > d/(r-1)$. On the other hand, the slope of $L_k'(\talpha)$ is strictly less than $d/(r-1)$. By Lemma~\ref{lemma_dual}, $P_k(\talpha_k)$ is the optimal solution of the linear program. Therefore
$$\gamma_{\mathrm{LP}}^*(\talpha) = g_k(\talpha_k) = \frac{2d+r-1}{k(2d+r-k)} = \tgamma_{\mathrm{MBCR}}$$
 for $\talpha \geq \alpha_{k}$.
\end{proof}

We  now cover the remaining cases which are not covered by Prop.~\ref{lemmaD}. Suppose that there is an integer $i$ between 1 and $k$ such that
$\mu(i) < \frac{d}{r-1}$.
Let $\ell$ be the quotient when $i$ is divided by $r$.

We claim that $\mu(j)$ is a concave function of $j$ for $j$ between $1+\ell r$ and  $(r-1)+\ell r$.
Consider integer $j$  in the form $j=\ell r +R$, for $0\leq R < r$.  Then, $\Psi_{j,r} = \ell r^2+R^2$.
\begin{align*}
& \phantom{=} \mu(\ell r +R) \\
& = \frac{(R+\ell r)(d-k)+((R+\ell r)^2+\ell r^2+R^2)/2}{(R+\ell r)r-\ell r^2-R^2} \\
&= \frac{2R^2+2R(\ell r+(d-k))+\ell r(\ell r+r+2(d-k))}{2R(r-R)} \\
&= \frac{R+r+\ell r+(d-k)}{r-R}+\frac{\ell r(\ell r+r+2(d-k))}{2R(r-R)}\\
&= -1+\frac{r+\ell r+(d-k)}{r-R}+\frac{\ell r(\ell r+r+2(d-k))}{2R(r-R)}.
\end{align*}
Each term in the above line is a concave function of $R$. Therefore the sum of them is also concave. This completes the proof of the claim.

By the above claim, we can find an index $j_0$ such that
\begin{equation}
   \mu(j_0)   = \min_{\ell r < j < (\ell +1)r}  \mu(j).
 \label{slope_assumption2}
\end{equation}


\begin{proposition} Suppose that $\mu(i) < d/(r-1)$ for some $i$, and let $j_0$ be defined as in~\eqref{slope_assumption2}. Let $i_2$ be the smallest integer larger than or equal to $j_0$  such that $\mu(i_2) < d/(r-1)$, and let $i_1$ be the largest integer smaller than or equal to $j_0$ such that $\mu(i_1) < d/(r-1)$.

\begin{enumerate}
\item  The integers $i_1$ and $i_2$ are well-defined, and they satisfy $\ell r< i_1 \leq j_0\leq i_2 < \min\{ k,r (\ell +1)\}$.
  \label{CaseA}

\item \label{CaseC}
\[
\gamma_{\mathrm{LP}}^*(\tilde{\alpha}_\ell')=\frac{1}{D_\ell'} (d+r-1).
\]
In particular, we have
\[
 \gamma_{\mathrm{LP}}^*(1/k) =  \gamma_{\mathrm{LP}}^*(\talpha_0') =  \frac{d+r-1}{k(d+r-k)}.
\]

\item \label{CaseD} Let \begin{align} c_1(j) &:= j(d-k)+(j^2+\Psi_{j,r})/2 ,\\
c_2(j) &:= jr-\Psi_{j,r},
 \end{align}
and
$\mathbf{A}$ be the matrix
\[
\mathbf{A} := \begin{bmatrix}
c_1(j_1+1) & c_2(j_1+1) \\
c_1(j_1) & c_2(j_1)
\end{bmatrix}.
\]
For $\talpha$ between $\tilde{\alpha}_\ell'$ and $\tilde{\alpha}_{j_1+1}$, we have
\[
 \gamma_{\mathrm{LP}}^*(\talpha) = \begin{bmatrix}d & r-1 \end{bmatrix}
 \mathbf{A}^{-1} \begin{bmatrix}
 1-(k-j_1-1)\talpha) \\1-(k-j_1)\talpha
 \end{bmatrix}.
\]

\item \label{CaseE}
For $\alpha$ between $\tilde{\alpha}_{j_2}$ and $\tilde{\alpha}_\ell'$, we have
\[
 \gamma_{\mathrm{LP}}^*(\alpha) = \begin{bmatrix}d & r-1 \end{bmatrix}
 \mathbf{B}^{-1} \begin{bmatrix}
 1-(k-j_2+1)\talpha) \\ 1-(k-j_2)\talpha
 \end{bmatrix},
\]
where
\[
\mathbf{B} := \begin{bmatrix}
c_1(j_2-1) & c_2(j_2-1) \\
c_1(j_2) & c_2(j_2)
\end{bmatrix}.
\]
\end{enumerate}
\label{lemmaE}
\end{proposition}

\begin{proof}
\ref{CaseA})
The first part of the lemma follows from the property that $\mu(j) = \infty$ if $j$ is an integral multiple of $r$ (Lemma~\ref{lemma:P_j} part~\ref{lemma:itemC}.)

\ref{CaseC})
We want to show that
$Q_\ell = (1/D_\ell', 1/D_\ell')$
is the optimal solution to the linear program when $\talpha = \talpha_\ell'$. We have shown in Prop.~\ref{prop:feasible} that  $Q_\ell$ is a feasible solution.
Because the slope of $L_{\ell r}(\tilde{\alpha}_{\ell}')$ is infinite and the magnitude of the slope of $L_{j_0}(\tilde{\alpha}_{\ell}')$ is less than $d/(r-1)$, by Lemma~\ref{lemma_dual}, $Q_\ell$ is the optimal to the linear program.

From part 4 of Lemma~\ref{lemma:P_j}, we have $\mu(1) < d/(r-1)$. Therefore, $Q_0'$ is the optimal solution to the linear program when $\talpha = \talpha_0' = 1/k$, and we get
\[
\gamma_{\mathrm{LP}}^*(1/k) = \gamma_{\mathrm{LP}}^*(\tilde{\alpha}_0')= \frac{1}{D_0'}(d+r-1) =  \frac{d+r-1}{k(d+r-k)}.
\]

\ref{CaseD}) Consider $\talpha$ which is within the range $\tilde{\alpha}_\ell' < \talpha < \tilde{\alpha}_{i_2}$.
Let $P_{opt}(\talpha) = (\tbeta_{1,opt}(\talpha), \tbeta_{2,opt}(\talpha))$ be the intersection point of $L_{j_1}(\talpha)$ and $L_{j_1+1}(\talpha)$, i.e.,
\[
 \begin{bmatrix}
 \tbeta_{1,opt} \\ \tbeta_{2,opt}
 \end{bmatrix} =
 \mathbf{A}^{-1}
 \begin{bmatrix}
 1-(k-i_2-1)\talpha \\ 1-(k-i_2)\talpha
 \end{bmatrix}.
\]
We have
$$ (\tbeta_{1,opt}(\talpha_\ell'), \tbeta_{2,opt}(\talpha_\ell'))
= Q_\ell, \text{ and }
$$
$$ (\tbeta_{1,opt}(\talpha_{i_2+1}), \tbeta_{2,opt}(\talpha_{i_2+1}))
= P_{i_2+1}(\talpha_{i_2+1}).
$$
For $\talpha$ between $\talpha_\ell'$ to $\talpha_{i_2+1}$, the point $(\tbeta_{1,opt}(\talpha), \tbeta_{2,opt}(\talpha))$  is a convex combination of $Q_\ell$ to $P_{i_2+1}(\talpha_{i_2+1})$. Therefore,
$(\tbeta_{1,opt}(\talpha), \tbeta_{2,opt}(\talpha))$ is a feasible solution to the linear program with the corresponding parameter~$\talpha$. (See the remark after \eqref{eq:C_LP}.) The slope of $L_{j_1+1}(\talpha)$ has magnitude  larger than or equal to $d/(r-1)$, while  the slope of $L_{j_1}(\talpha)$ has magnitude strictly less than~$d/(r-1)$. Thus by Lemma~\ref{lemma_dual}, $P_{opt}(\talpha)$ is the optimal solution to the linear program.

\ref{CaseE}) The proof is analogous to the previous part and is omitted.
\end{proof}

For $j=1,2,\ldots, k$, let
\[
\xi_j = \begin{cases}
\frac{1}{D_j} (d-k+j+(r-1)/2) & \text{if } \mu(j) \geq d/(r-1)\\
\frac{1}{D_{\lfloor j/r \rfloor}'}( d-k+r(\lfloor j/r\rfloor+1)) & \text{if } \mu(j) < d/(r-1),
\end{cases}
\]
be the $\talpha$-coordinates of the operating points in Theorem~\ref{thm:cornerpoint}.
Divide the interval $[\talpha_{\mathrm{MSCR}},\talpha_{\mathrm{MBCR}})$ into subintervals
$$[\xi_1,\xi_2), \ [\xi_2,\xi_3),\ldots, [\xi_{k-1}, \xi_k),$$
with $\xi_1 = \talpha_{\mathrm{MSCR}}$ and $\xi_k = \talpha_{\mathrm{MBCR}}$.
From Lemma~\ref{lemmaD} and Lemma~\ref{lemmaE}, the function $\gamma_{\mathrm{LP}}^*(\talpha)$ is an affine function of $\talpha$ in each subinterval. Consequently, the corner points of the graph $$\{(\gamma_{\mathrm{LP}}^*(\talpha), \talpha):\, \talpha_{\mathrm{MSCR}} \leq \talpha < \infty \}$$
are precisely the operating points defined in Theorem~\ref{thm:cornerpoint}.

\section{Proof of Lemma~\ref{lemma:MBCR}}
\label{app:lemmaMBCR}
We first show that it is sufficient to verify that the condition~\eqref{eq:AAA} in Lemma~\ref{lemma:MBCR} holds for subsets $\mathcal{S}$ in the form of
\begin{equation}
\mathcal{S}= \Big(\bigcup_{i\in\mathcal{A}} \{\In_{i}, \Mid_{i}, \Out_i\}\Big) \cup \Big(\bigcup_{j\in\mathcal{B}} \{v_{j}, \Out_j\}\Big),
\label{eq:S}
\end{equation}
where $\mathcal{A}$ is a subset of $\{1,2,\ldots, r\}$ and $\mathcal{B}$ is a subset of $\{r+1,r+2,\ldots, n\}$.

An example of a subset $\mathcal{S}$ in the form of~\eqref{eq:S} is illustrated in Fig.~\ref{fig:auxiliary}. For notational convenience, we let
$$
\kappa(\mathcal{S}) := \lb(\Delta^+ \mathcal{S}) - \ub(\Delta^- \mathcal{S}),
$$
for subset $\mathcal{S}$ of the vertices in the auxiliary graph.

(a) Suppose for some $j\in\{r+1,r+2,\ldots, n\}$, $\mathcal{S}$ contains $v_j$ but does not contain $\Out_j$. Then the directed edge $e = (v_j, \Out_j)$ is in $\Delta^+ \mathcal{S}$, and makes a contribution of $h_j$ to the term $\lb(\Delta^+ \mathcal{S})$. But the inequality
 $\kappa(\mathcal{S}) \leq \sigma(\mathcal{S})$
holds if and only if
\[
\lb(\Delta^+ \mathcal{S}) -   \ub(\Delta^-\mathcal{S}) -h_j \leq \sigma(\mathcal{S}) - h_j,
\]
which is equivalent to $ \kappa(\mathcal{S}\cup \{\Out_j\})  \leq \sigma(\mathcal{S}\cup \{\Out_j\})$.

An analogous argument shows that if $\mathcal{S}$ contains $\Out_j$ but does not contains $v_j$ for some $j\in\{r+1,r+2,\ldots, n\}$, then the validity of $\kappa(\mathcal{S})  \leq \sigma(\mathcal{S})$ is equivalent to
\begin{align*}
\lb(\Delta^+ \mathcal{S}) - \ub(\Delta^- \mathcal{S}) +h_j &\leq \sigma(\mathcal{S}) + h_j\\
\Leftrightarrow \kappa(\mathcal{S}\setminus \{\Out_j\}) & \leq \sigma(\mathcal{S}\setminus \{\Out_j\}).
\end{align*}
Hence it is sufficient to consider subset $\mathcal{S}$ which either contains both $v_j$ and $\Out_j$, or  none of them.

(b) For each $i=1,2,\ldots, r$, we distinguish eight cases as shown in the following table.
\smallskip

\centerline{
\begin{tabular}{|c||c|c|c|} \hline
Case & $\In_i\in \mathcal{S}$? & $\Mid_i\in\mathcal{S}$? & $\Out_i\in\mathcal{S}$? \\ \hline
1 & no & no& no \\ \hline
2 & no & no& yes \\ \hline
3 & no & yes& no \\ \hline
4 & no & yes& yes \\ \hline
5 & yes & no& no \\ \hline
6 & yes & no& yes \\ \hline
7 & yes & yes& no \\ \hline
8 & yes & yes& yes \\ \hline
\end{tabular}
}
\smallskip

We want to show that case 2 to case 7 are dominated by case 1 and 8, so that we only need to consider case 1 and~8.

Suppose that $\mathcal{S}$ contains $\Mid_i$. Since the link from $\In_i$ to $\Mid_i$ has infinite upper bound, the left-hand side of \eqref{eq:AAA} is equal to $-\infty$ if $\In_i$ is not included in $\mathcal{S}$. Then the inequality in \eqref{eq:AAA} holds trivially. We can assume without loss of generality that $\In_i\in\mathcal{S}$ if $\Mid_i\in\mathcal{S}$. This eliminates case 3 and case 4 in the above table.

Suppose that $\mathcal{S}$ contains $\Mid_i$ and $\In_i$ but not  $\Out_i$ (case 7), the inequality
 $\kappa(\mathcal{S})  \leq \sigma(\mathcal{S})$ is implied by
\[
 \kappa(\mathcal{S}\cup \{\Out_i\})  \leq \sigma(\mathcal{S}\cup \{\Out_i\}).
\]
Indeed, if we assume that the above inequality holds, then
\begin{align*}
\kappa(\mathcal{S})  = \kappa(\mathcal{S}\cup \{\Out_i\})
&\leq \sigma(\mathcal{S}\cup \{\Out_i\}) \\
&= \sigma(\mathcal{S})-h_i \leq \sigma(\mathcal{S}).
\end{align*}
Thus case 7 is implied by case 8.

Consider case 2, where $\mathcal{S}$ contains $\Out_i$ but does not contains $\In_i$ and $\Mid_i$. In this case, the inequality $\kappa(\mathcal{S})  \leq \sigma(\mathcal{S})$ is implied by the following two inequalities
\begin{align}
\kappa(\mathcal{S}\setminus\{\Out_i\}) & \leq \sigma(\mathcal{S} \setminus\{\Out_i\}), \label{eq:first_ineq} \\
\kappa(\{\Out_i\}) &\leq \sigma(\{\Out_i\}). \label{eq:second_ineq}
\end{align}
The  inequality in \eqref{eq:second_ineq} is simply equivalent to $-\alpha \leq -h_i$, which holds by the assumption on $\mathbf{h}$.
If we add  \eqref{eq:first_ineq}  and \eqref{eq:second_ineq}, we will get  $\kappa(\mathcal{S})  \leq \sigma(\mathcal{S})$. Thus case 2 can be eliminated.

Consider case 5. Suppose that  $\mathcal{S}$ contains $\In_i$ but  does not contain $\Mid_i$ and $\Out_i$. In this case the inequality $\kappa(\mathcal{S})  \leq \sigma(\mathcal{S})$
is implied by
\begin{align*}
\kappa(\mathcal{S}\setminus\{\In_i\})  \leq \sigma(\mathcal{S}\setminus\{\In_i\}).
\end{align*}
If we assume that the above inequality holds, then
\[
 \kappa(\mathcal{S})=
\kappa(\mathcal{S}\setminus\{\In_i\})
\leq \sigma(\mathcal{S}\setminus\{\In_i\})
= \sigma(\mathcal{S}).
\]

Finally, case 6 can be taken care of by combining the argument as in case 2 and case 5. This completes the proof of the claim.

Now, we prove that the inequality in~\eqref{eq:AAA} is valid for a subset $\mathcal{S}$ in the form of~\eqref{eq:S}. We let the cardinality of $\mathcal{A}$ and $\mathcal{B}$ be $a$ and $b$, respectively. Obviously we have $a\leq r$ and $b\leq n-r$.

In the following we will use $(x)^+$ as a short-hand notation for $\max(0,x)$.

Because $\lb(\Delta^+{\mathcal{S}})=0$, we have
\begin{align*}
\kappa(\mathcal{S})
&= -\ub(\Delta^-{\mathcal{S}}) \\
&\leq -a((d-b)^+\beta_1 +(r-a)\beta_2) \\
&= -a(2(d-b)^+ +(r-a)).
\end{align*}
It suffices to show that
\[
 -a(2(d-b)^+ +(r-a)) \leq  \sigma(\mathcal{S}).
\]
Since $ \sigma(\mathcal{S}) = \theta_b -\mathbf{h}(\mathcal{A})- \mathbf{h}(\mathcal{B})$ and
$\mathbf{h}(\mathcal{A})+ \mathbf{h}(\mathcal{B})\leq\theta_{a+b}$ by hypothesis, it is sufficient to prove
\[
-a(2(d-b)^+ +r-a) \leq \theta_{b} -\theta_{a+b}
\]
or equivalently
\begin{equation}
 \theta_{a+b} -\theta_{b} \leq  a(2(d-b)^+ +r-a) .  \label{eq:claimA}
\end{equation}

We prove the asserted inequality in~\eqref{eq:claimA} by distinguishing three cases.

{\em Case A, $a+b\leq k$:}
We first note that for $j\leq k$, we have
\begin{align*}
 \theta_j - \theta_{j-1}
 &= \begin{cases}
  \alpha & \text{if }0 < j \leq z \\
  \alpha - 2(j-z-1) & \text{if } z < j \leq k
 \end{cases}\\
\intertext{or equivalently}
 \theta_j - \theta_{j-1}  &= \alpha - 2(j-z-1)^+ \qquad \text{for }0<j\leq k.
\end{align*}
Recall that $\alpha = 2(d-z)+r-1$.
For $0< j \leq k$, we have the following upper bound
\begin{align*}
  \theta_j - \theta_{j-1} &= 2(d-z)+r-1-2(j-z-1)^+ \\
  & \leq 2(d-z)+r-1-2(j-z-1) \\
  &= 2(d-j)+r+1.
\end{align*}

Summing the above inequality for $j$ from $b+1$ to $a+b$, we obtain
\begin{align*}
 \theta_{a+b} -\theta_{b} = \sum_{j=b+1}^{a+b} \theta_j - \theta_{j-1}
&\leq \sum_{j=b+1}^{a+b} (2(d-j)+r+1) \\
 &= a(2(d-b)+r-a) \\
 &= a(2(d-b)^+ +r-a).
\end{align*}
We have use the assumptions that $d\geq k$ and $k\geq b$ for the last equality.

{\em Case B, $b \leq k < a+b$:} Since
$\theta_k = \theta_{k+1} = \cdots = \theta_{a+b}$
in this case, we have
\begin{align*}
\theta_{a+b} - \theta_{b} &= \theta_{k} - \theta_b
\leq (k-b)(2(d-b) +r-(k-b)).
\end{align*}
The inequality follows from the previous case.
We observe that the quadratic function
$$f(x)= x(2(d-b)+r-x)$$ is a concave with zeroes $x=0$ and $x=2(d-b)+r$. Thus we have $f(x) \geq f(k-b)$ for all $x$ between $k-b$ and $2d-b+r-k$. We check that $a\leq k-b$, because $k\leq a+b$ by the hypothesis in case B, and $a \leq 2d-b+r-k$ because $a\leq r$ and $d\geq k \geq b$. Therefore
\begin{align*}
\theta_{a+b} - \theta_{b} & \leq f(a) = a(2(d-b) +r-a)  \\
& = a(2(d-b)^+ +r-a).
\end{align*}

{\em Case C, $k< b$:} \eqref{eq:claimA} holds because $\theta_{a+b}-\theta_b = 0$ on the left-hand side, while the right-hand side is non-negative.

This completes the verification that condition~\eqref{eq:AAA} in Lemma~\ref{lemma:MBCR} holds.

\section{Proof of Theorem~\ref{thm:MSCR}}
\label{app:MSCR}

We give a sketch of proof of Theorem~\ref{thm:MSCR}, which is  along the same line as in the proof of Theorem~\ref{thm:MBCR}.

We draw the same auxiliary graph as in Fig.~\ref{fig:auxiliary}, except that $\beta_1$ is equal to 1, $\alpha$ is equal to $d+r(\ell+1)-k$, and $\mathbf{h}$ is a vector majorized by the vector $\mathbf{q}_\ell$ in~\eqref{eq:AP2}.
We define a submodular function
$$\rho(\mathcal{S}) := g(\mathcal{S}\cap \mathcal{O}_{s-1}) - \mathbf{h}(\mathcal{S}\cap \mathcal{O}_s)$$  on the vertex set of the auxiliary graph.

As in Lemma~\ref{lemma:MBCR}, we want to show that the inequality $\lb(\Delta^+ \mathcal{S}) - \ub(\Delta^- \mathcal{S}) \leq \rho(\mathcal{S})$ holds for all subsets $\mathcal{S}$ which is in the form of~\eqref{eq:S}.

Analogous to~\eqref{eq:claimA}, we need to show
\begin{equation}
\varphi_{a+b} - \varphi_{b} \leq a( (d-b)^+ +r-a)
\label{eq:varphi}
\end{equation}
for $0\leq a\leq r$ and $0\leq b\leq d$.

We distinguish three cases:  $a+b\leq k$, $b\leq k< a+b$ and $k < b$. We only consider the first case $a+b\leq k$. The proof for the second and third case is omitted.

Note that the difference $\varphi_{a+b} - \varphi_b$, i.e., the $x$-th component in~\eqref{eq:varphi} can be written as \begin{equation}
 \alpha - \Big\lceil \frac{x-k+\ell r}{r} \Big\rceil r.
 \label{eq:phi_difference}
\end{equation}
We need to take the sum of \eqref{eq:phi_difference} for $b < x \leq b+a$. Note that the value of~\eqref{eq:phi_difference} is constant for $r$ consecutive values of~$x$. Since $a$ is not larger than $r$, $\lceil (x-k+\ell r)/r \rceil$ assumes at most two values for $b < x \leq  b+a$. We further divide into two subcases.

{\em First subcase:} $\lceil (x-k+\ell r)/r \rceil$ is constant for $b < x \leq b+a$. For $x$ in this range, we have
$$\lceil (x-k+\ell r)/r \rceil = \lceil (a+b-k+\ell r)/r \rceil.$$
Hence,
\begin{align*}
&\phantom{=} a((d-b)^+ +r-a)-(\varphi_{a+b}-\varphi_b ) \\
&=a((d-b)^+ +r-a) - \sum_{x=b+1}^{a+b} \Big(\alpha - \Big\lceil \frac{a+b-k+\ell r}{r} \Big\rceil  r\Big)\\
&\geq a \Big( d-b+r-a - \alpha +  (a+b-k+\ell r) \Big) \\
& = a(d+(\ell+1)r-k - \alpha) = 0.
\end{align*}

{\em Second subcase:} $\lceil (x-k+\ell r)/r \rceil$ is not  constant for $$b < x \leq b+a.$$
Suppose that
$a+b > k-\ell r + \xi r$
 and
 $b \leq k-\ell r +\xi r$
for some integer $\xi$.
We have
\[\Big\lceil \frac{x-k+\ell r}{r}\Big\rceil  =
\begin{cases}
\xi  & \text{ for } b < x\leq k-\ell r + \xi r \\
\xi+1 & \text{ for } k-\ell r + \xi r < x\leq a+b.
\end{cases}
\]
For the ease of presentation, we use $\delta$ to stand for $a+b-(k-\ell r + \xi r)$, and let $Y$ be $d-b+r-a$. The value of
$\delta$ is positive.
With these notations, we get
\begin{align*}
 & \phantom{=} a((d-b)^+ +r-a)- (\varphi_{a+b} - \varphi_b) \\
 &\geq  aY- (\varphi_{a+b} - \varphi_b) \\
&= aY - (a-\delta)(\alpha - \xi r)  - \delta (\alpha - (\xi+1) r)  \\
&= (a-\delta)(Y-\alpha+\xi r) + \delta (Y-\alpha+(\xi+1)r).
\end{align*}
Since
\[
Y - \alpha +\xi r = d-b+r-a-d-r\ell-r+k+\xi r = - \delta,
\]
we get
\begin{align*}
 & \phantom{=} a((d-b)^+ +r-a)- (\varphi_{a+b} - \varphi_b) \\
& \geq -(a-\delta)\delta + \delta(r-\delta) \\
& = \delta(r-a) \geq 0.
\end{align*}
This proves \eqref{eq:varphi} for $a+b\leq k$.

The proof proceeds by applying Frank's theorem repeatedly, thereby iteratively constructing a flow on the modified information flow graph.




%
%
\end{document}